\documentclass[11pt]{article}

\usepackage{amsthm,amsmath,amssymb}
\usepackage{mathabx}
\usepackage{fullpage}
\usepackage{hyperref}
\usepackage{cleveref}
\usepackage{thmtools}
\usepackage{thm-restate}
\usepackage{authblk}
\usepackage{enumitem}
\usepackage[usenames,dvipsnames,svgnames,table]{xcolor}
\usepackage{graphicx}
\usepackage{caption}
\usepackage{subcaption}
\usepackage{setspace}
\usepackage[ruled]{algorithm2e}
\usepackage[noend]{algorithmic}

\SetAlFnt{\small}
\SetAlCapFnt{\small}
\SetAlCapNameFnt{\small}
\SetAlCapHSkip{0pt}
\IncMargin{-\parindent}

\newtheorem{problem}{Problem}
\newsavebox{\mybox}

\newtheorem{issue}{Obstacle}

\newtheorem{theorem}{Theorem}[section]
\newtheorem{lemma}[theorem]{Lemma}
\newtheorem{proposition}[theorem]{Proposition}

\newtheorem{myclaim}{Claim}[theorem]
\newtheorem{subclaim}{Subclaim}[myclaim]

\newenvironment{proofof}[1]{\medskip\noindent\emph{Proof of #1. }\ignorespaces}{\hfill\qed\medskip\par\noindent\ignorespacesafterend}
\newcommand{\qedclaim}{\hfill $\diamond$ \medskip}
\newenvironment{proofclaim}{\noindent\ignorespaces{\em Proof.}}{\hfill\qedclaim\par\noindent} 
\newcommand{\qedsubclaim}{\hfill $\circ$ \medskip}
\newenvironment{proofsubclaim}{\noindent{\em Proof.}}{\qedsubclaim}

\theoremstyle{definition}
\newtheorem{definition}{Definition}[section]
\newtheorem{remark}{Remark}[section]
\newtheorem{operation}{Operation}

\newcommand{\labeledsubtree}[3][1]{%
\def\treename{%
\if#11\ensuremath{T}%
\else\if#12\ensuremath{T'}%
\else\if#13\ensuremath{T''}%
\else\ensuremath{T^{(#1)}}%
\fi\fi\fi}%
\def\labeledtreename{\ensuremath{\treename_{#2}}}%
\ifx&#3&\labeledtreename%
\else\ensuremath{\labeledtreename \langle #3 \rangle}%
\fi%
}

\newcommand{\unlabeledsubtree}[2][1]{\labeledsubtree[#1]{}{#2}}

\newcommand{\tree}[1][1]{\unlabeledsubtree[#1]{}}

\newcommand{\partialsolution}[1]{\labeledsubtree[1]{#1}{}}

\newcommand{\centre}[1]{\ensuremath{{\cal C}\left(#1\right)}}

\newcommand{\CI}[1]{\ensuremath{{\cal X}\left(#1\right)}}
\newcommand{\MAXK}[1]{\ensuremath{{\cal K}\left(#1\right)}}
\newcommand{\SEP}[1]{\ensuremath{{\cal S}\left(#1\right)}}
\newcommand{\WSEP}[1]{\ensuremath{{\cal W}\left(#1\right)}}

\def\ci{\ensuremath{X}}
\def\maxk{\ensuremath{K}}
\def\sep{\ensuremath{S}}

\newcommand{\cliquetree}[1]{\labeledsubtree{#1}{}}
\newcommand{\edgeset}[2]{\ensuremath{E_{#2}\left(\cliquetree{#1}\right)}}
\newcommand{\superedgeset}[2]{\ensuremath{\bigcup_{#2', #2 \subseteq #2'}\edgeset{#1}{#2'}}}
\newcommand{\superstrictedgeset}[2]{\ensuremath{\bigcup_{#2', #2 \subset #2'}\edgeset{#1}{#2'}}}

\newcommand{\border}[1]{\ensuremath{\Omega\left(#1\right)}}

\newcommand{\induced}[2]{\ensuremath{#1\left[#2\right]}}

\newcommand{\real}[1][1]{%
\if#11\ensuremath{v}%
\else\if#12\ensuremath{u}%
\else\if#13\ensuremath{w}%
\else\if#14\ensuremath{x}%
\else\if#15\ensuremath{y}%
\else\ensuremath{z}%
\fi\fi\fi\fi\fi}

\newcommand{\steiner}[1][1]{%
\if#11\ensuremath{\alpha}%
\else\ensuremath{\beta}%
\fi}

\newcommand{\KLP}[1][k]{\ensuremath{#1}-{\sc Leaf Power}}
\newcommand{\KSP}[1][k]{\ensuremath{#1}-{\sc Steiner Power}}
\newcommand{\DCR}{\textsc{Distance-Constrained Root}}

\begin{document}

\title{Polynomial-time Recognition of $4$-Steiner Powers}
\author[1,2,3]{Guillaume Ducoffe}
\affil[1]{\small National Institute for Research and Development in Informatics, Romania}
\affil[2]{\small The Research Institute of the University of Bucharest ICUB, Romania}
\affil[3]{\small University of Bucharest, Faculty of Mathematics and Computer Science, Romania}
\date{}

\maketitle

\begin{abstract}
The $k^{\text{th}}$-power of a given graph $G=(V,E)$ is obtained from $G$ by adding an edge between every two distinct vertices at a distance $\leq k$ in $G$.
We call $G$ a {\em $k$-Steiner power} if it is an induced subgraph of the $k^{\text{th}}$-power of some tree.
Our main contribution is a polynomial-time recognition algorithm of $4$-Steiner powers, thereby extending the decade-year-old results of (Lin, Kearney and Jiang, {\it ISAAC'00}) for $k=1,2$ and (Chang and Ko, {\it WG'07}) for $k=3$.
Our motivation for studying the $k$-Steiner powers comes from the following open problem in the graph literature.
A graph $G$ is termed {\em $k$-leaf power} if there is some tree \tree~such that: all vertices in $V(G)$ are leaf-nodes of \tree, and $G$ is an induced subgraph of the $k^{\text{th}}$-power of \tree.
As a byproduct of our main result, we give the first known polynomial-time recognition algorithm for $6$-leaf powers.
%Here also, this is the first advance on this open problem in a decade.
%
Our work combines several new algorithmic ideas that help us overcome the previous limitations on the usual dynamic programming approach for these problems.
We expect several components of this new framework to be further used in the recognition of $k$-leaf powers, $k$-Steiner powers and related graph classes. 

\end{abstract}

\renewcommand{\thetheorem}{\Alph{theorem}}
\section{Introduction}\label{sec:intro}

A basic problem in computational biology is, given some set of species and a dissimilarity measure in order to compare them, find a {\em phylogenetic tree} that explains their respective evolution.
Namely, such a rooted tree starts from a common ancestor and branches every time there is a separation between at least two of the species we consider.
In the end, the leaves of the phylogenetic tree should exactly represent our given set of species.
%This problem was brought to Graph theory under several disguises but, unfortunately, several of these formulations are NP-hard to solve~\cite{BFW92,Ste92}.
%
Standard formulations of this problem are NP-hard to solve~\cite{BFW92,Ste92}.
We here study a related problem whose complexity status remains open.
Specifically, a common assumption in the literature is that our dissimilarity measure can only tell us whether the separation between two given species has occurred quite recently.
Let $G=(V,E)$ be a graph whose vertices are the species we consider and such that an edge represents two species with a quite ``close'' common ancestor according to the dissimilarity measure.
Formally, given some fixed $k \geq 1$, we ask whether there exists some tree \tree~whose leaf-nodes are exactly $V$ and such that there is an edge $uv$ in $E$ if and only if the two corresponding nodes in \tree~are at a distance $\leq k$.
This is called the \KLP~problem~\cite{NRT02}.

%\medskip
%\paragraph{Related work.}
The structural properties of {\em $k$-leaf powers} ({\it i.e.}, graphs for which a tree as above exists) have been intensively studied~\cite{BPP10,BrH08,BHMW10,BaBV06,BLS08,BLR09,BrW09,CFM11,DGFN06,DGHN08,DGN05,KLY06,KKLY10,Laf17,NeR16,Rau06}.
From the algorithmic point of view, $k$-leaf powers are a subclass of bounded {\em clique-width} graphs, and many NP-hard problems can be solved efficiently for these graphs~\cite{FMRS+08,GuW07}.
However, the computational complexity of recognizing $k$-leaf powers is a longstanding open problem.
Very recently, parameterized (FPT) algorithms were proposed for every fixed $k$ on the graphs with degeneracy at most $d$, where the parameter is $k+d$~\cite{EpH18}.
%parameterized algorithms were proposed for every fixed $k$ on the graphs with bounded degeneracy~\cite{EpH18}.
Without this additional restriction on the degeneracy of the graphs, polynomial-time recognition algorithms are known only for $k \leq 5$~\cite{BaBV06,BLS08,ChK07}.
It is noteworthy that every algorithmic improvement for this problem, while bringing several new important insights on the structure of leaf powers, has been incredibly hard to generalize to larger values of $k$.
We contribute to this frustrating chain of improvements by providing the first known polynomial-time recognition algorithm for $6$-leaf powers.

\begin{restatable}{theorem}{leafPower}
\label{thm:main-leaf-power}
There is a polynomial-time algorithm that given a graph $G=(V,E)$, correctly decides whether $G$ is a $6$-leaf power (and if so, outputs a corresponding tree \tree).
\end{restatable}

Proving this above Theorem~\ref{thm:main-leaf-power}, while it may look like a modest improvement in our understanding of the \KLP~problem, was technically challenging.
We will further sketch in Section~\ref{sec:dyn-prog-gal} why the previous approach for \KLP, $k \leq 5$, was already showing its limitations with \KLP[6].
Apart from pushing further the tractable cases for an important open problem in the graph literature, we believe that  one of the main merits of our paper is to bring several new ideas in order to tackle with these aforementioned limitations.
As such, we expect further uses of our new framework in the study of $k$-leaf powers and their relatives. 

Several variations of $k$-leaf powers were introduced in the literature~\cite{BLR10,BrW10,ChK07,HaT10,JKL00}.
In this work, we consider {\em $k$-Steiner} powers: a natural relaxation of $k$-leaf powers where the vertices in the graph may also be internal nodes in the tree \tree.
Interestingly, for every $k \geq 3$, the notions of $k$-leaf powers and $(k-2)$-Steiner powers are equivalent for a {\em twin-free} graph.
The latter implies a linear-time reduction from \KLP~to \KSP[(k-2)]~\cite{BLS08}.
Furthermore, there exist polynomial-time recognition algorithms for $k$-Steiner powers, for every $k \leq 3$~\cite{ChK07,JKL00}.
As our main contribution in the paper we obtain the first improvement on the recognition of $k$-Steiner powers in a decade.
Specifically we prove that there is a polynomial-time recognition algorithm for the $4$-Steiner powers.

\begin{restatable}{theorem}{SteinerPower}
\label{thm:main-steiner-power}
There is a polynomial-time algorithm that given a graph $G=(V,E)$, correctly decides whether $G$ is a $4$-Steiner power (and if so, outputs a corresponding tree \tree).
\end{restatable}

Note that Theorem~\ref{thm:main-leaf-power} follows from the combination of Theorem~\ref{thm:main-steiner-power} with the aforementioned reduction from \KLP~to \KSP[(k-2)]~\cite{BLS08}.
{\em Hence, the remaining of this paper is devoted to a polynomial-time solution for \KSP[4].}
We think that our general approach (presented next) could be generalized to larger values of $k$. 
However, this would first require to strenghten the structure theorems we use in this paper and probably to find less intricate proofs for some of our intermediate statements.

\medskip
Before we can introduce our technical contributions in this paper, we must explain in Sec.~\ref{sec:dyn-prog-gal} the dynamic programming approach behind the $k$-Steiner power recognition algorithms, for $k \leq 3$, and why it is so hard to apply this approach to \KSP[4].
Doing so, we identify Obstacles~\ref{issue-1} and~\ref{issue-2} as the two main issues to fix in order to prove Theorem~\ref{thm:main-steiner-power}.
Then in Sec.~\ref{sec:org}, we briefly summarize our proposed solutions for Obstacles~\ref{issue-1} and~\ref{issue-2} while presenting the organization of the technical sections of this paper.

\subsection{The difficulties of a dynamic programming approach}\label{sec:dyn-prog-gal}

We refer to Sec.~\ref{sec:prelim} for any undefined graph-theoretic terminology in this introduction.
In what follows we give a high-level overview of our approach, that we compare to prior work on \KSP~and \KLP.
As our starting point we restrict our study to {\em chordal graphs} and {\em strongly chordal graphs}, that are two well-known classes in algorithmic graph theory of which $k$-Steiner powers form a particular subclass~\cite{ABNT16}.
Doing so, we can use various properties of these classes of graphs such as: the existence of a tree-like representation of chordal graphs, that is called a {\em clique-tree}~\cite{BlP93} and is commonly used in the design of dynamic programming algorithms on this class of graphs; and an auxiliary data structure which is called ``{\em clique arrangement}'' and is polynomial-time computable on strongly chordal graphs~\cite{NeR16}.
Roughly, this clique arrangement encodes all possible intersections of a subset of maximal cliques in a graph.
It is worth noticing that clique arrangements were introduced in the same paper as leaf powers, under the different name of ``clique graph''~\cite{NRT02}.

Our first result is that every maximal clique, minimal separator and, more generally, any intersection $\ci$ of maximal cliques in a $k$-Steiner power must be contained in an ${\cal O}(k \cdot |\ci|)$-node subtree with very specific properties -- detailed next in Sec.~\ref{sec:org} -- of the tree \tree~we aim at computing.
This result extends to any $k$ the structural results that were presented in~\cite{ChK07} for $k \leq 3$, and it is a prerequisite for the design of a dynamic programming algorithm on a clique-tree.
Unfortunately, as the value of $k$ increases it becomes more and more difficult to derive from such structural results a polynomial-time recognition algorithm for $k$-Steiner powers.
Our proposed solutions for $k=4$ are quite different from those used in the previous works on $k$-Steiner powers~\cite{ChK07,JKL00} which results in an embarrassingly long and intricate proof.  

\medskip
To give a flavour of the difficulties we met, we consider the following common situation in a dynamic programming algorithm on chordal graphs.
Given a graph $G=(V,E)$, let \sep~be a minimal separator of $G$ and $C$ be a full component of $G \setminus \sep$ ({\it i.e.}, such that every vertex in \sep~has a neighbour in $C$).
If $G$ is a $k$-Steiner power then, so must be the induced subgraph $G[C \cup \sep]$.
We want to store the partial solutions obtained for $G[C \cup \sep]$, as at least one of them should be extendible to all of $G$ (otherwise, $G$ is not a $k$-Steiner power).
%We sketch in what follows the two main obstacles that we met in the design of a ``naive'' dynamic programming algorithm for solving Theorem~\ref{thm:main-steiner-power}: 
However, for doing so efficiently we must overcome the following two obstacles:

\begin{itemize}[label=-]
\item There may be exponentially many partial solutions already when $G[C \cup \sep]$ is a complete subgraph and $k \geq 3$.
Therefore, we cannot afford to store all possible solutions explicitly. 
Nevertheless it seems at the minimum we need to keep the part of these solutions which contains \sep: in order to be able to check later whether the solutions found for $G[C \cup \sep]$ can be extended to all of $G$.
We will prove in this paper that such a part \unlabeledsubtree{\sep} of the partial solutions is a subtree of diameter at most $k-1$, and so, there may be exponentially many possibilities to store whenever $k \geq 4$.

\begin{issue}\label{issue-1}
Decrease the number of possibilities to store for \unlabeledsubtree{\sep} to a polynomial.
\end{issue}

\item Furthermore, since there can be no edge between $C$ and $V \setminus \left( C \cup \sep\right)$, the tree $T$ that we want to compute for $G$ must satisfy that all vertices in $C$ stay at a distance $\geq k+1$ from all vertices in $V \setminus (\sep \cup C)$.
In order to ensure this will be the case, we wish to store a ``distance profile'' $(dist_{\tree}(r,C))_{r \in V(\unlabeledsubtree{S})}$ in the encoding of all the partial solutions found for $G[C \cup \sep]$.
Storing this information would result in a combinatorial explosion of the number of possible encodings, even if there are only a few possibilities for the subtree \unlabeledsubtree{\sep}.
Chang and Ko proposed two nice ``heuristic rules'' in order to overcome this distance issue for $k=3$~\cite{ChK07}.
Unfortunately, these rules do not easily generalize to larger values of $k$.

\begin{issue}\label{issue-2}
For any fixed \unlabeledsubtree{S}, decrease the number of possibilities to store for the ``distance profiles'' to a polynomial.
\end{issue}

\end{itemize}

In order to derive a polynomial-time algorithm for the case $k=4$, we further restrict the structural properties of the ``useful'' partial solutions we need to store.
This is done by carefully analysing the relationships between the structure of these solutions and the intersections between maximal cliques in the graph.
Perhaps surprisingly, we need to combine these stronger properties on the partial solutions with several other tricks so as to bound the number of partial solutions that we need to store by a polynomial ({\it e.g.}, we also impose local properties on the clique-tree we use, and we introduce a new greedy selection procedure based on graph matchings).
%Furthermore if $k=4$ then, we derive enough constraints from the previous observation in order to compute, for every maximal clique of the graph, a polynomial-size candidate set of subtrees.
%
%Aggregating these subtrees into a solution $T$ is not straightforward due to the need to satisfy additional {\em distance} properties for the vertices in different maximal cliques.
%Remarkably, we show that the difficulty of this aggregation phase only depends on the minimal separators whose removal disconnects the graphs in at least three components.
%In order to circumvent the issue such separators may pose, we introduce a special routine at the end of the paper that is based on separator trees. 

\subsection{Organization of the paper}\label{sec:org}
We give the required graph-theoretic terminology for this paper in Section~\ref{sec:prelim}.
We emphasize on Section~\ref{sec:algo-highlight}: where we also provide a more detailed presentation of our algorithm, as a guideline for all the other sections.

\medskip
Given a $k$-Steiner power $G$, let us call {\em $k$-Steiner root} a corresponding tree \tree.
In Sections~\ref{sec:representation} and~\ref{sec:restricted-root} we present new results on the structure of $k$-Steiner roots that we use in the analysis of our algorithm.
Specifically, we show in Section~\ref{sec:representation} that any intersection of maximal cliques in a graph $G$ must induce a particular subtree in any of its $k$-Steiner roots \tree~where no other vertex of $G$ can be present.
Furthermore, the inclusion relationships between these ``clique-intersections'' in $G$ are somewhat reflected by the diameter of their corresponding subtrees in \tree.
An intriguing consequence of our results is that, in any $k$-Steiner power, there can be no chain of more than $k$ minimal separators ordered by inclusion.
This slightly generalizes a similar result obtained in~\cite{NRT02} for $k$-leaf powers.

Then, we partly complete this above picture in Section~\ref{sec:restricted-root} for the case $k=4$.
For every clique-intersection \ci~in a chordal graph $G$, we classify the vertices in \ci~into two main categories: ``free'' and ''constrained'', that depend on the other clique-intersections these vertices are contained into.
Our study shows that ``free'' vertices cause a combinatorial explosion of the number of partial solutions we should store in a naive dynamic programming algorithm. %for the \KSP[4]~problem.
However, on the positive side we prove that there always exists a ``well-structured'' $4$-Steiner root where such free vertices are leaves with very special properties.
This result will be instrumental in ruling out Obstacle~\ref{issue-1}.
%rules out one of the main difficulties we met in the design of our algorithm.  

\medskip
Sections~\ref{sec:clique-tree},~\ref{sec:ci-subtrees},~\ref{sec:encoding} and~\ref{sec:greedy} are devoted to the main steps of the algorithm.
We start presenting a constructive proof of a rooted clique-tree with quite constrained properties in Section~\ref{sec:clique-tree}.
Roughly we carefully control the ancestor/descendant relationships between the edges that are labelled by different minimal separators of the graph.
These technicalities are the cornerstone of our approach in Section~\ref{sec:representation} in order to bound the number of ``distance profiles'' which we need to account in our dynamic programming (cf. Obstacle~\ref{issue-2}).
%We pre-process an arbitrary rooted clique-tree in order to enforce some specific 
%after some preprocessing, we root our clique-tree in such a way that smaller minimal separators should label the edges closer to the root.
%Our construction ensures that several complications that could occur by using our approach with an arbitrary clique-tree will never occur. 
%Our technical construction is partly motivated by the results in Section~\ref{sec:representation}. 

Then in Section~\ref{sec:ci-subtrees}, we completely rule out Obstacle~\ref{issue-1} --- in fact, we solve a more general subproblem.
For that, let \cliquetree{G}~be the rooted clique-tree of Section~\ref{sec:clique-tree}. 
Recall that the maximal cliques and the minimal separators of $G$ can be mapped to the nodes and edges of \cliquetree{G}, respectively.
We precompute by dynamic programming, for every node and edge in \cliquetree{G}, a family of all the potential subtrees to which the corresponding clique-intersection of $G$ could be mapped in some well-structured $4$-Steiner root of $G$.
Of particular importance is Section~\ref{sec:minsep} where for any minimal separator $\sep$, we give a polynomial-time algorithm in order to generate all the candidate smallest subtrees into which $\sep$ can be contained in a $4$-Steiner root of $G$.
The result is then easily extended to the maximal cliques that appear as leaves in our clique-tree (Section~\ref{sec:leaf-case}). 
%In Section~\ref{sec:ci-subtrees} we continue using the results in Sections~\ref{sec:representation} and~\ref{sec:restricted-root} in order to 
Correctness of these two first parts follows from Sec.~\ref{sec:restricted-root}.
Finally, in Section~\ref{sec:super-selection} we give a more complicated representation of a family of candidate subtrees $\unlabeledsubtree{\maxk_i}$ for all the other maximal cliques $\maxk_i$.
This part is based on a careful analysis of clique-intersections in $\maxk_i$ and several additional tricks.  
%Furthermore unlike in Sec.~\ref{sec:minsep} and~\ref{sec:leaf-case}, ``free'' vertices are no more the only cause of combinatorial explosion of the number of possibilities to store.
Roughly, our representation in Sec.~\ref{sec:minsep} is composed of partially constructed subtrees and of ``problematic'' subsets that need to be inserted to these subtrees in order to complete the construction.
The exact way these insertions must be done is postponed until the very end of the algorithm (Sec.~\ref{sec:greedy}).

Section~\ref{sec:encoding} is devoted to the ``distance profiles'' of partial solutions and how to overcome Obstacle~\ref{issue-2}.
Specifically, instead of computing partial solutions at each node of the clique-tree and storing their encodings, we rather pre-compute a polynomial-size subset of {\em imposed} encodings for each node.
Then, the problem becomes to decide whether given such an imposed encoding, there exists a corresponding partial solution.
We formalize our approach by introducing an intermediate problem where the goal is to compute a $4$-Steiner root with additional constraints on its structure and the distances between some sets of nodes.

Finally, we detail in Section~\ref{sec:greedy} the resolution of our intermediate problem, thereby completing the presentation of our algorithm.
An all new contribution in this part is a greedy procedure, based on {\sc Maximum-Weight Matching}, in order to ensure some distances' constraints are satisfied by the solutions we generate during the algorithm.
Interestingly, this procedure is very close in spirit to the implementation of the {\tt alldifferent} constraint in constraint programming~\cite{Reg94}.

Due to the intricacy of our proofs we gave up optimizing the runtime of our algorithm.
We will only provide enough arguments in order to show it is polynomial.

\medskip
We end up this paper in Section~\ref{sec:ccl} with some ideas for future work.
\renewcommand{\thetheorem}{\arabic{section}.\arabic{theorem}}

\section{Preliminaries}\label{sec:prelim}

We refer to~\cite{BoM08} for any undefined graph terminology.
All graphs in this study are finite, simple (hence, with neither loops nor multiple edges), unweighted and connected -- unless stated otherwise.
Given a graph $G=(V,E)$, let $n := |V|$ and $m := |E|$.
The neighbourhood of a vertex $v \in V$ is defined as $N_G(v) := \{ u \in V \mid uv \in E \}$.
By extension, we define the neighbourhood of a set $S \subseteq V$ as $N_G(S) := \left( \bigcup_{v \in S} N_G(v) \right) \setminus S$.
The subgraph induced by any subset $U \subseteq V$ is denoted by $G[U]$.

For every $u,v \in V$, we denote by $dist_G(u,v)$ the minimum length (number of edges) of a $uv$-path.
The eccentricity of vertex $v$ is defined as $ecc_G(v) := \max_{u \in V} dist_G(u,v)$.
The radius and the diameter of $G$ are defined, respectively, as $rad(G) := \min_{v \in V} ecc_G(v)$ and $diam(G) := \max_{v \in V} ecc_G(v)$.
We denote by $\centre{G}$ the center of $G$, {\it a.k.a.} the vertices with minimum eccentricity.  

\subsection{Problems considered}

The {\em $k^{\text{th}}$-power} of $G$, denoted $G^k$ has same vertex-set $V$ as $G$ and edge-set $\{ uv \mid 0 < dist_G(u,v) \leq k \}$. 
We call $G$ a {\em $k$-Steiner power} if there is some tree \tree~such that $G$ is an induced subgraph of $\tree^k$.
Conversely, \tree~is called a {\em $k$-Steiner root} of $G$.
If in addition, $G$ has a $k$-Steiner root where all vertices in $V$ are leaves (degree-one nodes) then, we call $G$ a {\em $k$-leaf power}.

\begin{center}
	\fbox{
		\begin{minipage}{.95\linewidth}
			\begin{problem}[\KSP]\
				\label{prob:steiner-root} 
					\begin{description}
					\item[Input:] A graph $G=(V,E)$.
					\item[Output:] Is $G$ a $k$-Steiner power?
				\end{description}
			\end{problem}     
		\end{minipage}
	}
\end{center}

\begin{center}
	\fbox{
		\begin{minipage}{.95\linewidth}
			\begin{problem}[\KLP]\
				\label{prob:leaf-power} 
					\begin{description}
					\item[Input:] A graph $G=(V,E)$.
					\item[Question:] Is $G$ a $k$-leaf power?
				\end{description}
			\end{problem}     
		\end{minipage}
	}
\end{center}

\begin{theorem}[~\cite{BLS08}]\label{thm:reduction}
There is a linear-time reduction from \KLP~to \KSP[(k-2)]~for every $k \geq 3$.
\end{theorem}

%\paragraph{Additional terminology for Steiner root.}
If \tree~is any $k$-Steiner root of $G$ then, nodes in $V(G)$ are called {\em real}, whereas nodes in $V(\tree) \setminus V(G)$ are called {\em Steiner}.
We so define, for any $S \subseteq V(\tree)$ (for any subtree $\tree[2] \subseteq \tree$, resp.): 
$$Real(S) := S \cap V(G) \ \text{and} \ Steiner(S) := S \setminus V(G)$$ 
(we define $Real(\tree[2]) := Real(V(\tree[2]))$ and $Steiner(\tree[2]) = Steiner(V(\tree[2]))$, resp.).

Note that throughout all this paper we consider two (sub)trees being equivalent if they are equal up to an appropriate identification of their Steiner nodes, namely (see also Fig.~\ref{fig:steiner-equiv}):

\begin{definition}\label{def:steiner-equiv}
Given $G=(V,E)$, we call any two trees $\tree,\tree[2]$ {\em Steiner-equivalent}, denoted $\tree \equiv_G \tree[2]$, if and only if $Real(\tree) = Real(\tree[2]) = S$ and there exists an isomorphism $\iota : V(\tree) \to V(\tree[2])$ such that $\iota(v) = v$ for any $v \in S$.
\end{definition}

\begin{figure}
\centering
\includegraphics[width=.5\textwidth]{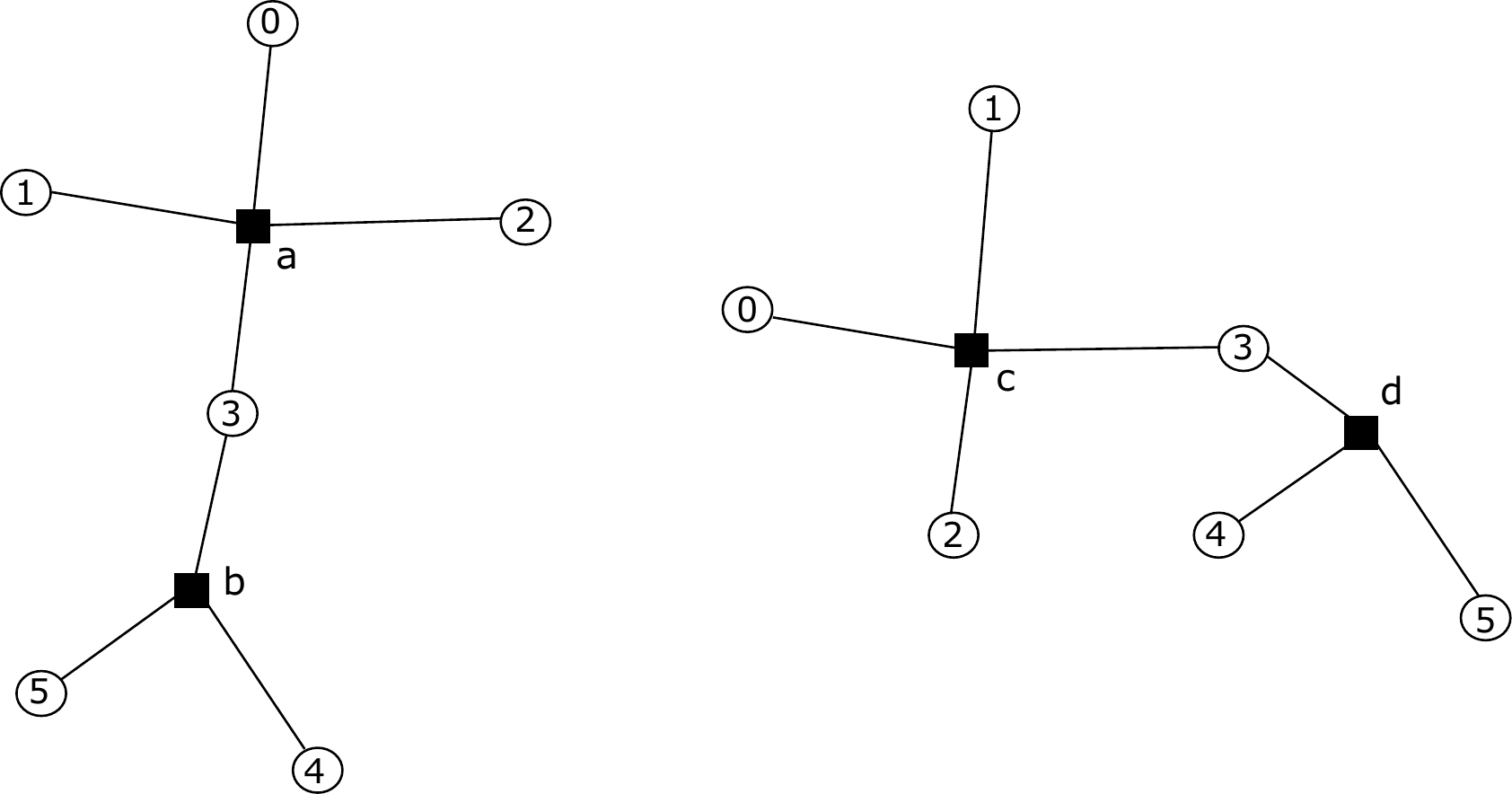}
\caption{Two Steiner-equivalent trees. Cycles and rectangles represent real and Steiner nodes, respectively.}
\label{fig:steiner-equiv}
\end{figure}

\medskip
Finally, given a node-subset $X \subseteq V(\tree)$, \unlabeledsubtree{X} is the smallest subtree of \tree~such that $X \subseteq V(\unlabeledsubtree{X})$.
Note that for a vertex-subset $X \subseteq V$, this is the smallest subtree of \tree~such that $X \subseteq Real(\unlabeledsubtree{X})$.
Furthermore we observe $\induced{\tree}{X} \subseteq \unlabeledsubtree{X}$, with equality if and only if $\induced{\tree}{X}$ is connected.

\subsection{Algorithmic tool-kit: (Strongly) Chordal graphs}

Given $G=(V,E)$, we call it a {\em chordal} graph if every induced cycle in $G$ is a triangle.
If in addition, for every cycle of even length in $G$, there exists a chord between two vertices at an odd distance ($>1$) apart from each other in the cycle then, $G$ is termed {\em strongly chordal}.
Chordal graphs and strongly chordal graphs can be recognized in ${\cal O}(m)$-time and ${\cal O}(m\log{n})$-time, respectively~\cite{PaT87,RTL76}.

\smallskip
The following property is well-known:

\begin{theorem}[~\cite{ABNT16}]\label{thm:partial-characterization}
For every $k \geq 1$, every $k$-Steiner power is a strongly chordal graph.
\end{theorem}

\paragraph{Minimal separators and Clique-tree.}
Our main algorithmic tool in this paper is a {\em clique-tree} of $G$, defined as a tree $\cliquetree{G}$ whose nodes are the maximal cliques of $G$ and such that for every $v \in V$, the maximal cliques containing $v$ induce a subtree of $\cliquetree{G}$.

\begin{theorem}[~\cite{BlP93}]\label{thm:clique-tree-construction}
A graph $G=(V,E)$ is chordal if and only if it has a clique-tree.
Moreover if $G$ is chordal then, we can construct a clique-tree for $G$ in ${\cal O}(m)$-time.
\end{theorem}

An $uv$-separator is a subset $\sep \subseteq V \setminus \{u,v\}$ such that $u$ and $v$ are disconnected in $G \setminus \sep$.
If in addition, no strict subset of $\sep$ is an $uv$-separator then, $\sep$ is a {\em minimal} $uv$-separator.
A {\em minimal separator} of $G$ is a minimal $uv$-separator for some $u,v \in V$.
It is known that any minimal separator in a chordal graph $G$ is the intersection of two distinct maximal cliques of $G$.
Specifically, the following stronger relationship holds between minimal separators and clique-trees:

\begin{theorem}[~\cite{BlP93}]\label{thm:clique-tree-pties}
Given $G=(V,E)$ chordal, any of its clique-trees $\cliquetree{G}$ satisfies the following properties:
\begin{itemize}
\item For every edge $\maxk_i\maxk_j \in E(\cliquetree{G})$, $\maxk_i \cap \maxk_j$ is a minimal separator;
\item Conversely, for every minimal separator $\sep$ of $G$, there exist two maximal cliques $\maxk_i,\maxk_j$ such that $\maxk_i\maxk_j \in E(\cliquetree{G})$ and $\maxk_i \cap \maxk_j = \sep$.
\end{itemize} 
\end{theorem}

Based on the above theorem, we can define $\edgeset{G}{\sep} := \{ \maxk_i\maxk_j \in E(\cliquetree{G}) \mid \maxk_i \cap \maxk_j = \sep \}$.
The cardinality $|\edgeset{G}{\sep}|$ of this subset does not depend on $\cliquetree{G}$~\cite{BlP93}.
We sometimes say that edges in $\edgeset{G}{\sep}$ are {\em labeled} by $\sep$.

\medskip
A {\em rooted} clique-tree of $G$ is obtained from any clique-tree $\cliquetree{G}$ by identifying an arbitrary maximal clique $\maxk_0$ as its root.
Let $(\maxk_q,\maxk_{q-1},\ldots,\maxk_1,\maxk_0)$ be a postordering of $\cliquetree{G}$ (obtained by depth-first search).
For any $i > 0$, we define $\maxk_{p(i)}$ as the father node of $\maxk_i$.
The common intersection of $\maxk_i$ with its father node is the minimal separator $\sep_i := \maxk_i \cap \maxk_{p(i)}$.
By convention, we set $\sep_0 := \emptyset$.
We refer to Fig.~\ref{fig:rooted-clique-tree} for an illustration.

We define $\cliquetree{G}^i$ as the subtree rooted at $\maxk_i$, and let $G_i$ be the subgraph induced by all the maximal cliques in $V(\cliquetree{G}^i)$.
In particular, we have $\cliquetree{G}^0 = \cliquetree{G}$ and $G_0 = G$.
We also define $V_i := V(G_i)$ and $W_i := V_i \setminus \sep_i$ as shorthands.
We will use these above notations for rooted clique-trees throughout the remaining of our paper.

\begin{figure}
\centering
\includegraphics[width=.75\textwidth]{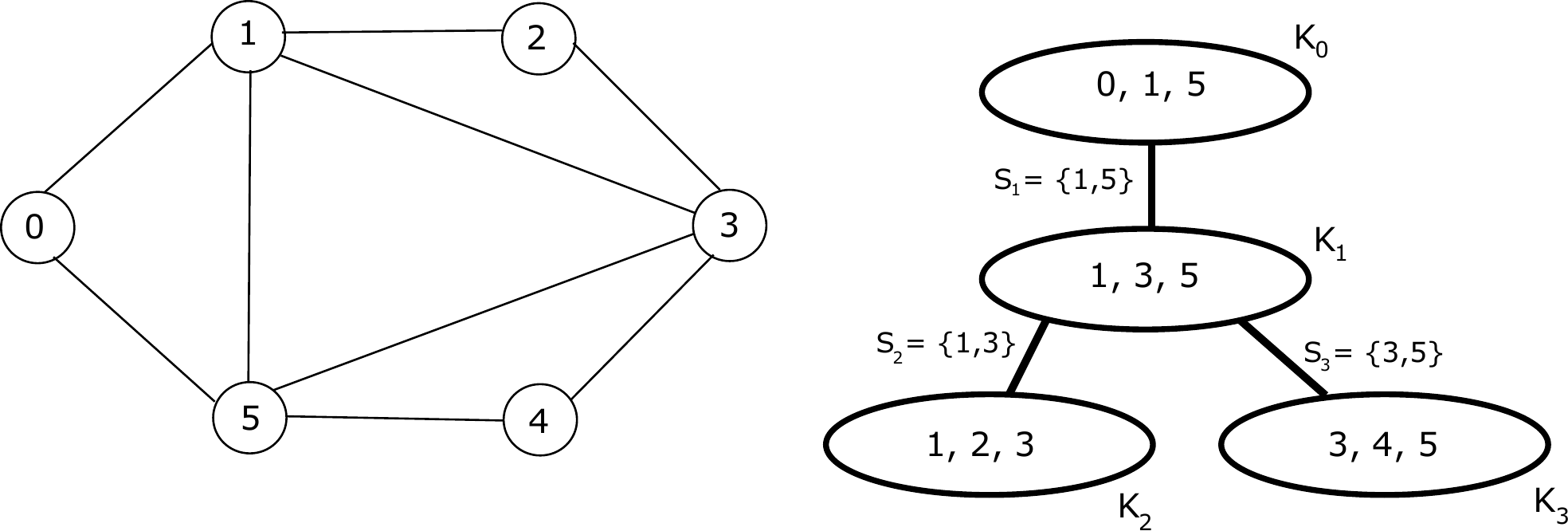}
\caption{A chordal graph $G$ (left) and a rooted clique-tree $\cliquetree{G}$ (right).}
\label{fig:rooted-clique-tree}
\end{figure}

\paragraph{Clique arrangement.}
We introduce a common generalization of both maximal cliques and minimal separators, that will play a key role in our analysis. Specifically, a {\em clique-intersection} in $G$ is the intersection of a subset of maximal cliques in $G$.
The family of all clique-intersections in $G$ is denoted by $\CI{G}$.
For strongly chordal graphs, it is known~\cite{NeR15} that every clique-intersection is the intersection of at most two maximal cliques.
In particular, a (nonempty) clique-intersection of a given strongly chordal $G$ is either: a maximal clique; or a minimal separator; or a {\em weak} minimal separator -- {\it i.e.}, whose removal strictly increases the distance between two vertices that remain in the graph (see~\cite{McK11}).
We denote by $\MAXK{G}, \SEP{G}$ and $\WSEP{G}$ the subfamilies of all the maximal cliques, minimal separators and weak minimal separators of $G$, respectively.

The {\em clique arrangement} of $G$ is the inclusion (directed) graph of the clique-intersections of $G$.
That is, there is a node for every clique-intersection, and there is an arc from $\ci$ to $\ci'$ if and only if we have $\ci \subseteq \ci'$.

\begin{theorem}[~\cite{NeR15}]\label{thm:clique-arrangement}
Given $G=(V,E)$ strongly chordal, the clique arrangement of $G$ can be constructed in ${\cal O}(m\log{n})$-time.
\end{theorem}

\subsection{Highlights of the algorithm}\label{sec:algo-highlight}

The remaining of the paper is devoted to the proof of Theorem~\ref{thm:main-steiner-power}.
By Theorem~\ref{thm:reduction}, this will also imply Theorem~\ref{thm:main-leaf-power}.
We start sketching our algorithm below in order to guide the readers throughout the next sections.
Its analysis is based on the structure theorems in Sections~\ref{sec:representation} and~\ref{sec:restricted-root}.
Perhaps surprisingly, we need several tricks in order to keep the running time of this algorithm polynomial.

\begin{enumerate}[label=\textbf{S.\theenumi},ref=S.\theenumi]
\addtocounter{enumi}{-1}
\item\label{step-0}({\it Initialization.})
Given $G=(V,E)$, we check whether $G$ is strongly chordal.
If this is not the case then, by Theorem~\ref{thm:partial-characterization}, $G$ cannot be a $k$-Steiner power for any $k \geq 1$, and we stop.
Otherwise by Theorem~\ref{thm:clique-arrangement} we can compute the clique-arrangement of $G$ in polynomial time.
Throughout all the remaining sections, we implicitly use the fact that we can access in polynomial time to the clique-arrangement of $G$.
We will also assume in what follows that $G$ is not a complete graph (otherwise, $G$ is trivially a $k$-Steiner power for any $k$, and so we also stop in this case).

\item\label{step-1}({\it Construction of the rooted clique-tree.})
We construct a clique-tree $\cliquetree{G}$ of $G$ that we root in some $\maxk_0 \in \MAXK{G}$.
This clique tree must satisfy very specific properties of which we postpone the precise statement until Section~\ref{sec:clique-tree}. 
In order to give the main intuition behind this construction, let us consider an arbitrary maximal clique $\maxk_i$ that is not the root ({\it i.e.}, $i > 0$).
Roughly, we would like to minimize the number of minimal separators in $G_i$ to which the vertices in $\sep_i := \maxk_i \cap \maxk_{p(i)}$ belong to.
Indeed, doing so will help in bounding the number of possible partial solutions that we will need to consider for processing the father node $\maxk_{p(i)}$ of $\maxk_i$.
More specifically, for every descendant $\maxk_j$ of $\maxk_i$ in $\cliquetree{G}$ we would like to impose $\sep_j \not\subseteq \sep_i$ and $\sep_i \not\subseteq \sep_j$.
However, both objectives are conflicting and so, we need to find a trade-off.
%let us consider two maximal cliques $\maxk_i$ and $\maxk_j$ such that $\max_i$ is an ancestor of $\max_j$ in 
%Roughly, 
%Roughly, we want to ensure that a minimal separator $S$ can occur as en edge $X_iX_{p(i)} \in E(T_G)$, between a maximal clique $X_i$ and its father node, if and only if there is no minimal separator contained into $S$ that appears as an edge in the subtree rooted at $X_i$.
%However, we cannot do that exactly due to some recursive complications in our algorithm. 
The technical motivations behind our choices will be further explained in Sections~\ref{sec:encoding} and~\ref{sec:greedy}.

\item\label{step-2}({\it Candidate set generation.})
This step exploits a result of Section~\ref{sec:representation} which states that, for any $4$-Steiner root \tree~of $G$ and for any clique-intersection \ci, we have $Real(\unlabeledsubtree{\ci}) = \ci$. Equivalently, this result says that in the smallest subtree containing \ci~there can be no other real nodes than those in \ci.
Then, our goal is, for every $\ci \in \CI{G}$, to compute a polynomial-size family ${\cal T}_{\ci}$ of ``candidate subtrees'' whose real nodes are exactly \ci.
Intuitively, ${\cal T}_{\ci}$ should contain all possibilities for \unlabeledsubtree{\ci}~in a ``well-structured'' $4$-Steiner root \tree~(such a root must satisfy additional properties given in Sec.~\ref{sec:restricted-root}).
Note that in practice, we only need to compute this above family for {\em minimal separators} and {\em maximal cliques}.
\begin{itemize}
\item In Section~\ref{sec:minsep} we present an algorithm for computing the collection $({\cal T}_{\sep})_{\sep \in \SEP{G}}$ for the minimal separators. This algorithm serves as a brick-basis construction for computing all the other families.
%For every minimal separator $S$ we compute a polynomial-size family ${\cal T}_S$ of subtrees whose real nodes are exactly $S$ ().
%The  is constructed in such a way that assuming $G$ has a $4$-Steiner root, there must be one such a root $T$ such that $\tree{S} \in {\cal T}_S$ for every minimal separator $S$.
\item Then, we consider in Section~\ref{sec:leaf-case} the maximal cliques $\maxk_i$ that are leaf-nodes of \cliquetree{G}. We use the well-known property that all vertices in $\maxk_i \setminus \sep_i$ are simplicial in order to generalize the algorithm of the previous section to this new case.
%We then proceed similarly for the maximal cliques $X_i$ that are either leaf-nodes () 
\item Finally, we consider the maximal cliques $\maxk_i$ that are internal nodes of \cliquetree{G}~(Section~\ref{sec:super-selection}). Unsurprisingly, several new difficulties arise in the construction of ${\cal T}_{\maxk_i}$. Our bottleneck is solving the following subproblem: compute (up to Steiner equivalence) all possible central nodes and their neighbourhood in any subtree $\unlabeledsubtree{\maxk_i}$ of diameter four.
We solved this subproblem in most situations, {\it e.g.}, when there is a minimal separator $\sep \subseteq \maxk_i$ such that $\unlabeledsubtree{\sep}$ must be a bistar (diameter-three subtree). However in some other situations we failed to do so. That left us with some ``problematic subsets'' called {\em thin branches}: with exponentially many possible ways to include them in candidate subtrees. As a way to circumvent this combinatorial explosion, we also include in ${\cal T}_{\maxk_i}$ some partially constructed subtrees where the thin branches are omitted. We will greedily decide how to include the thin branches in these subtrees ({\it i.e.}, how to complete the construction) at Step~\ref{step-4}. 
\end{itemize}
Correctness of this part mostly follows from our structure theorem of Section~\ref{sec:restricted-root}.

\item\label{step-3}({\it Selection of the encodings.})
For the remaining of the algorithm, let $(\maxk_q,\maxk_{q-1},\ldots,\maxk_0)$ be a post-ordering of the maximal cliques ({\it i.e.}, obtained by depth-first-search traversal of \cliquetree{G}).
We consider the maximal cliques $\maxk_i \in \MAXK{G}$ sequentially, from $i = q$ downto $i=0$.
Indeed if $G$ is a $4$-Steiner power then (by hereditarity), so is the subgraph $G_{i} = (V_{i},E_{i})$ that is induced by all the maximal cliques in the subtree $\cliquetree{G}^{i}$ rooted at $\maxk_{i}$.
Steps~\ref{step-3} and~\ref{step-4} are devoted to the computation of a subset ${\cal T}_{i}$ of $4$-Steiner roots for $G_{i}$.
Specifically, for any $4$-Steiner root $\partialsolution{i}$ of $G_i$ we define the following encoding:
$$\texttt{encode}(\partialsolution{i}) := \left[ \ \labeledsubtree[1]{i}{\sep_i} \ \mid \ \left( dist_{\partialsolution{i}}(r,W_i) \right)_{r \in V(\labeledsubtree[1]{i}{\sep_i})} \ \right].$$
During Step~\ref{step-3} we compute a polynomial-size subset of {\em allowed} encodings for the partial solutions in ${\cal T}_{i}$.
That is, we only want to add in ${\cal T}_{i}$ some partial solutions for which the encoding is in the list.
Formally, we define an auxiliary problem called \DCR, where given an encoding as input, we ask whether there exists a corresponding $4$-Steiner root of $G_i$.
Our set of allowed encodings for ${\cal T}_{i}$ can be seen as a set of inputs for which we need to solve \DCR.

We stress that in order to compute these encodings, we use: the families computed at Step~\ref{step-2}, some properties of the rooted clique-tree \cliquetree{G}, {\em and} some pre-computed subsets ${\cal T}_j$ of partial solutions for the siblings $\maxk_j$ of $\maxk_i$. 
Specifically, if $\maxk_{p(i)} \in \MAXK{G}$ has children nodes $\maxk_{i_1}, \maxk_{i_2}, \ldots, \maxk_{i_p}$, where $p(i) < i_1 < i_2 < \ldots < i_p$ then, we impose that $\sep_{i_1}, \sep_{i_2}, \ldots, \sep_{i_p}$ are ordered by decreasing size.
Step~\ref{step-3} can start for $\maxk_{i_j}$ only after that Steps~\ref{step-3} and~\ref{step-4} are completed for all of $\maxk_{i_{j+1}}, \maxk_{i_{j+2}},\ldots,\maxk_{i_p}$.
This means in particular that executions of Steps~\ref{step-3} and~\ref{step-4} (for different maximal cliques) are intertwined.
%
%Here again this post-ordering is not arbitrary.
%Specifically, 
%If $X_i$ is internal then, let $X_{i_1}, X_{i_2},\ldots, X_{i_p}$ be its children nodes.
%For every $1 \leq j \leq p$ we have if $G$ is a $4$-Steiner power then (by hereditarity), so is the subgraph $G_{i_j} = (V_{i_j},E_{i_j})$ that is induced by all the maximal cliques in the subtree $T_G^{i_j}$ rooted at $X_{i_j}$.
%Our objective in the next Step will be to compute a set ${\cal T}_{i_j}$ of $4$-Steiner roots for $G_{i_j}$.
%As a way to avoid a combinatorial explosion of the number of partial solutions we will need to store, we sketch in Section~\ref{sec:encoding} how to define -- using $({\cal T}_S)_{S \in {\cal S}(G)}$ -- a polynomial-size subset of ``encodings'' for these solutions.
%By combining some local optimization rules with properties of our clique-tree $T_G$, we show that at most one solution per possibility for the encoding needs to be stored.
%In particular, we will explain how our above restrictions on the post-ordering can help us to derive additional distances' constraints from the siblings of a node before we can process it.

\item\label{step-4}({\it Greedy strategy.})
%We end up considering one more time the maximal cliques $X_i \in {\cal K}(G)$ sequentially, from $i = q$ downto $i=0$.
After Step~\ref{step-3} is completed, $\maxk_i$ received a polynomial-size subset of constraints -- {\it a.k.a.}, encodings -- for the $4$-Steiner roots of $G_i$ we want to compute.
For every such constraints, we are left to decide whether there exists a $4$-Steiner root of $G_i$ which satisfies all of them ({\it i.e.}, we must solve \DCR). 
\begin{itemize}
\item {\bf Case $\maxk_i$ is a leaf-node.}
In this situation, $V_i = \maxk_i$.
After Step~\ref{step-2} is completed, we are given a family of all possible subtrees $\unlabeledsubtree{\maxk_i}$.
We are left verifying whether there exists a solution in this family which satisfies all of the constraints.
\item {\bf Case $\maxk_i$ is an internal node.}
Let $\maxk_{i_1}, \maxk_{i_2},\ldots, \maxk_{i_p}$ be the children nodes of $\maxk_i$ in $\cliquetree{G}$.
We will construct ${\cal T}_i$ from the partial solutions in ${\cal T}_{i_1}, {\cal T}_{i_2},\ldots, {\cal T}_{i_p}$.
For that, we try to combine all the possible subtrees $\unlabeledsubtree{\maxk_i}$ (computed during Step~\ref{step-2}) with the partial solutions stored in the sets ${\cal T}_{i_j}$ by using a series of tests based on a maximum-weight matching algorithm (Section~\ref{sec:greedy}).
-- We use the same strategy in order to incorporate first the so called thin branches, so as to complete the construction of the subtrees $\unlabeledsubtree{\maxk_i}$.--
We stress the intriguing relationship between our approach and the implementation of the {\tt alldifferent} constraint in constraint programming~\cite{Reg94}.
%\item {\bf Case $X_i = X_0$ is the root.}
\end{itemize}

\item\label{step-5}({\it Output.}) Overall since $G_0 = G$, we have $G$ is a $4$-Steiner power if and only if ${\cal T}_0 \neq \emptyset$.
Furthermore, any tree $T \in {\cal T}_0$ is a $4$-Steiner root of $G$.
\end{enumerate}

\section{Playing with the root}\label{sec:representation}

Some general relationships between $k$-Steiner roots and clique-intersections are proved in Section~\ref{sec:structure}, for any $k$.
These structural results will be the cornerstone of our algorithm and its analysis.
Before presenting all these properties, we establish several useful facts on trees in Section~\ref{sec:tree} (most of them being likely to be known).

\subsection{General results on trees}\label{sec:tree}

We first recall the {\em unimodality} property for the eccentricity function on trees (\ref{pty-lem:tree:1} below), as well as some other related properties.
They mostly follow from a seminal paper of Jordan~\cite{Jor69}.
-- See also~\cite{CDHVA18+} for a recent example of their applications to other graph problems. --

\begin{lemma}[folklore]\label{lem:tree}
The following hold for any tree $T$:
\begin{enumerate}[label=\textbullet,ref=P-\ref{lem:tree}.\theenumi]
\item\label{pty-lem:tree:1}({\small\ref{pty-lem:tree:1}.}) For every node $v \in V(T)$ we have $ecc_T(v) = dist_T(v, {\cal C}(T)) + rad(T)$;
\item\label{pty-lem:tree:2}({\small\ref{pty-lem:tree:2}.}) Every diametral path in $T$ contains all the nodes in ${\cal C}(T)$ (as its middle nodes);
\item\label{pty-lem:tree:3}({\small\ref{pty-lem:tree:3}.}) ${\cal C}(T)$ is reduced to a node if $diam(T)$ is even, and to an edge if $diam(T)$ is odd;
\item\label{pty-lem:tree:4}({\small\ref{pty-lem:tree:4}.}) $rad(T) = \left\lceil diam(T)/2 \right\rceil$.
\end{enumerate}
\end{lemma}

Based on the above, the following properties on subtree intersections can be derived:

%%\begin{lemma}\label{lem:center-tree}
%%If $T$ is a tree with at least two nodes and $v \in V(T)$ is a leaf then, ${\cal C}(T) \cap {\cal C}(T \setminus v) \neq \emptyset$.
%%\end{lemma}
%%
%%\begin{proof}
%%If $rad(T) = rad(T \setminus v)$ then, ${\cal C}(T) \subseteq {\cal C}(T \setminus v)$, and so we are done.
%%From now on assume $rad(T) = rad(T \setminus v) + 1$.
%%This implies $diam(T) = 2d+1$ is odd and $diam(T \setminus v) = 2d$.
%%In particular, any diametral path in $T$ was between $u$ and $v$.
%%Fix such a diametral pair $u,v$ and let $w$ be the unique node adjacent to $v$ in $T$.
%%We have that $u,w$ is a diametral pair in $T \setminus v$.
%%Moreover, the two nodes in ${\cal C}(T)$ are at a distance $d$ and $d+1$ from $v$, respectively, onto the unique $uv$-path in $T$ (this follows from the unimodality property and the fact that every central node must be on the $uv$-path).
%%Therefore, the vertex at a distance exactly $d = rad(T \setminus v)$ from $w$ (hence, $d+1 = rad(T)$ from $v$) onto the unique $uw$-path must be in ${\cal C}(T)$, and by the aforementioned properties on trees this vertex must be in ${\cal C}(T \setminus v)$.  
%%\end{proof}

\begin{lemma}\label{lem:tree-intersection}
Given a tree $T$ let $T_1,T_2$ be two subtrees such that $diam(T_1 \cap T_2) = diam(T_1)$.
Then, we have $diam(T_1 \cup T_2) = diam(T_2)$.
\end{lemma}

\begin{proof}
First we claim that ${\cal C}(T_1 \cap T_2) = {\cal C}(T_1)$.
Indeed, since $T_1 \cap T_2$ and $T_1$ are trees with equal diameter, and we have $T_1 \cap T_2 \subseteq T_1$, every diametral path for $T_1 \cap T_2$ is also a diametral path for $T_1$.
Furthermore, since on every diametral path in a tree, the middle nodes are exactly the center nodes (Prop.~\ref{pty-lem:tree:2}), we obtain as claimed that ${\cal C}(T_1 \cap T_2) = {\cal C}(T_1)$.

\begin{figure}[h!]
\centering
\includegraphics[width=.35\textwidth]{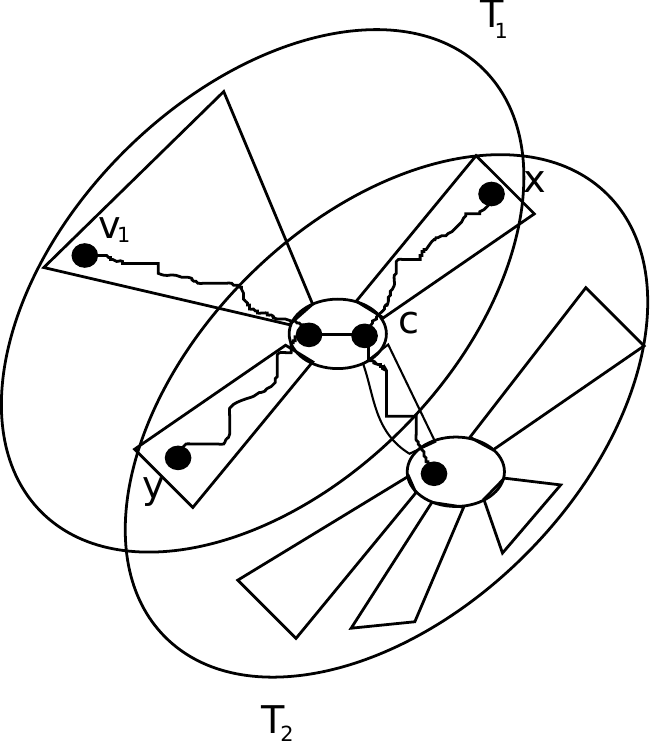}
\caption{To the proof of Lemma~\ref{lem:tree-intersection}.}
\label{fig:tree-intersections}
\end{figure}

Then, let $x,y \in V(T_1 \cap T_2)$ be the two ends of a diametral path in the subtree $T_1 \cap T_2$.
We set $z \in \{x,y\}$ maximizing $dist_T(z,{\cal C}(T_2))$ and we claim that, for every $v_1 \in V(T_1)$, $dist_T(v_1,{\cal C}(T_2)) \leq dist_T(z,{\cal C}(T_2))$.
Before we prove this claim, let us explain why this proves the lemma.
Every node of $V(T_1)$ is at a distance $\leq dist_T(z,{\cal C}(T_2)) + rad(T_2)$ from any node in $V(T_2)$.
By unimodality (Prop.~\ref{pty-lem:tree:1}), $ecc_{T_2}(z) = dist_T(z,{\cal C}(T_2)) + rad(T_2) \leq diam(T_2)$, and so, $diam(T_1 \cup T_2) = diam(T_2)$.

Finally, in order to prove this above claim there are two cases.
\begin{itemize}
\item
First assume ${\cal C}(T_1) \subseteq {\cal C}(T_2)$.
We recall that since the unique $xy$-path in $T$ must contain all of ${\cal C}(T_1)$ (Prop.~\ref{pty-lem:tree:2}), there can be no component of $T \setminus {\cal C}(T_1)$ that contains both $x,y$.
In particular, there exists $z \in \{x,y\}$ such that no component of $T \setminus {\cal C}(T_1)$ can both contain $z$ and intersect ${\cal C}(T_2)\setminus {\cal C}(T_1)$.
Then, $dist_T(z,{\cal C}(T_2)) = dist_T(z,{\cal C}(T_1))$.
Furthermore by unimodality (Prop.~\ref{pty-lem:tree:1}) every node $v_1 \in V(T_1)$ has eccentricity $dist_T(v_1,{\cal C}(T_1)) + rad(T_1)$.
Since $z$ is an end in a diametral path of $T_1$ it maximizes $dist_T(z,{\cal C}(T_1))$, and so, for every $v_1 \in V(T_1)$ we have $dist_T(v_1,{\cal C}(T_2)) \leq dist_T(v_1,{\cal C}(T_1)) \leq dist_T(z,{\cal C}(T_1)) = dist_T(z,{\cal C}(T_2))$.
\item Otherwise, let $c \in {\cal C}(T_1)$ minimize $dist_T(c,{\cal C}(T_2))$.
Note that since we have ${\cal C}(T_1) \not\subseteq {\cal C}(T_2)$, there is a unique possible choice for $c$.
Furthermore, every $v_1 \in V(T_1)$ satisfies $dist_T(v_1,{\cal C}(T_2)) \leq dist_T(v_1,c) +  dist_T(c,{\cal C}(T_2)) \leq rad(T_1) +  dist_T(c,{\cal C}(T_2))$, and we will show this upper-bound is reached for at least one of $x$ or $y$.
Specifically, we can refine one observation from the previous case as follows: there exists $z \in \{x,y\}$ such that no component of $T \setminus {\cal C}(T_1)$ can both contain $z$ and intersect ${\cal C}(T_2)\setminus {\cal C}(T_1)$; and in the special case where ${\cal C}(T_1)$ is an edge, $c$ is not the closest central node to $z$.
In this situation, $dist_T(z,c) = rad(T_1)$ and the path between $z$ and ${\cal C}(T_2)$ goes by $c$.
See Fig.~\ref{fig:tree-intersections} for an illustration.
\end{itemize}
In both cases we obtain, as claimed, $dist_T(v_1,{\cal C}(T_2)) \leq dist_T(z,{\cal C}(T_2))$ for every $v_1 \in V(T_1)$. 
\end{proof}

\begin{lemma}\label{lem:center-inclusion}
Given a tree $T$ let $T_1,T_2$ be two subtrees such that ${\cal C}(T_1) \subseteq {\cal C}(T_2)$.
Then, we have that $diam(T_1 \cup T_2) = \max\{diam(T_1),diam(T_2)\}$.
\end{lemma}

\begin{proof}
Since ${\cal C}(T_1) \subseteq {\cal C}(T_2)$ we have for every $v_1 \in V(T_1)$: 
$$ecc_{T_1 \cup T_2}(v_1) \leq dist_{T}(v_1,{\cal C}(T_1)) + \max\{rad(T_1),rad(T_2)\}.$$
By the unimodality property (Property~\ref{pty-lem:tree:1}) we have that $dist_T(v_1,{\cal C}(T_1)) = ecc_{T_1}(v_1) - rad(T_1) \leq diam(T_1) - rad(T_1)$, and so by Property~\ref{pty-lem:tree:4}: 
$$dist_T(v_1,{\cal C}(T_1)) \leq \left\lfloor diam(T_1)/2 \right\rfloor \leq \max\{ \left\lfloor diam(T_1)/2 \right\rfloor, \left\lfloor diam(T_2)/2 \right\rfloor \}.$$ 
In the same way, Property~\ref{pty-lem:tree:4} implies that:
$$\max\{rad(T_1),rad(T_2)\} = \max\{ \left\lceil diam(T_1)/2 \right\rceil, \left\lceil diam(T_2)/2 \right\rceil \}.$$
We so obtain that $ecc_{T_1 \cup T_2}(v_1) \leq \max\{ diam(T_1), diam(T_2) \}$.

\medskip
Then, for every $v_2 \in V(T_2)$: 
\begin{align*}
ecc_{T_1 \cup T_2}(v_2) &\leq dist_T(v_2,{\cal C}(T_2)) + \max\{ rad(T_2), diam({\cal C}(T_2)) + rad(T_1) \} \\
&\leq dist_T(v_2,{\cal C}(T_2)) + \max\{ rad(T_2), 1 + rad(T_1) \}.
\end{align*}
We may assume $rad(T_1) \geq rad(T_2)$ since otherwise, $ecc_{T_1 \cup T_2}(v_2) \leq dist_T(v_2,{\cal C}(T_2)) + rad(T_2) = ecc_{T_2}(v_2) \leq diam(T_2)$ by unimodality (Property~\ref{pty-lem:tree:1}).
In particular we claim that it implies $diam(T_1) \geq diam(T_2)$.
Indeed, by Property~\ref{pty-lem:tree:4} we must have $\left\lceil diam(T_1)/2\right\rceil \geq \left\lceil diam(T_2)/2 \right\rceil$, and so $diam(T_1) \geq diam(T_2) - 1$.
Suppose for the sake of contradiction $diam(T_1) = diam(T_2) - 1$.
By the hypothesis we also have ${\cal C}(T_1) \subseteq {\cal C}(T_2)$.
Therefore, by Property~\ref{pty-lem:tree:3} we obtain that $diam(T_1)$ and $diam(T_2)$ must be even and odd, respectively.
But then, we cannot have $\left\lceil diam(T_1)/2\right\rceil \geq \left\lceil diam(T_2)/2 \right\rceil$, a contradiction.
So, we proved as claimed $diam(T_1) \geq diam(T_2)$.
There are now two cases to consider:
\begin{itemize}
\item Case $diam(T_1) = diam(T_2)$. Then, ${\cal C}(T_1) = {\cal C}(T_2)$ and we can strengthen our previous inequality as follows:  $ecc_{T_1 \cup T_2}(v_2) \leq dist_T(v_2,{\cal C}(T_2)) + \max\{ rad(T_2), rad(T_1) \} \leq diam(T_2)$.
\item Case $diam(T_1) > diam(T_2)$. Recall that $dist_T(v_2,{\cal C}(T_2)) \leq \left\lfloor diam(T_2)/2 \right\rfloor$. In particular, either $diam(T_1) \geq diam(T_2) + 2$, and so, $dist_T(v_2,{\cal C}(T_2)) \leq \left\lfloor diam(T_1)/2 \right\rfloor - 1$; or $diam(T_1) = diam(T_2) + 1$ but then, since by the hypothesis we have ${\cal C}(T_1) \subseteq {\cal C}(T_2)$, by Property~\ref{pty-lem:tree:3} $diam(T_1)$ must be even, and so, $dist_T(v_2,{\cal C}(T_2)) \leq \left\lfloor diam(T_1)/2 \right\rfloor - 1$ also in this case.
Overall: 
\begin{align*}
ecc_{T_1 \cup T_2}(v_2) &\leq dist_T(v_2,{\cal C}(T_2)) + \max\{ rad(T_2), diam({\cal C}(T_2)) + rad(T_1) \} \\
&\leq \left\lfloor diam(T_1)/2 \right\rfloor - 1 + rad(T_1) + 1 = diam(T_1).
\end{align*}
\end{itemize}
Therefore, in both cases we obtain $diam(T_1 \cup T_2) \leq \max\{diam(T_1),diam(T_2)\}$.
\end{proof}

\subsection{A structure theorem}\label{sec:structure}

We are now ready to state the main result in this section:

\begin{theorem}\label{thm:clique-intersection}
Given $G=(V,E)$ and $\tree$ any $k$-Steiner root of $G$, the following properties hold for any clique-intersection $\ci \in \CI{G}$:
\begin{enumerate}[label=\textbullet,ref=P-\ref{thm:clique-intersection}.\theenumi]
\item\label{pty-ci:1}({\small\ref{pty-ci:1}.}) We have $Real(\unlabeledsubtree{\ci}) = \ci$ and $diam(\unlabeledsubtree{\ci}) \leq k$;
\item\label{pty-ci:2}({\small\ref{pty-ci:2}.}) If $\labeledsubtree[2]{\ci}{} \supset \unlabeledsubtree{\ci}$ then, either $\ci = Real(\labeledsubtree[2]{\ci}{})$ or $diam(\labeledsubtree[2]{\ci}{}) > diam(\unlabeledsubtree{\ci})$;
%There is no supertree $\labeledsubtree[2]{\ci}{} \supset \unlabeledsubtree{\ci}$ with $\ci \subset Real(\labeledsubtree[2]{\ci}{})$ and $diam(\labeledsubtree[2]{\ci}{}) = diam(\unlabeledsubtree{\ci})$;
%\item\label{pty-ci:3}({\small\ref{pty-ci:3}.}) If $\ci \subset \ci' \in \CI{G}$ then, $diam(\unlabeledsubtree{\ci}) < diam(\unlabeledsubtree{\ci'})$.
\item\label{pty-ci:3}({\small\ref{pty-ci:3}.}) If $\centre{\unlabeledsubtree{\ci}} \subseteq \centre{\unlabeledsubtree{\ci'}}$ then, $\ci \cup \ci'$ is a clique of $G$.
\end{enumerate}
\end{theorem}

\begin{proof}
We prove each property separately.

({\it Proof of Property~\ref{pty-ci:1}.}) First assume $\ci \in \MAXK{G}$ to be a maximal clique.
Since all leaves of $\unlabeledsubtree{\ci}$ are in $\ci$, $diam(\unlabeledsubtree{\ci}) = \max_{u,v \in \ci} dist_{\tree}(u,v)$.
By the hypothesis $\tree$ is a $k$-Steiner root of $G$, and so, since $\ci$ is a clique of $G$, $\max_{u,v \in \ci} dist_{\tree}(u,v) \leq k$.
In particular, $diam(\unlabeledsubtree{\ci}) \leq k$, that implies in turn the vertices of $Real(\unlabeledsubtree{\ci})$ must induce a clique of $G$.
We can conclude that $Real(\unlabeledsubtree{\ci}) = \ci$ by maximality of $\ci$.
More generally, let $\ci = \bigcap_{i=1}^{\ell} \maxk_i$, for some family $\maxk_1,\maxk_2,\ldots,\maxk_{\ell} \in \MAXK{G}$.
Clearly, $\unlabeledsubtree{\ci} \subseteq \bigcap_{i=1}^{\ell} \unlabeledsubtree{\maxk_i}$, and so, $\ci \subseteq Real(\unlabeledsubtree{\ci}) \subseteq \bigcap_{i=1}^{\ell} Real(\unlabeledsubtree{\maxk_i})$.
As we proved before, $Real(\unlabeledsubtree{\maxk_i}) = \maxk_i$ for every $1 \leq i \leq \ell$, and so, $Real(\unlabeledsubtree{\ci}) \subseteq \bigcap_{i=1}^{\ell} \maxk_i = \ci$.
Altogether combined, we obtain that $Real(\unlabeledsubtree{\ci}) = \ci$.

({\it Proof of Property~\ref{pty-ci:2}.}) Let $\labeledsubtree[2]{\ci}{} \supset \unlabeledsubtree{\ci}$ be such that $diam(\labeledsubtree[2]{\ci}{}) = diam(\unlabeledsubtree{\ci})$.
We claim $Real(\labeledsubtree[2]{\ci}{}) = \ci$, that will prove the second part of the theorem.
Indeed, for any maximal clique $\maxk_j$ that contains $\ci$, we have $diam(\unlabeledsubtree{\maxk_j} \cap \labeledsubtree[2]{\ci}{}) \geq diam(\unlabeledsubtree{\ci}) = diam(\labeledsubtree[2]{\ci}{})$, and so, $diam(\unlabeledsubtree{\maxk_j} \cup \labeledsubtree[2]{\ci}{}) = diam(\unlabeledsubtree{\maxk_j}) \leq k$ by Lemma~\ref{lem:tree-intersection}.
It implies $Real(\labeledsubtree[2]{\ci}{}) \subseteq \maxk_j$.
Furthermore, since $\ci \in \CI{G}$, \ci~is exactly the intersection of all the maximal cliques $\maxk_j$ that contains it, thereby proving the claim.
%
%({\it Proof of Property~\ref{pty-ci:3}.}) Furthermore, assume $\ci \subset \ci'$.
%Since $\unlabeledsubtree{\ci'} \supset \unlabeledsubtree{\ci}$, we cannot have $diam(\unlabeledsubtree{\ci'}) = diam(\unlabeledsubtree{\ci})$ (otherwise, $\ci' = \ci$ by Property~\ref{pty-ci:2}).
%Therefore, $diam(\unlabeledsubtree{\ci'}) > diam(\unlabeledsubtree{\ci})$. 
%

({\it Proof of Property~\ref{pty-ci:3}.}) Finally, assume $\centre{\unlabeledsubtree{\ci}} \subseteq \centre{\unlabeledsubtree{\ci'}}$.
By Lemma~\ref{lem:center-inclusion} we obtain that $diam(\unlabeledsubtree{\ci} \cup \unlabeledsubtree{\ci'}) = \max\{diam(\unlabeledsubtree{\ci}),diam(\unlabeledsubtree{\ci'})\} \leq k$.
In particular, $\ci \cup \ci'$ must be a clique of $G$.
\end{proof}

\begin{remark}\label{rk:min-sep-chain}
Property~\ref{pty-ci:2} implies that in any $k$-Steiner power $G$, there can be no chain of more than $k$ minimal separators ordered by inclusion.
Indeed, let $\sep_1 \subset \sep_2 \subset \ldots \subset \sep_{\ell}$ be such a chain.
Then, in any $k$-Steiner root $\tree$ of $G$, $0 \leq diam(\unlabeledsubtree{\sep_1}) < diam(\unlabeledsubtree{\sep_2}) < \ldots < diam(\unlabeledsubtree{\sep_{\ell}}) < k$.
\end{remark}

\begin{remark}\label{rk:size-subtrees}
Given $G = (V,E)$ and a $k$-Steiner root $\tree$ of $G$, we can also easily derive from Property~\ref{pty-ci:1} that we have $|V(\unlabeledsubtree{\ci})| = {\cal O}(k \cdot |\ci|)$ for every clique-intersection $\ci \in \CI{G}$. 
To see that, fix any $r \in V(\unlabeledsubtree{\ci})$.
For every $x \in \ci$ that is a leaf of $\unlabeledsubtree{\ci}$, the unique $xr$-path has length $\leq k$.
Overall, these above paths cover all of $\unlabeledsubtree{\ci}$, thereby proving as desired $|V(\unlabeledsubtree{\ci})| = {\cal O}(k \cdot |\ci|)$.
This fact is often used implicitly in our analysis.
\end{remark}

Before ending this section, we slightly strenghten Property~\ref{pty-ci:3} of Theorem~\ref{thm:clique-intersection}, as follows:

\begin{lemma}\label{lem:no-center-intersect}
Given $G=(V,E)$ and \tree~any $2k$-Steiner root of $G$, we have $\centre{\unlabeledsubtree{\maxk_i}} \cap \centre{\unlabeledsubtree{\maxk_j}} = \emptyset$ for any two different maximal cliques $\maxk_i,\maxk_j \in \MAXK{G}$.
\end{lemma}

\begin{proof}
Suppose for the sake of contradiction $\centre{\unlabeledsubtree{\maxk_i}} \cap \centre{\unlabeledsubtree{\maxk_j}} \neq \emptyset$, and let $r \in \centre{\unlabeledsubtree{\maxk_i}} \cap \centre{\unlabeledsubtree{\maxk_j}}$.
By Theorem~\ref{thm:clique-intersection} (Prop.~\ref{pty-ci:1}), $\max\{diam(\unlabeledsubtree{\maxk_i}),diam(\unlabeledsubtree{\maxk_j}) \leq 2k$. Then, it follows from Prop.~\ref{pty-lem:tree:4} of Lemma~\ref{lem:tree} that any vertex of $\unlabeledsubtree{\maxk_i} \cup \unlabeledsubtree{\maxk_j}$ is at a distance $\leq k$ from $r$ in $\tree$.
In particular, $diam(\unlabeledsubtree{\maxk_i}\cup\unlabeledsubtree{\maxk_j}) \leq 2k$, and so, $\maxk_i \cup \maxk_j$ is a clique of $G$.
The latter contradicts the fact that $\maxk_i,\maxk_j$ are maximal cliques.
\end{proof}

\section{Well-structured $4$-Steiner roots}\label{sec:restricted-root}

We refine our results in the previous Section when $k=4$.
Let us start motivating our approach.
Given $G=(V,E)$ we recall that one of our intermediate goals is to compute, by dynamic programming on a clique-tree, subsets ${\cal T}_i$ of $4$-Steiner roots for some collection of subgraphs $G_i$, with the following additional property: assuming $G$ is a $4$-Steiner power, there must be a partial solution in ${\cal T}_i$ which can be extended to a $4$-Steiner root for $G$ (cf. Sec.~\ref{sec:algo-highlight}, Step~\ref{step-3}).
Ideally, one should store {\em all} the possible $4$-Steiner roots for $G_i$, however this leads to a combinatorial explosion.
In order to (partly) overcome this issue, we introduce the following important notion:

\begin{definition}\label{def:x-free}
Given $G=(V,E)$ and $\ci \in \CI{G}$, a vertex $v \in \ci$ is called {\em \ci-constrained}
%\footnote{We presented a simpler categorisation in an earlier version of this work. Unfortunately, we were omitting one important case doing so.} 
if it satisfies one of the following two conditions:
\begin{enumerate}
\item either there is another clique-intersection $\ci' \subset \ci$ such that $v \in \ci'$ and $|\ci'| \geq 2$ (we call $v$ {\em internally $\ci$-constrained});
\item or there exist $\ci_1,\ci_2 \in \CI{G}$ such that $\ci \subset \ci_1$ and $\ci \cap \ci_2 = \{v\} \subset \ci_1 \cap \ci_2$ (we call $v$ {\em $(\ci,\ci_1,\ci_2)$-sandwiched}). 
\end{enumerate} 
A vertex $v \in \ci$ that is not $\ci$-constrained is called {\em $\ci$-free}.
\end{definition}
%\begin{definition}\label{def:x-free}
%Given $G=(V,E)$ and $\ci \in \CI{G}$, a vertex $v \in \ci$ is called {\em \ci-free} if for any other $\ci' \in \CI{G}$ we have either $v \notin \ci'$, $\ci \subseteq \ci'$ or $\ci \cap \ci' = \{v\}$. A vertex $v \in \ci$ that is not \ci-free is called {\em \ci-constrained}.
%\end{definition}

\begin{figure}[h!]
\centering
\includegraphics[width=.25\textwidth]{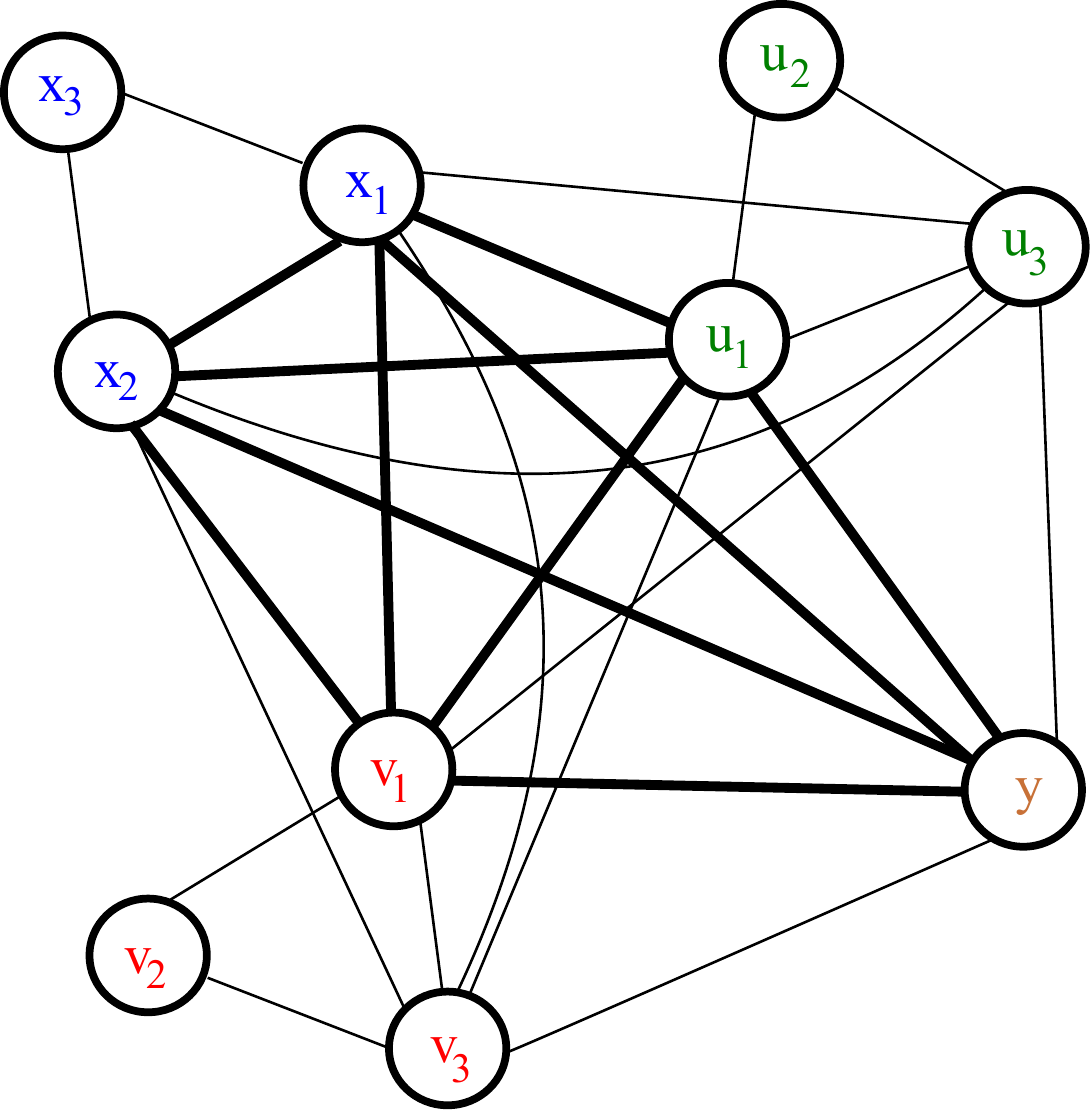}
\caption{Five maximal cliques $\maxk_1 = \{x_1,x_2,x_3\}, \ \maxk_2 = \{u_1,u_2,u_3\}, \ \maxk_3 = \{v_1,v_2,v_3\}, \ \maxk_4 = \{y,x_1,x_2,u_1,v_1,u_3\} \ \text{and} \ \maxk_5 = \{y,x_1,x_2,u_1,v_1,v_3\}$. In the clique-intersection $\ci = \{y,x_1,x_2,u_1,v_1\}$ both vertices $x_1,x_2$ are internally $\ci$-constrained (for $\ci' = \ci \cap \maxk_1$), vertex $u_1$ is $(\ci,\maxk_4,\maxk_2)$-sandwiched, vertex $v_1$ is $(\ci,\maxk_5,\maxk_3)$-sandwiched and vertex $y$ is $\ci$-free.}
\label{fig:all-categories}
\end{figure}

We refer to Fig.~\ref{fig:all-categories} for an illustration of all possibilities.
Our study reveals on the one hand that \ci-constrained vertices have a very rigid structure. 
It seems on the other hand that \ci-free vertices are completely unstructured and mostly responsible for the combinatorial explosion of possibilities for $\unlabeledsubtree{\ci}$.
As our main contribution in this section we prove that we can always force the \ci-free vertices to be {\em leaves} of this subtree $\unlabeledsubtree{\ci}$, thereby considerably reducing the number of possibilities for the latter.
Specifically:

\begin{theorem}\label{thm:x-free}
Let $G=(V,E)$ be a $4$-Steiner power.
There always exists a $4$-Steiner root \tree~of $G$ where, for any clique-intersection $\ci \in \CI{G}$: 
%any real node in ${\cal C}(\tree{X})$ is $X$-central, and 
%. Moreover:
\begin{enumerate}[label=\textbullet,ref=P-\ref{thm:x-free}.\theenumi]
\item\label{pty-x-free:1}({\small\ref{pty-x-free:1}.}) all the \ci-free vertices are leaves of $\unlabeledsubtree{\ci}$ with maximum eccentricity $diam(\unlabeledsubtree{\ci})$;
\item\label{pty-x-free:2}({\small\ref{pty-x-free:2}.}) there is a node $c \in \centre{\unlabeledsubtree{\ci}}$ such that for every \ci-free vertex $v$, except maybe one, $dist_{\tree}(v,c) = dist_{\tree}(v,\centre{\unlabeledsubtree{\ci}}) $;
\item\label{pty-x-free:3}({\small\ref{pty-x-free:3}.}) all the internal nodes on a path between $\centre{\unlabeledsubtree{\ci}}$ and a \ci-free vertex are Steiner nodes of degree two;
\item\label{pty-x-free:4}({\small\ref{pty-x-free:4}.}) and if $\ci \in \MAXK{G}$ and it has a \ci-free vertex then, $diam(\unlabeledsubtree{\ci}) = 4$.
\end{enumerate}
\end{theorem}

Theorem~\ref{thm:x-free} is proved by carefully applying a set of operations on an arbitrary $4$-Steiner root until it satisfies all of the desired properties.
We give two examples of such operations in Fig.~\ref{fig:s-free} and~\ref{fig:s-free-bis} (see the proof of Theorem~\ref{thm:x-free} in order to better understand these two examples).
It is crucial for the proof that in any $4$-Steiner root of $G$ all the minimal separators yield subtrees of diameter at most three.

In the remaining of the paper, we call a $4$-Steiner root with the above properties {\em well-structured}.
It will appear in Sec.~\ref{sec:ci-subtrees},~\ref{sec:encoding}, and~\ref{sec:greedy} that we only add partial solutions in the subsets ${\cal T}_i$ which are in some sense ``close'' to the subset of well-structured partial solutions.
We first prove Theorem~\ref{thm:x-free} for maximal cliques (Section~\ref{sec:simplicial}).
This first part of the proof looks easier to generalize to larger values of $k$.
Then, we prove the result in its full generality in Section~\ref{sec:x-free}.

\subsection{The case of (Almost) Simplicial vertices}\label{sec:simplicial}

Let $\maxk_i \in \MAXK{G}$ be fixed.
We start giving a simple characterization of $\maxk_i$-free vertices in terms of simplicial vertices and cut-vertices.
Then, we prove Theorem~\ref{thm:x-free} in the special case when $\ci$ is a maximal clique.

\begin{lemma}\label{lem:almost-simplicial-characterization}
Given $G=(V,E)$ and $\maxk_i \in \MAXK{G}$, a vertex $v \in \maxk_i$ is $\maxk_i$-free if and only if:
\begin{itemize}
\item either it is simplicial;
\item or it is a cut-vertex, and there is no other minimal separator of $G$ contained into $\maxk_i$ that can also contain $v$.
\end{itemize}
\end{lemma}	

\begin{proof}
By maximality of $\maxk_i$, a vertex $v \in \maxk_i$ can only be either $\maxk_i$-free or {\em internally} $\maxk_i$-constrained.
In particular, $v$ is $\maxk_i$-free if and only if for any other $\maxk_j \in \MAXK{G}$ we have either $v \notin\maxk_j$ or $\maxk_i \cap \maxk_j = \{v\}$.
As an extremal case, a vertex $v \in \maxk_i$ is $\maxk_i$-free if is not contained into any other maximal clique, and that is the case if and only if $v$ is simplicial.
Thus, from now on assume $v$ is not simplicial.
If $v \in \maxk_i \cap \maxk_j$ then, in any clique-tree $\cliquetree{G}$ of $G$, the vertex $v$ and more generally, all of $\maxk_i \cap \maxk_j$, is contained into all the minimal separators that label an edge of the $\maxk_i\maxk_j$-path in $\cliquetree{G}$. 
%that disconnects $X_i \setminus X_j$ from $X_j \setminus X_i$.
This implies that there is always a largest clique-intersection $\ci \subset \maxk_i$ containing $v$ that is a minimal separator.
Hence a non simplicial $v \in \maxk_i$ is $\maxk_i$-free if and only if it is a cut-vertex and there is no other minimal separator in $\maxk_i$ that contains this vertex.
\end{proof}

\begin{lemma}\label{lem:almost-simplicial-placement}
Let $G=(V,E)$ be a $4$-Steiner power.
There exists a $4$-Steiner root $\tree$ of $G$ such that the following hold for any maximal clique $\maxk_i$:
\begin{itemize}
\item If there is at least one $\maxk_i$-free vertex then, $diam(\unlabeledsubtree{\maxk_i}) = 4$;
\item every $\maxk_i$-free vertex $v$ is a leaf of $\unlabeledsubtree{\maxk_i}$. Moreover $dist_{\tree}(v,\centre{\unlabeledsubtree{\maxk_i}}) = 2$ and the internal node onto the unique $v\centre{\unlabeledsubtree{\maxk_i}}$-path is a degree-two Steiner node.
\end{itemize}
\end{lemma}

\begin{figure}
\centering
\includegraphics[width=\textwidth]{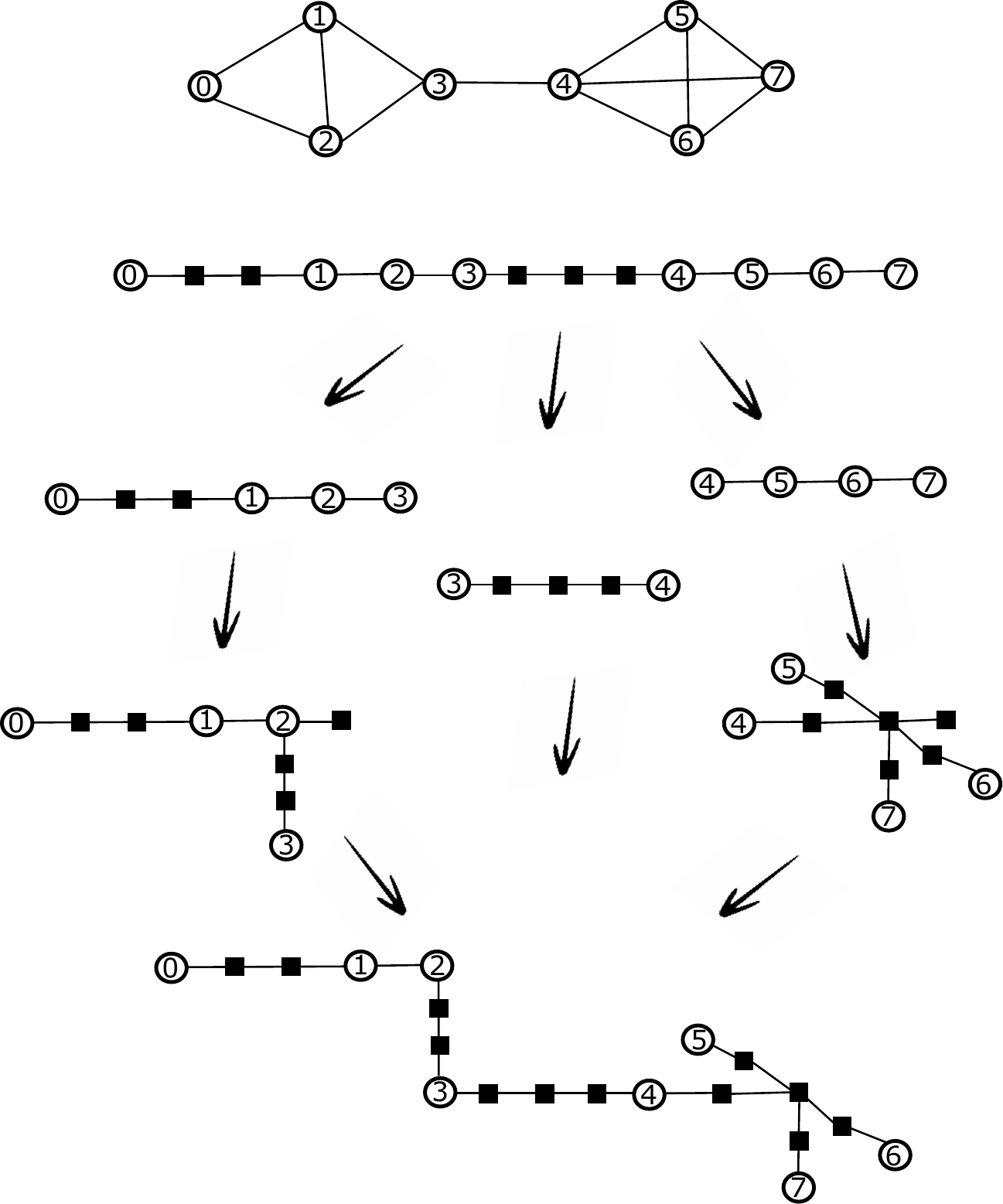}
\caption{The transformation of Lemma~\ref{lem:almost-simplicial-placement} applied to an arbitrary $4$-Steiner root.}
\label{fig:almost-simplicial}
\end{figure}

\begin{proof}
We give an illustration of the proof in Fig.~\ref{fig:almost-simplicial}.
First we pick an arbitrary $4$-Steiner root $\tree$ of $G$, that exists by the hypothesis.
Define ${\cal F}ree_1$ to be the set of all the cut-vertices in $G$ that are $\maxk_i$-free for some $\maxk_i \in \MAXK{G}$.
We now proceed by induction on $|{\cal F}ree_1|$.

Assume ${\cal F}ree_1 = \emptyset$ for the base case.
For every maximal clique $\maxk_i$, by Lemma~\ref{lem:almost-simplicial-characterization}, the $\maxk_i$-free vertices are exactly the simplicial vertices in $\maxk_i$.
We consider all the simplicial vertices $v \in \maxk_i$ sequentially, and we proceed as follows.
Let $c_i \in \centre{\unlabeledsubtree{\maxk_i}}$ minimize $dist_{\tree}(v,c_i)$ (possibly, $v = c_i$).
We first replace $v$ by a Steiner node $\alpha$.
In doing so, we get a $4$-Steiner root $\tree[2]$ for $G \setminus v$.
Then, let $c_i'$ be either $c_i$ (if $c_i \neq v$) or $\alpha$ (if $c_i = v$).
We connect $v$ to $c_i'$ via a path of length exactly $4 - \max_{u \in \maxk_i \setminus \{v\}}dist_{\tree[2]}(c_i',u)$ of which all internal nodes are Steiner.
In doing so, we obtain a tree $\tree[3]$ such that $Real(\tree[3]) = V$.
By construction, $\max_{u \in \maxk_i \setminus \{v\}}dist_{\tree[2]}(c_i',u) \leq ecc_{\unlabeledsubtree{\maxk_i}}(c_i)$.
Since in addition $diam(\unlabeledsubtree{\maxk_i}) \leq 4$, we have from Prop.~\ref{pty-lem:tree:4} of Lemma~\ref{lem:tree} $ecc_{\unlabeledsubtree{\maxk_i}}(c_i) \leq 2$. 
It implies: $$dist_{\tree}(v,c_i) \leq ecc_{\unlabeledsubtree{\maxk_i}}(c_i) \leq 4 - \max_{u \in \maxk_i \setminus \{v\}}dist_{\tree[2]}(c_i',u) = dist_{\tree[3]}(v,c_i').$$
As a result, the distances between real nodes can only increase compared to $\tree$, and this new tree $\tree[3]$ we get keeps the property of being a $4$-Steiner root of $G$.
Furthermore, $diam(\labeledsubtree[3]{}{\maxk_i}) = 4$ and the unique central node in $\centre{\labeledsubtree[3]{}{\maxk_i}}$ is onto the $vc_i'$-path by construction.
%Here it is also important to observe that, since $v$ is only contained into $\maxk_i$ and our transformation can only increase the distances between the real nodes, $\maxk_i$ and $v$ cannot falsify the properties of the lemma at any further loop.
We observe that $\maxk_i$ and $v$ cannot falsify the properties of the lemma at any further loop.
Overall, after this first phase is done we may assume that all the simplicial vertices $v$ are contained into some maximal clique $\maxk_i$ such that: $diam(\unlabeledsubtree{\maxk_i}) = 4$, $v$ is a leaf of $\unlabeledsubtree{\maxk_i}$ such that $dist_{\tree}(v,\centre{\unlabeledsubtree{\maxk_i}}) = 2$, and the internal node onto the $v\centre{\unlabeledsubtree{\maxk_i}}$-path is Steiner and has degree two.

From now on we assume ${\cal F}ree_1 \neq \emptyset$.
Let $v \in {\cal F}ree_1$ and let $C_1,C_2,\ldots,C_{\ell}$ be the connected components of $G \setminus v$.
For every $i \in \{1,2,\ldots,\ell\}$, the graph $G_i := G[C_i \cup \{v\}]$ is a $4$-Steiner power as this is a hereditary property.
Specifically, given a fixed $4$-Steiner root $\tree$ for $G$, we obtain a $4$-Steiner root $\tree[i]$ for $G_i$ by replacing every vertex in $V(G) \setminus V(G_i)$ by a Steiner node.
By induction, we can modify all the subtrees $\tree[i]$ into some new subtrees $\tree[i,2]$ that satisfy the properties of the lemma w.r.t. $G_i$.
Overall, by identifying all the $\tree[i,2]$'s at $v$, one obtains a tree $\tree[2]$.
We claim that $\tree[2]$ satisfies the two properties of the lemma.
Indeed, it follows from the characterization of Lemma~\ref{lem:almost-simplicial-characterization} that, for any maximal clique $\maxk_j \subseteq C_i \cup \{v\}$, the $\maxk_j$-free vertices in $G$ are still $\maxk_j$-free vertices in $G_i$.
-- Note that in particular, if $v$ is $\maxk_j$-free in $G$ then, $v$ is simplicial in $G_i$. -- 
Therefore, the claim is proved.
It remains now to show that $\tree[2]$ is indeed a $4$-Steiner root of $G$.
This is not the case only if there exist $x \in C_p, \ y \in C_q$ such that $p \neq q$ and $dist_{\tree[2]}(x,v) + dist_{\tree[2]}(v,y) \leq 4$.
%However, in this situation let $X_i,X_j \in {\cal K}(G)$ satisfy $x,v \in X_i$ and $y,v \in X_j$.
Our construction implies $dist_{\tree[2]}(x,v) \geq dist_{\tree}(x,v)$ and $dist_{\tree[2]}(y,v) \geq dist_{\tree}(y,v)$.
But then, we should already have $dist_{\tree}(x,y) \leq 4$ in the original root $\tree$.
Thus, since $\tree$ is a $4$-Steiner root of $G$, this case cannot happen and $\tree[2]$ is also a $4$-Steiner root of the graph $G$.
\end{proof}

\subsection{The general case}\label{sec:x-free}

We are now ready to prove Theorem~\ref{thm:x-free} in its full generality:

\begin{proofof}{Theorem~\ref{thm:x-free}}
Let $\tree$ be such that the result holds for maximal cliques (such a $\tree$ exists by Lemma~\ref{lem:almost-simplicial-placement}).
For any $\ci \in \CI{G} \setminus \MAXK{G}$ with at most two elements, the properties of the theorem always hold (for any $\tree$).
We so only consider the clique-intersections $\ci \in \CI{G} \setminus \MAXK{G}$ with at least three elements.
Then, $diam(\unlabeledsubtree{\ci}) \leq 3$ by Theorem~\ref{thm:clique-intersection} -- Prop~\ref{pty-ci:2} ({\it i.e.}, because $\ci$ is strictly contained into some maximal clique $\maxk$ and so, $diam(\unlabeledsubtree{\ci}) < diam(\unlabeledsubtree{\maxk})$). 
In what follows, we will often use properties of the subtree $\unlabeledsubtree{\ci}$ that only hold if $diam(\unlabeledsubtree{\ci}) \leq 3$; namely:
\begin{itemize}
\item $\unlabeledsubtree{\ci} \setminus \centre{\unlabeledsubtree{\ci}}$ is a collection of leaf-nodes;
\item there are at least two leaf-nodes ({\it i.e.}, because $|\ci| \geq 3$), and there is at least one leaf-node adjacent to every central node in \centre{\unlabeledsubtree{\ci}};
\item (as a direct consequence of the previous property) every central node in \centre{\unlabeledsubtree{\ci}} has at least two neighbours in \unlabeledsubtree{\ci}.
\end{itemize}
In particular, Property~\ref{pty-x-free:3} is now implied by Property~\ref{pty-x-free:1} as there can be no internal node on the path between a leaf and the closest central node.
In the same way, Property~\ref{pty-x-free:1} can be slightly simplified as {\em every} leaf has maximum eccentricity; hence, we only need to ensure that \ci-free vertices are leaves.
Finally, this also implies that $Steiner(\unlabeledsubtree{\ci}) \subseteq \centre{\unlabeledsubtree{\ci}}$, as all leaves of \unlabeledsubtree{\ci}~must be in \ci ({\it i.e.}, by the very definition of \unlabeledsubtree{\ci}).

The proof follows from different uses of a special operation on the tree $\tree$ that we now introduce:

\begin{operation}\label{op:move}
Let $\ci \in \CI{G} \setminus \MAXK{G}$ have size at least three and let $v \in \ci$.
We define $R_v$ to be the forest of all the subtrees in $\tree \setminus \unlabeledsubtree{\ci}$ that contain one node adjacent to $v$ and have {\em no} real node at a distance $\leq 4$ from $\ci \setminus \{v\}$ in $\tree$.
Let $Q_v$ be the subtree of $\tree$ that is induced by $R_v \cup \{v\}$.

We construct a new tree $\tree[2]$ from $\tree$ in two steps:
\begin{enumerate}
\item We remove $R_v$ and we replace $v$ by a Steiner node $\alpha_v$. In doing so, we obtain an intermediate tree denoted by $\labeledsubtree{v}{}$. Note that $\labeledsubtree{v}{}$ contains the subtree $\labeledsubtree{v}{}^{\ci}$: that is obtained from $\unlabeledsubtree{\ci}$ by replacing $v$ with $\alpha_v$ (and so, is isomorphic to $\unlabeledsubtree{\ci}$);
\item In order to obtain $\tree[2]$ from $\labeledsubtree{v}{}$, we add a copy of $Q_v$ and an edge $vc$ between $v$ and a center node of $\labeledsubtree{v}{}^{\ci}$
%$T_v\langle X \setminus v \rangle$ 
(possibly, $c = \alpha_v$).
\end{enumerate}
\end{operation}

We refer to Fig.~\ref{fig:s-free} and~\ref{fig:s-free-bis} for some particular applications of Operation~\ref{op:move}.
Furthermore in what follows we prove that under some conditions of use, this above Operation~\ref{op:move} always outputs a $4$-Steiner root $\tree[2]$ that is closer to satisfy all the properties of the theorem than $\tree$.
Specifically:

\begin{myclaim}\label{claim:pties-op}
Assume that $v$ is \ci-free and that every center node of $\unlabeledsubtree{\ci}$ is adjacent to a real node in $\ci \setminus \{v\}$.
Then, $\tree[2]$ keeps the property of being a $4$-Steiner root of $G$ if and only if the following conditions hold:
\begin{enumerate}[label=\textbullet,ref=\ref{claim:pties-op}.\alph{enumi}]
\item\label{cond-a}(Condition~\ref{cond-a}) either $dist_{\tree}(Real(R_v),v) \geq 4$ or $c$ is Steiner;
\item\label{cond-b}(Condition~\ref{cond-b}) if $c \neq \alpha_v$ then, $dist_{\labeledsubtree{v}{}}(c,V \setminus N_G[v]) > 3$.
\end{enumerate}
Moreover, for any $\ci' \in \CI{G} \setminus \{\ci\}$, if any of Properties~\ref{pty-x-free:1},~\ref{pty-x-free:2},~\ref{pty-x-free:3} or~\ref{pty-x-free:4} is satisfied for $\ci'$ in $\tree$ then, this stays so in $\tree'$.
\end{myclaim}

\begin{proofclaim}
First we prove that all the real vertices in $R_v$ are at a distance $> 4$ from $V(G) \setminus V(Q_v)$ in the original tree $\tree$ (Subclaim~\ref{subclaim:pties-op-a}).
This will help us in better understanding the structure of \tree~in the remaining of the proof.

\begin{subclaim}\label{subclaim:pties-op-a}
$dist_{\tree}(Real(R_v), V(G) \setminus V(Q_v)) > 4$.
\end{subclaim}

\begin{proofsubclaim}
Suppose for the sake of contradiction that there exist $x \in Real(R_v), y \notin V(Q_v)$ such that $dist_{\tree}(x,y) \leq 4$. 
Then, $v$ is onto the unique $xy$-path in $\tree$.
Furthermore, recall that by definition of $R_v$, $dist_{\tree}(x,\ci \setminus \{v\}) > 4$.
By the hypothesis every central node of $\centre{\unlabeledsubtree{\ci}}$ is adjacent to a vertex of $\ci \setminus v$, and so we have $dist_{\tree}(v,\ci \setminus \{v\}) \leq 2$.
As a result, we obtain $dist_{\tree}(v,x) = 3$ and $dist_{\tree}(v,y) = 1$.
But then, $y$ is at a distance $\leq diam(\unlabeledsubtree{\ci})+1 \leq 4$ from any vertex in $\ci$, that implies the existence of a clique-intersection $\ci_1 \supseteq \ci \cup \{y\}$.
Let $\ci_2$ be any clique-intersection that contains all of $x,y,v$.
Since we assume $dist_{\tree}(x,\ci \setminus \{v\}) > 4$, $\ci \cap \ci_2 = \{v\}$ and in particular, $\ci_1 \neq \ci_2$.
It implies that $v$ is $(\ci,\ci_1,\ci_2)$-sandwiched, a contradiction.
\end{proofsubclaim}

It follows from Subclaim~\ref{subclaim:pties-op-a} that in order for $\tree[2]$ to be a $4$-Steiner root for $G$, one must ensure $dist_{\tree[2]}(Real(R_v), V \setminus V(Q_v)) > 4$.
We then prove that this above condition is equivalent to Condition~\ref{cond-a}.
Indeed, note that we always have $dist_{\tree[2]}(Real(R_v), V \setminus V(Q_v)) > 4$ if $dist_{\tree}(v,Real(R_v)) \geq 4$.
Otherwise, by the hypothesis every center node of $\centre{\unlabeledsubtree{\ci}}$ is adjacent to a real node in $\ci \setminus \{v\}$, thereby implying $dist_{\tree}(v, V \setminus V(Q_v)) \leq dist_{\tree}(v, \ci \setminus \{v\}) \leq 2$, and so, $dist_{\tree}(v,Real(R_v)) = 3$.
Then, a necessary and sufficient condition for having that \\ $dist_{\tree[2]}(Real(R_v), V \setminus V(Q_v)) > 4$ is that $c$ is Steiner.

In the same way, Condition~\ref{cond-b} implies the necessary condition $dist_{\tree[2]}(u,v) > 4$ for every $u \notin N_G(v)$.
However, the above does not prove that $\tree[2]$ is a $4$-Steiner root of $G$ yet, as we also need to ensure $dist_{\tree[2]}(u,v) \leq 4$ for every $u \in N_G(v)$.
In order to prove this is the case, and to also prove the second part of the claim, we now consider the clique-intersections $\ci' \in \CI{G} \setminus \{\ci\}$ such that $v \in \ci'$.
(Note that if $v \notin \ci'$ then, $\unlabeledsubtree{\ci'} = \labeledsubtree[2]{}{\ci'}$ and so, the result of our claim trivially holds for such a $\ci'$).
Since we have $dist_{\tree}(Real(R_v), V \setminus V(Q_v)) > 4$ (Subclaim~\ref{subclaim:pties-op-a}), there are only two possibilities: either $\unlabeledsubtree{\ci'}$ is fully contained into $Q_v$ -- in which case it is not modified --; or it does not intersect $R_v$. %and so, it must intersect $\unlabeledsubtree{\ci} \setminus \{v\}$.
We then consider two different cases:

\begin{itemize}
\item Assume $\ci \subset \ci'$. 
In particular, $\unlabeledsubtree{\ci'} \cap R_v = \emptyset$.
It implies that $\labeledsubtree[2]{}{\ci'}$ is obtained from $\unlabeledsubtree{\ci'}$ by replacing $v$ by a Steiner node (only if it were an internal node of $\unlabeledsubtree{\ci'}$) then, making of $v$ a leaf.

\begin{subclaim}\label{subclaim:pties-op-b}
$diam(\labeledsubtree[2]{}{\ci'}) = diam(\unlabeledsubtree{\ci'})$.
\end{subclaim}

\begin{proofsubclaim}
%Let us assume $\unlabeledsubtree{\ci} \neq \labeledsubtree[2]{}{\ci}$ (otherwise since $\unlabeledsubtree{\ci'} \cap R_v = \emptyset$ then, $\unlabeledsubtree{\ci'} = \labeledsubtree[2]{}{\ci'}$, and so we are done).
We divide the proof into two parts:
\begin{itemize}
\item In order to prove $diam(\labeledsubtree[2]{}{\ci'}) \leq diam(\unlabeledsubtree{\ci'})$, we consider an auxiliary subtree $\labeledsubtree[3]{}{\ci'}$: obtained from the original \unlabeledsubtree{\ci'}~by adding a leaf $v'$ to some arbitrary central node $c'$ in \centre{\unlabeledsubtree{\ci}}.
Note that \labeledsubtree[2]{}{\ci'}~is isomorphic to a subtree of $\labeledsubtree[3]{}{\ci'}$ for the choice of $c' = c$.
Furthermore since $c'$ has at least two neighbours in $\unlabeledsubtree{\ci} \subseteq \unlabeledsubtree{\ci'}$ we have: 
$$diam\left(\unlabeledsubtree{\ci'} \cap N_{\labeledsubtree[3]{}{\ci'}}[c'] \right) = diam(N_{\unlabeledsubtree{\ci'}}[c']) = 2 = diam(N_{\labeledsubtree[3]{}{\ci'}}[c']).$$
We so deduce from Lemma~\ref{lem:tree-intersection} that: $$diam(\labeledsubtree[3]{}{\ci'}) = diam(\unlabeledsubtree{\ci'} \cup N_{\labeledsubtree[3]{}{\ci'}}[c']) = diam(\unlabeledsubtree{\ci'}).$$
In particular, $diam(\labeledsubtree[2]{}{\ci'}) \leq diam(\labeledsubtree[3]{}{\ci'}) = diam(\unlabeledsubtree{\ci'})$. 

\item For the converse direction, it suffices to prove, if $\unlabeledsubtree{\ci'} \neq \labeledsubtree[2]{}{\ci'}$, the existence of a diametral path in $\unlabeledsubtree{\ci'}$ of which $v$ is not an end.
Since the ends of a diametral path must be leaves, we can restrict our study to the case where: $v$ is a leaf of $\unlabeledsubtree{\ci}$, and there is no subtree of $\unlabeledsubtree{\ci'} \setminus \unlabeledsubtree{\ci}$ that contains a neighbour of $v$ in $\tree$.
In particular, $\unlabeledsubtree{\ci} \neq \labeledsubtree[2]{}{\ci}$ (otherwise $\unlabeledsubtree{\ci'} = \labeledsubtree[2]{}{\ci'}$, and so we are done).
We deduce from the above that $v$ is adjacent to a central node $c' \in \centre{\unlabeledsubtree{\ci}}$, and $c' \neq c$ (otherwise $\unlabeledsubtree{\ci} = \labeledsubtree[2]{}{\ci}$, a contradiction).
Then, $\centre{\unlabeledsubtree{\ci}} = \{c,c'\}$ and so, $\unlabeledsubtree{\ci}$ is a bistar (diameter-three subtree).
Recall that both $c,c'$ have a neighbour in $\ci \setminus \{v\}$ by the hypothesis.
This implies that node $c'$ must have at least two neighbours in $\labeledsubtree{v}{}^{\ci} \setminus \{\alpha_v\}$; moreover each such neighbour must be on the path between $c'$ and a real node in $\ci \setminus \{v\}$.
But then, since there are at least three branches in $\unlabeledsubtree{\ci'} \setminus \{c'\}$, we can obtain as follows the existence of a diametral path in $\unlabeledsubtree{\ci'}$ of which $v$ is not an end: let us fix a longest path in $\unlabeledsubtree{\ci'}$ of which $v$ is an end; then, such a path goes by $c'$ and we can simply replace $v$ by the path between $c'$ and any real node in one of the two other branches of $\unlabeledsubtree{\ci'} \setminus \{c'\}$ by which this path does not go.
As a result, $diam(\labeledsubtree[2]{}{\ci'}) \geq diam(\unlabeledsubtree{\ci'})$.  
\end{itemize}
Altogether combined, the above proves the subclaim that $diam(\labeledsubtree[2]{}{\ci'}) = diam(\unlabeledsubtree{\ci'})$.
\end{proofsubclaim}
By Subclaim~\ref{subclaim:pties-op-b}, we cannot change $diam(\unlabeledsubtree{\ci'})$ and so, it immediately implies that we cannot break Property~\ref{pty-x-free:4}.
The proof of this subclaim is actually more precise as it also shows that if $\labeledsubtree[2]{}{\ci'} \neq \unlabeledsubtree{\ci'}$ then, these two subtrees have a common diametral pair $(x,y)$, for some $x,y \neq v$.
In particular by considering the middle nodes on the unique $xy$-path in $T$ and $T'$, respectively, it follows from Prop.~\ref{pty-lem:tree:2} of Lemma~\ref{lem:tree} that we also have: $$\centre{\labeledsubtree[2]{}{\ci'}} = \begin{cases} 
\centre{T\langle \ci'\rangle} \ \text{if} \ v \notin \centre{T\langle \ci'\rangle} \\ 
\left( \centre{T\langle \ci'\rangle} \setminus \{v\} \right) \cup \{\alpha_v\} \ \text{otherwise.} \end{cases}$$
Furthermore, for any $u \in \ci' \setminus \{v\}$, if $P_u$ denotes the $u\centre{\unlabeledsubtree{\ci'}}$-path in $\tree$ then, the $u\centre{\labeledsubtree[2]{}{\ci'}}$-path in $\tree[2]$ can be obtained from $P_u$ by replacing $v$ with $\alpha_v$.
-- This above analysis excludes $v$; however, $v$ cannot be $\ci'$-free because $v \in \ci \subset \ci'$ and $|\ci| \geq 3$. -- Hence, this above mapping implies that we cannot break Properties~\ref{pty-x-free:2} and~\ref{pty-x-free:3}.
It also implies $dist_{\tree[2]}(u,\centre{\labeledsubtree[2]{}{\ci'}}) = dist_{\tree}(u,\centre{\unlabeledsubtree{\ci'}})$ and so, by unimodality (Prop.~\ref{pty-lem:tree:1} of Lemma~\ref{lem:tree}) we did not change the eccentricity of any vertex $u \in \ci' \setminus \{v\}$.
Altogether combined, every $\ci'$-free vertex that was a leaf of $\unlabeledsubtree{\ci'}$ is also a leaf of \labeledsubtree[2]{}{\ci'} with same eccentricity (Property~\ref{pty-x-free:1} cannot be broken).

Finally, we also obtain $dist_{\tree[2]}(u,v) \leq 4$ for every $u \in \ci'$ ({\it i.e.}, $\tree[2]$ is a $4$-Steiner root of $G$).
%Furthermore, as neither $c$ nor $v$ can be $\ci'$-free, we cannot break 
%The above implies that any $\ci'$-free vertex that was a leaf in $\unlabeledsubtree{\ci'}$ is also a leaf of $\labeledsubtree[2]{}{\ci'}$ with same eccentricity (Property~\ref{pty-x-free:1}).
%Finally, the above transformation cannot add an internal real node onto the path between such a leaf and the center nodes, that implies we cannot break Property~\ref{pty-x-free:3} of the theorem.

\item Otherwise, $\ci \not\subseteq \ci'$ and we prove $\unlabeledsubtree{\ci'} = \labeledsubtree[2]{}{\ci'}$.
To see that, suppose for the sake of contradiction $\unlabeledsubtree{\ci'} \neq \labeledsubtree[2]{}{\ci'}$.
We first note this may be the case only if $\unlabeledsubtree{\ci'}$ is not fully contained into $Q_v$.
\begin{subclaim}\label{subclaim:deux-ex-machina}
$\unlabeledsubtree{\ci'}$ must intersect $\unlabeledsubtree{\ci} \setminus \{v\}$.
\end{subclaim}
\begin{proofclaim}
Suppose by contradiction that $\unlabeledsubtree{\ci'} \cap \unlabeledsubtree{\ci} = \{v\}$.
As we also assume $\unlabeledsubtree{\ci'}$ is not fully contained into $Q_v$, we have $\unlabeledsubtree{\ci'} \cap R_v = \emptyset$.
Then, $\unlabeledsubtree{\ci'}$ must intersect some subtree $\tree[sub]$ of $\tree \setminus \unlabeledsubtree{\ci}$ that is adjacent to $v$ but {\em not} in $R_v$.
Let $x \in \ci' \cap \tree[sub]$.
By the definition of $R_v$, there must exist $x' \in Real(\tree[sub])$ such that: $dist_{\tree}(x', \ci \setminus \{v\}) \leq 4$ (possibly, $x=x'$).
%In particular, note that there exists a $xx'$-path that does not go by $v$ in $\tree$ (otherwise, $x,x'$ would be in different subtrees of $\tree \setminus \unlabeledsubtree{\ci}$).
On one hand as we assume $v$ to be $\ci$-free, there must exist a clique-intersection $\ci_1 \supseteq \ci \cup \{x'\}$.
It implies $dist_{\tree}(x',v) \leq 2$ because $ecc_{\unlabeledsubtree{\ci}}(v) \geq 2$.
On the other hand, we also know that $dist_{\tree}(x,v) \leq 4$. 
Furthermore note that there exists a $xx'$-path that does not go by $v$ in $\tree$ (otherwise, $x,x'$ would be in different subtrees of $\tree \setminus \unlabeledsubtree{\ci}$).
By considering the median node $r$ of the triple $x,x',v$ we so obtain: $dist_{\tree}(x',r) \leq 1$, $dist_{\tree}(x,r) \leq 3$, and so $dist_{\tree}(x,x') \leq 4$.
Let $\ci_2$ be a clique-intersection that contains all of $x,x',v$ (possibly, $\ci_2 = \ci'$).
There are two cases:
\begin{itemize}
\item If $\ci \not\subseteq \ci_2$ then, $\ci \cap \ci_2 = \{v\}$ (otherwise, $v$ would be internally $\ci$-constrained). But in this case, $v$ is $(\ci,\ci_1,\ci_2)$-sandwiched, a contradiction.
\item Otherwise, $\ci \subseteq \ci_2$. However by the hypothesis $\ci \not\subseteq \ci'$, and so $\ci \cap \ci' = \{v\}$ (otherwise, $v$ would be internally $\ci$-constrained). It implies $v$ is $(\ci,\ci_2,\ci')$-sandwiched, a contradiction. 
\end{itemize}
As a result, $\unlabeledsubtree{\ci'}$ must intersect $\unlabeledsubtree{\ci} \setminus \{v\}$.
\end{proofclaim}
Finally, since $v$ is $\ci$-free any node $\beta \in V(\unlabeledsubtree{\ci'}) \cap \left( V(\unlabeledsubtree{\ci}) \setminus \{v\}\right)$ must be Steiner.
This leaves $\beta \in \centre{\unlabeledsubtree{\ci'}} \setminus \{v\}$ as the only possibility. 
Furthermore, since $\beta$ is Steiner there must exist $y \in \ci'$ such that the unique $vy$-path in $\tree$ goes by $\beta$.
However, this implies by Lemma~\ref{lem:tree-intersection} $diam(\unlabeledsubtree{\ci' \cup N_{\tree}[\beta]}) = diam(\unlabeledsubtree{\ci'})$. 
We recall that there exists at least one leaf node $u \in Real(N_{\tree}(\beta)) \setminus \{v\}$ by the hypothesis.
Thus, by Property~\ref{pty-ci:2} of Theorem~\ref{thm:clique-intersection} we have $u,v \in \ci \cap \ci'$, thereby contradicting that $v$ is $\ci$-free.
\end{itemize}
The claim directly follows from this above case analysis.
\end{proofclaim}

The proof is now divided into two main phases.

\begin{figure}[h!]
\centering
\includegraphics[width=.55\textwidth]{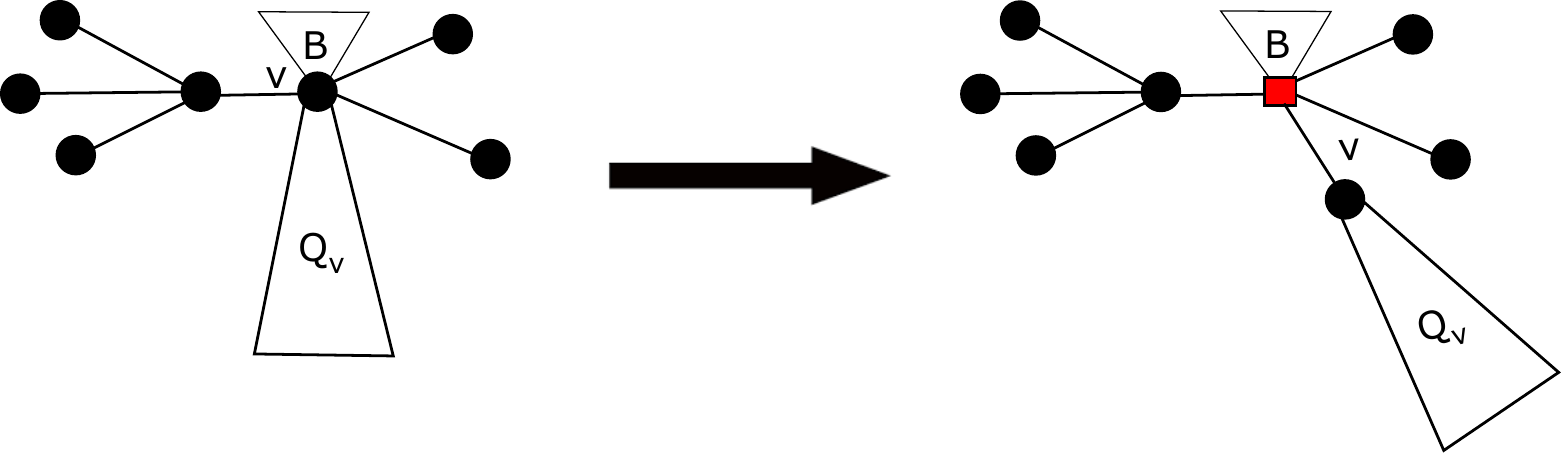}
\caption{Forcing the $\ci$-free vertices as leaves.}
\label{fig:s-free}
\end{figure}

\paragraph{Phase 1: Transformation into leaves} (see Fig.~\ref{fig:s-free} for an illustration).
Let $\ci \in \CI{G} \setminus \MAXK{G}, \ |\ci| \geq 3$ be fixed.
We first transform $\tree$ so that all the $\ci$-free vertices are leaves in $\unlabeledsubtree{\ci}$.
Assume the existence of a $\ci$-free vertex $v \in \ci$ that is not a leaf.
Note that we have $v \in \centre{\unlabeledsubtree{\ci}}$ because $diam(\unlabeledsubtree{\ci}) \leq 3$.
In particular, every node in $\centre{\unlabeledsubtree{\ci}}$ is adjacent to a leaf in $\ci \setminus \{v\}$.
We apply Operation~\ref{op:move} with $c=\alpha_v$ ({\it i.e.}, the Steiner node substituting $v$ in the intermediate tree $\labeledsubtree{v}{}$).
Since $c=\alpha_v$ and so, in particular $c$ is Steiner, both Conditions~\ref{cond-a} and~\ref{cond-b} of Claim~\ref{claim:pties-op} must hold.
Overall, by Claim~\ref{claim:pties-op} we can repeat the above transformation until all the $\ci$-free vertices are leaves of $\unlabeledsubtree{\ci}$. 

\begin{figure}[h!]
\centering
\includegraphics[width=.55\textwidth]{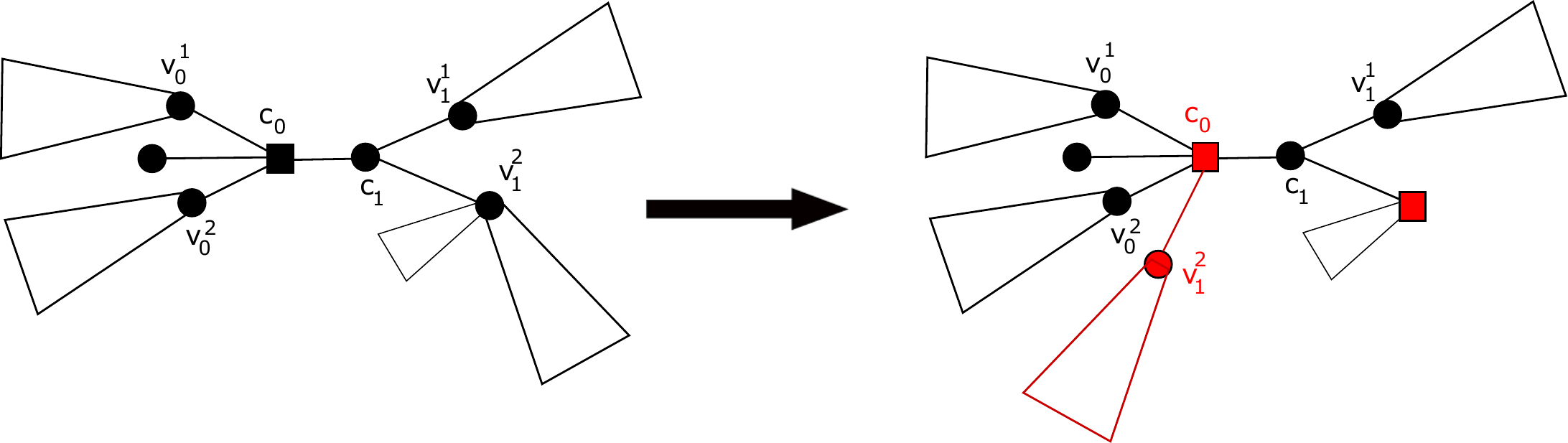}
\caption{Grouping the $\ci$-free vertices on a same side.}
\label{fig:s-free-bis}
\end{figure}

\paragraph{Phase 2: Grouping the $\ci$-free vertices} (see Fig.~\ref{fig:s-free-bis} for an illustration).
After Phase 1 is completed, the properties of the theorem are true for any $\ci \in \CI{G}$ such that $\unlabeledsubtree{\ci}$ is a star.
Thus, from now on assume $\unlabeledsubtree{\ci}$ is a bistar (diameter-three subtree).
Write $\centre{\unlabeledsubtree{\ci}} = \{c_0,c_1\}$ and assume that each $c_j$ is adjacent to two $\ci$-free vertices, denoted $v_j^1,v_j^2$.
\begin{itemize}
\item
For any $i \in \{1,2\}$, there is no vertex $x_i \in Real(R_{v_j^i})$ such that $dist_{\tree}(x_i,v_j^i) \leq 2$ %(otherwise, there would exist a maximal clique containing $x_i,v_j^1,v_j^2$, and so, the two of $v_j^1,v_j^2$ would be $\ci$-constrained).
(otherwise, $dist_{\tree}(x_i, \ci \setminus \{v_j^i\}) \leq dist_{\tree}(x_i,v_j^{3-i}) \leq 4$, a contradiction).
More specifically, either $dist_{\tree}(v_j^i, Real(R_{v_j^i})) \geq 4$, or $dist_{\tree}(v_j^i, Real(R_{v_j^i})) = 3$ but then $c_j$ must be Steiner (otherwise, $dist_{\tree}(x_i, \ci \setminus \{v_j^i\}) \leq dist_{\tree}(x_i,c_j) \leq 4$, a contradiction).
%(otherwise, we can prove as above that $v_j^i$ should be $\ci$-constrained).
\item
Furthermore, we claim that we can have $dist_{\tree}(u,c_j) \leq 3$ only if $\ci \subseteq N_G[u]$. 
Indeed, suppose for the sake of contradiction $dist_{\tree}(u,c_j) \leq 3$ and $\ci \not\subseteq N_G[u]$. 
Let $\maxk_j$ be a maximal clique which contains $\{u\} \cup (\ci \cap N_G(u))$.
Then, $\ci' = \maxk_j \cap \ci = \ci \cap N_G(u)$. 
Moreover, $\ci'$ contains $v_j^1,v_j^2$, thereby contradicting our assumption that these two vertices are $\ci$-free. 
\end{itemize}

W.l.o.g., assume either $dist_{\tree}(v_j^i, Real(R_{v_j^i})) \geq 4$ for any $i,j$ or $c_0$ is Steiner.
If in addition both $c_0,c_1$ are Steiner nodes (real nodes, resp.) then, we further assume w.l.o.g. $c_0$ is adjacent to more $\ci$-free vertices than $c_1$.
We apply Operation~\ref{op:move} for $v = v_1^1$ and $c = c_0$.
%\footnote{As suggested by Fig.~\ref{fig:s-free-bis}, we could actually apply the transformation to $v_1^1$ and $v_1^2$ simultaneously. However, we do not need this refinement for the proof.}.
Note that by construction, $c_0$ always satisfies Condition~\ref{cond-a} of Claim~\ref{claim:pties-op}.
Since we also have $dist_{\tree}(u,c_0) \leq 3 \Longrightarrow \ci \subseteq N_G[u]$, Condition~\ref{cond-b} of Claim~\ref{claim:pties-op} is satisfied.
Overall, we can repeat this transformation until $\ci$ satisfies all the properties stated in the theorem, that does not impact the properties of the other clique-intersections $\ci'$ by Claim~\ref{claim:pties-op}.
\end{proofof}

\section{Step~\ref{step-1}: Construction of the rooted clique-tree}\label{sec:clique-tree}

Our main result in this section is that every chordal graph has a polynomial-time computable rooted clique-tree where some technical conditions must hold on the minimal separators that can label an edge (Theorem~\ref{thm:final-clique-tree}). 
%Moreover, such a rooted clique-tree can be constructed in polynomial time. 
We will use these special properties of our rooted clique-tree in order to ensure that our main algorithm runs in polynomial time.

\smallskip
We start digressing on the motivations behind our construction.
Given a rooted clique-tree \cliquetree{G}~of $G$ and an arbitrary maximal clique $\maxk_i$, recall that $G_i$ is the subgraph induced by all the maximal cliques in the subtree $\cliquetree{G}^i$ rooted at $\maxk_i$.
Assume that $\unlabeledsubtree{\sep_i}$: the subtree induced by the minimal separator $\sep_i := \maxk_i \cap \maxk_{p(i)}$, is fixed.
We ask for the number of possible distance profiles $\left( dist_{\labeledsubtree{i}{}}(r, V_i \setminus \sep_i) \right)_{r \in V(\unlabeledsubtree{\sep_i})}$ over all possible {\em well-structured} $4$-Steiner roots of $G_i$ that contain $\unlabeledsubtree{\sep_i}$.
By Theorem~\ref{thm:x-free}, one way to bound this number would be to force most vertices in $\sep_i$ to be {\em $\maxk_i$-free} in the subgraph $G_i$.
Guided by this intuition, our construction aims at preventing $\sep_i$ to contain, or to be contained, in a minimal separator of $G_i$.
Since both objectives are conflicting, this results in a technical compromise. 

\medskip
We present a first, simpler construction in Section~\ref{sec:flat-clique-tree} that only reaches half of the goal.
Then, we introduce the new notion of (weak) convergence, and we show its relationship with $4$-Steiner powers (Section~\ref{sec:weak-convergence}).
We end up proving the main result of this part in Section~\ref{sec:final-clique-tree}.

%\paragraph{Additional terminology.}
%The following additional terminology will be used throughout all this section.
%A rooted clique-tree of $G$ is obtained from any clique-tree $T_G$ by identifying an arbitrary maximal clique $X_0$ as its root.
%Let $(X_q,X_{q-1},\ldots,X_1,X_0)$ be a postordering of $T_G$ (obtained by depth-first search).
%For any $i > 0$, we define $X_{p(i)}$ as the father node of $X_i$.
%The common intersection of $X_i$ with its father node is the minimal separator $S_i := X_i \cap X_{p(i)}$.
%Finally, let $G_i$ be the subgraph induced by all the maximal cliques in the subtree $T_G^i$ rooted at $X_i$ (in particular, we have $T_G^0 = T_G$ and $G_0 = G$).

\subsection{A flat clique-tree}\label{sec:flat-clique-tree}

We start with an intermediate construction.

\begin{theorem}\label{thm:flat-clique-tree}
Given $G=(V,E)$ chordal, we can compute in polynomial time a rooted clique-tree $\cliquetree{G}$ such that, for any $\sep_i = \maxk_i \cap \maxk_{p(i)}$ and for any child $\maxk_j$ of $\maxk_{p(i)}$, there is no minimal separator of $G_j$ that is contained into $\sep_i$.
\end{theorem}

\begin{proof}
We modify an arbitrary clique-tree $\cliquetree{G}$ of $G$ until the property of the theorem is satisfied.
Specifically, root $\cliquetree{G}$ at some arbitrary maximal clique $\maxk_0$.
We consider all the minimal separators $\sep \in \SEP{G}$ by decreasing size.
Let $\maxk_{\sep}$ be incident to an edge in $\superedgeset{G}{\sep}$ and the closest possible to the root.
We observe that $\maxk_{\sep}$ is the least common ancestor of all maximal cliques that are incident to an edge in $\superedgeset{G}{\sep}$.
Furthermore all edges in $\edgeset{G}{\sep}$ can be made incident to $\maxk_{\sep}$, as follows.
Assume there exists $\maxk_i\maxk_{p(i)} \in \edgeset{G}{\sep}$ such that $\maxk_{\sep} \notin \{\maxk_i,\maxk_{p(i)}\}$.
By the above observation, $\maxk_i,\maxk_{p(i)}$ are into the subtree rooted at $\maxk_{\sep}$.
%W.l.o.g., $Z$ is further than $Y$ from $X_S$.
Since $\sep \subseteq \maxk_{\sep} \cap \maxk_i$, $\sep$ is contained into all the maximal cliques on the $\maxk_{\sep}\maxk_i$-path.
In particular, we still obtain a clique-tree of $G$ if we replace $\maxk_i\maxk_{p(i)}$ by $\maxk_i\maxk_{\sep}$ and in doing so, $\sep \subseteq \maxk_{\sep} \cap \maxk_i \subseteq \maxk_{p(i)} \cap \maxk_i = \sep$.
Furthermore after this transformation, $\maxk_{\sep}$ became the new father node of the maximal clique $\maxk_i$ in $\cliquetree{G}$.

%\smallskip
It now remains to prove that the gotten clique-tree $\cliquetree{G}$ satisfies the conditions of the theorem.
Suppose for the sake of contradiction that there exist $i > 0 $ and $\sep_k \subseteq \sep_i$ a minimal separator of $G_j$, where $\maxk_j$ is a child of $\maxk_{p(i)}$ (possibly, $i=j$).
We observe that after we processed $\sep_i$, all edges in $\edgeset{G}{\sep_i}$ must label an edge between $\maxk_{\sep_i} = \maxk_{p(i)}$ and its children nodes.
Therefore, our transformation ensures $\sep_k \neq \sep_i$, {\it i.e.}, $\sep_k \subset \sep_i$.
Since the subtree rooted at $\maxk_j$ is a rooted clique-tree of $G_j$, there must exist some edge $\maxk_k\maxk_{p(k)}$ in this subtree such that $\maxk_k \cap \maxk_{p(k)} = \sep_k$.
However, since we consider minimal separators by increasing size, the edge $\maxk_i\maxk_{p(i)}$ should already exist when we process $\sep_k$.
It implies that the maximal clique $\maxk_{\sep_k}$ to which we connected all edges in $\edgeset{G}{\sep_k}$ should be an ancestor of $\maxk_{p(i)}$, that is a contradiction.
\end{proof}

\subsection{Weak convergence}\label{sec:weak-convergence}

Unfortunately, the ``flat'' clique-tree of Theorem~\ref{thm:flat-clique-tree} does not prevent the case when a minimal separator $\sep_i$ {\em is contained} into a minimal separator of $G_i$.
%Unfortunately, this ``flat'' clique-tree that we get is not exactly what we need yet, due to some encoding issues.
%Specifically, consider the particular case in our dynamic programming algorithm where a minimal separator $S_i$ induces a star in some partial solutions we found for the subgraph $G_i$.
%In order to bound the number of partial solutions that we should store where $\tree{S_i}$ is a star, we would like most vertices in $S_i$ to be simplicial in $G_i$.
%However, this cannot be the case if we encountered a larger minimal separator $S' \supset S_i$ in the corresponding clique-subtree.
%See Section~\ref{sec:encoding} for more details.
As a new step toward our final construction, we now introduce the following notions:

%\begin{definition}\label{def:max-sep}
%Given $G=(V,E)$, a minimal separator $S$ is called {\em maximal} (or a maximal separator) if and only if it is inclusion wise maximal in ${\cal S}(G)$.
%\end{definition}

\begin{definition}\label{def:cvgt}
Given a clique-tree $\cliquetree{G}$ of $G = (V,E)$, we say that a minimal separator $\sep$ is {\em weakly $\cliquetree{G}$-convergent} if there exists some maximal clique $\maxk_{\sep}$ that is incident to all edges in $\superstrictedgeset{G}{\sep}$.
%\begin{itemize}
%\item Any minimal separator $S'$ strictly containing $S$ is maximal and contained into exactly two maximal cliques;
%\item Moreover, there exists some maximal clique $X_S$ that is incident to all edges in $\bigcup\limits_{S \subset S'} E_{S'}(T_G)$.
%\end{itemize} 
%
$\sep$ is termed {\em $\cliquetree{G}$-convergent} if it is weakly $\cliquetree{G}$-convergent and the maximal clique $\maxk_{\sep}$ is also incident to all edges in $\edgeset{G}{\sep}$.
\end{definition}

In order to motivate Definition~\ref{def:cvgt}, in what follows are two observations on the relationships between clique-trees, minimal separators and $4$-Steiner roots:

\begin{lemma}\label{lem:bistar-center}
Given $G=(V,E)$ and $\tree$ any $4$-Steiner root of $G$, let $\maxk_i \in \MAXK{G}$ and let $\ci \subset \maxk_i$ be a clique-intersection.
If $\unlabeledsubtree{\ci}$ is a bistar then, $\centre{\unlabeledsubtree{\maxk_i}} \subset \centre{\unlabeledsubtree{\ci}}$.

In particular, there are exactly two maximal cliques that contain $\ci$.
\end{lemma}

\begin{proof}
We have by Theorem~\ref{thm:clique-intersection} $diam(\unlabeledsubtree{\maxk_i}) > diam(\unlabeledsubtree{\ci})$, and so, $diam(\unlabeledsubtree{\maxk_i}) = 4$.
In particular, write $\centre{\unlabeledsubtree{\maxk_i}} = \{c_i\}$.
Every component in $\unlabeledsubtree{\maxk_i} \setminus \{c_i\}$ has diameter at most two, thereby implying $c_i \in V(\unlabeledsubtree{\ci})$.
Furthermore since $ecc_{\unlabeledsubtree{\ci}}(c_i) \leq rad(\unlabeledsubtree{\maxk_i}) = 2$, $c_i$ cannot be a leaf of $\unlabeledsubtree{\ci}$.
Equivalently, $c_i \in \centre{\unlabeledsubtree{\ci}}$.
By Lemma~\ref{lem:no-center-intersect}, there can be no two maximal cliques $\maxk_i,\maxk_j \in \MAXK{G}$ such that $\centre{\unlabeledsubtree{\maxk_i}} = \centre{\unlabeledsubtree{\maxk_j}}$.
Therefore, the above implies that $\ci$ can only be contained in at most two maximal cliques.
Finally, since $\ci$ is not a maximal clique, it is contained into exactly two maximal cliques.
\end{proof}

\begin{lemma}\label{lem:star-center}
Given $G=(V,E)$ and $\tree$ any $4$-Steiner root of $G$, let $\sep \in \SEP{G}$.
If $\unlabeledsubtree{\sep}$ is a non-edge star then, $\sep$ is weakly $\cliquetree{G}$-convergent for {\em any} clique-tree $\cliquetree{G}$ of $G$.
\end{lemma}

\begin{proof}
We may assume that $\sep$ is strictly contained into at least one minimal separator $\sep'$ for otherwise there is nothing to prove.
By Theorem~\ref{thm:clique-intersection}, $\unlabeledsubtree{\sep'}$ is a bistar and $\sep'$ must be inclusion wise maximal in $\SEP{G}$. 
This implies $\centre{\unlabeledsubtree{\sep}} \subset \centre{\unlabeledsubtree{\sep'}}$.
Furthermore, it follows from Lemma~\ref{lem:bistar-center} that $\sep'$ must be contained into exactly two maximal cliques $\maxk_i,\maxk_j$ and $\centre{\unlabeledsubtree{\maxk_i}} \cup \centre{\unlabeledsubtree{\maxk_j}} = \centre{\unlabeledsubtree{\sep'}}$.
In particular, we may assume w.l.o.g. that $\centre{\unlabeledsubtree{\sep}} = \centre{\unlabeledsubtree{\maxk_i}}$.
But then, still by Lemma~\ref{lem:bistar-center}, any minimal separator $\sep''$ that strictly contains $\sep$ must be contained into $\maxk_i$ and exactly one other maximal clique $\maxk_{\sep''}$.
Let $\cliquetree{G}$ be an arbitrary clique-tree of $G$.
By Theorem~\ref{thm:clique-tree-pties}, the above implies $\maxk_i\maxk_{\sep''} \in E(\cliquetree{G})$ and $\maxk_i \cap \maxk_{\sep''} = \sep''$. 
In particular, one obtains by setting $\maxk_{\sep} := \maxk_i$ that $\sep$ is weakly $\cliquetree{G}$-convergent.
\end{proof}

Roughly, in order to prove Theorem~\ref{thm:final-clique-tree} we will slightly modify the construction of Theorem~\ref{thm:flat-clique-tree} so as to force weak convergence to imply convergence.
In doing so we will obtain that for a fixed minimal separator $\sep$, we may encounter some ``inclusion issue'' between \sep~and another minimal separator \sep'~at most {\em once}.
We will show in Sec.~\ref{sec:encoding} that this above ``local'' property of our clique-tree is enough in order to bound the number of possible distance profiles at each node of the clique-tree by a polynomial.

\subsection{The final construction}\label{sec:final-clique-tree}

The remaining of this section is now devoted to prove the following technical result:

\begin{theorem}\label{thm:final-clique-tree}
Given $G=(V,E)$ chordal, we can compute in polynomial time a rooted clique-tree $\cliquetree{G}$ where the following conditions are true for any $\sep_i := \maxk_i \cap \maxk_{p(i)}, i > 0$:
\begin{enumerate}[label=\textbullet,ref=P-\ref{thm:final-clique-tree}.\theenumi]
\item\label{pty-fct:1}({\small Prop.~\ref{pty-fct:1}.}) If $\sep_i$ is weakly $\cliquetree{G}$-convergent and $|\sep_i| \geq 3$ then, $\sep_i$ is $\cliquetree{G}$-convergent;
%\item If $S_i$ contains a minimal separator of $G_i$ then, $|S_i| \geq 3$ and $S_i$ is $T_G$-convergent;
\item\label{pty-fct:2}({\small Prop.~\ref{pty-fct:2}.}) Any minimal separator of $G_i$ that is contained into $\sep_i$ is $\cliquetree{G}$-convergent, it has at least three vertices and it is strictly contained into a minimal separator of $G_i$.
\end{enumerate}
\end{theorem}

%IMPORTANT: we cannot change the status of an edge if we did not change the subtree of its ancestor. To see that, note that it could happen only if the convergence changes. But then, we would find some edge adjacent to the root of the subtree strictly containing the separator, a contradiction.

We stress that in the above two statements, $3$ can be replaced by any positive integer $q$.
However, please note that such a change would affect both properties of the theorem.

\begin{proof}
We start describing the algorithm before proving its correctness.
In what follows, let $\SEP{G} = ( \sep_1, \sep_2, \ldots, \sep_{\ell} )$ be totally ordered in such a way that $|\sep_i| < |\sep_j| \Longrightarrow i < j$. 
For every phase of the algorithm, we consider all the minimal separators of $G$ by decreasing order, {\it i.e.}, from $\sep_{\ell}$ to $\sep_1$.
%The algorithm proceeds in three phases.
\begin{enumerate}[label={\bf Phase~\theenumi}]
\item Let $\cliquetree{G}$ be an arbitrary {\em unrooted} clique-tree.
We consider all the minimal separators $\sep \in \SEP{G}$ by decreasing order.
Let $\sep^1, \sep^2, \ldots, \sep^p$ be the list of all minimal separators containing $\sep$ by decreasing order.
-- In particular, $\sep^p = \sep$, and $\sep^1$ is a largest minimal separator containing $\sep$. --
For every maximal clique $\maxk$, we can define a binary vector $\overrightarrow{v_{\sep}}(\maxk) = (\delta_{\sep^1}^{\maxk},\delta_{\sep^2}^{\maxk},\ldots,\delta_{\sep^p}^{\maxk})$: where $\delta_{\sep^i}^{\maxk} = 1$ if and only if $\maxk$ is incident to an edge in $\edgeset{G}{\sep^i}$.
Then, we choose $\maxk_{\sep}^1 \in \MAXK{G}$ such that $\overrightarrow{v_{\sep}}(\maxk_{\sep}^1)$ is lexicographically maximal.
While there exists an edge $\maxk\maxk' \in \edgeset{G}{\sep}$ such that $\maxk_{\sep}^1 \notin \{\maxk,\maxk'\}$, we do as follows.
W.l.o.g., $\maxk$ is on the $\maxk_{\sep}^1\maxk'$-path.
We replace the edge $\maxk\maxk'$ by $\maxk_{\sep}^1\maxk'$.
Doing so, all edges in $\edgeset{G}{\sep}$ are now incident to $\maxk_{\sep}^1$.

\item We root $\cliquetree{G}$ at some arbitrary maximal clique $\maxk_0$.
Then, we consider all the minimal separators $\sep \in \SEP{G}$ by decreasing order.
Let $\maxk_{\sep}^2$ be, under the following conditions, the closest possible to the root:
\begin{itemize}
\item $\maxk_{\sep}^2$ is incident to an edge in $\superedgeset{G}{\sep}$; 
\item and if $|\sep| \geq 3$ and $\sep$ is $\cliquetree{G}$-convergent then, $\maxk_{\sep}^2$ is incident to all edges in $\superstrictedgeset{G}{\sep}$.
\end{itemize}
Note that $\maxk_{\sep}^2$ is always an ancestor of $\maxk_{\sep}^1$.
In particular if $|\sep| \geq 3$ and $\sep$ is $\cliquetree{G}$-convergent then, $\maxk_{\sep}^2$ is either $\maxk_{\sep}^1$ or its parent node (with the latter being possible only if $|\superstrictedgeset{G}{\sep}|\leq 1$).
Then, we use the same operation as in the proof of Theorem~\ref{thm:flat-clique-tree} in order to make all edges in $\edgeset{G}{\sep}$ incident to $\maxk_{\sep}^2$.

\item We end up considering all the minimal separators $\sep \in \SEP{G}$ by decreasing order.
If $|\sep| \geq 3$ and $\sep$ is {\em not} $\cliquetree{G}$-convergent then, we search for a child $\maxk_{\sep}^3$ of $\maxk_{\sep}^2$ that is incident to all edges in $\superstrictedgeset{G}{\sep}$.
For simplicity of our analysis, we will also set $\maxk_{\sep}^3 = \maxk_{\sep}^2$ when $|\sep| \leq 2$, or $\sep$ is already $\cliquetree{G}$-convergent, or there is no child node of $\maxk_{\sep}^2$ which satisfies the desired property.
Then, we consider all the children $\maxk$ of $\maxk_{\sep}^2$, $\maxk \neq \maxk_{\sep}^3$, such that $\maxk \cap \maxk_{\sep}^2 = \sep$.
We replace the edge $\maxk\maxk_{\sep}^2$ by $\maxk\maxk_{\sep}^3$.
\end{enumerate}
Before proving correctness of this above algorithm, we want to emphasize two of its main invariants.
In what follows, let $\sep \in \SEP{G}$ be of size at least three.
\begin{enumerate}[label=\textbullet,ref=I\theenumi]
\item\label{invariant-1}({\it Inv.~\ref{invariant-1}}) Assume that at the time we considered $\sep$ during Phase 2, $\sep$ was $\cliquetree{G}$-convergent. 
We observe that our choice for $\maxk_{\sep}^2$ preserved this property. In particular, $\sep$ is $\cliquetree{G}$-convergent in the final clique-tree $\cliquetree{G}$ that we output.
\item\label{invariant-2}({\it Inv.~\ref{invariant-2}}) If $\maxk_{\sep}^3 \neq \maxk_{\sep}^2$ (or equivalently, we modify $\edgeset{G}{\sep}$ during Phase 3) then, we also claim that $\sep$ is $\cliquetree{G}$-convergent in the final clique-tree $\cliquetree{G}$ that we output. For that, by the definition of Phase 3 it suffices to prove that the edge between $\maxk_{\sep}^2$ and its parent node (if any) cannot be labelled by $\sep$. According to Invariant~\ref{invariant-1}, $\sep$ was not $\cliquetree{G}$-convergent when it was considered during Phase 2 (otherwise, it should have stayed so until we considered $\sep$ during Phase 3, and so we could have not modified $\edgeset{G}{\sep}$). Therefore by the definition of Phase 2, $\maxk_{\sep}^2$ is the least common ancestor in $\cliquetree{G}$ of all the maximal cliques that are incident to an edge in $\superedgeset{G}{\sep}$, thereby proving the claim.
\end{enumerate}
We will often use these above two invariants in the remaining of the proof.
We finally prove that both properties of the theorem are true for $\cliquetree{G}$.
\paragraph{Proof of Property~\ref{pty-fct:1}.}
Suppose for the sake of contradiction $\sep \in \SEP{G}$ is weakly $\cliquetree{G}$-convergent but not $\cliquetree{G}$-convergent, and $|\sep| \geq 3$.
For clarity, we divide the proof into small claims.

\begin{myclaim}\label{claim:final-1}
After Phase 2 was completed, $\sep$ was not $\cliquetree{G}$-convergent. 
\end{myclaim}
\begin{proofclaim}
Otherwise, every minimal separator $\sep' \supset \sep$ should also be $\cliquetree{G}$-convergent after Phase 2.
But then, the edge-set $\superedgeset{G}{\sep}$ could not be modified during Phase 3.
This would imply $\sep$ is $\cliquetree{G}$-convergent, a contradiction.
\end{proofclaim}
\begin{myclaim}\label{claim:final-2}
From the time we considered $\sep$ during Phase 2 until the end of the algorithm, $\maxk_{\sep}^2$ is the least common ancestor of all maximal cliques incident to an edge in $\superedgeset{G}{\sep}$.
\end{myclaim}
\begin{proofclaim}
We recall that whenever we consider a minimal separator of size at least three during Phase 2, we preserve the property of being $\cliquetree{G}$-convergent (Invariant~\ref{invariant-1}). 
Therefore, Claim~\ref{claim:final-1} implies that at the time we considered $\sep$ during Phase 2, $\sep$ was not $\cliquetree{G}$-convergent.
In particular, we chose $\maxk_{\sep}^2$ to be the least common ancestor of all maximal cliques incident to an edge in $\superedgeset{G}{\sep}$.
This must stay so until the end of the algorithm because during Phase 3, the edges in $\edgeset{G}{\sep}$ were not modified (Invariant~\ref{invariant-2}) and all the edges in $\superstrictedgeset{G}{\sep}$ could only be made incident to a descendant of $\maxk_{\sep}^2$.
\end{proofclaim}
\begin{myclaim}\label{claim:final-3}
After Phase 2 was completed, $\maxk_{\sep}^2$ was not incident to any edge in $\superstrictedgeset{G}{\sep}$.
\end{myclaim}
\begin{proofclaim}
Suppose for the sake of contradiction that $\maxk_{\sep}^2$ was incident to such an edge.
We recall that $\maxk_{\sep}^2$ is the least common ancestor of all maximal cliques incident to an edge in $\superedgeset{G}{\sep}$ (Claim~\ref{claim:final-2}).
Therefore, the only possibility left for $\sep$ being weakly $\cliquetree{G}$-convergent but not $\cliquetree{G}$-convergent is that, during Phase $3$, all edges in $\superstrictedgeset{G}{\sep}$ that were incident to $\maxk_{\sep}^2$ were made incident to one of its children $\maxk$.
But then, during Phase 3 all edges in $\edgeset{G}{\sep}$ should have been made incident to $\maxk$.
This contradicts our assumption that $\sep$ is not $\cliquetree{G}$-convergent.
\end{proofclaim}
In what follows, recall that $\sep^1$ is a largest minimal separator containing $\sep$.
Since $\sep$ is weakly $\cliquetree{G}$-convergent but not $\cliquetree{G}$-convergent, we have $\sep \neq \sep^1$.
\begin{myclaim}\label{claim:final-4}
$\maxk_{\sep^1}^3 = \maxk_{\sep^1}^2 = \maxk_{\sep}^1$.
\end{myclaim}
\begin{proofclaim}
By maximality of $\sep^1$ there is no minimal separator of $G$ which strictly contains $\sep^1$.
In particular, it easily follows from this observation that we have $\maxk_{\sep^1}^3 = \maxk_{\sep^1}^2$.
So, we will only prove $\maxk_{\sep^1}^2 = \maxk_{\sep}^1$.
By maximality of $\overrightarrow{v_{\sep}}(\maxk_{\sep}^1)$, $\maxk_{\sep}^1$ was incident to an edge in $\edgeset{G}{\sep^1}$ after Phase 1 was completed.
%Furthermore, $\sep^1$ must be $\cliquetree{G}$-convergent by maximality of $|\sep^1|$.
Therefore, at the time we considered $\sep^1$ during Phase 2, we chose a maximal clique $\maxk_{\sep^1}^2$ such that: either $\maxk_{\sep^1}^2 = \maxk_{\sep}^1$ or $\maxk_{\sep^1}^2$ was a strict ancestor of $\maxk_{\sep}^1$.
Suppose by contradiction $\maxk_{\sep^1}^2$ was a strict ancestor of $\maxk_{\sep}^1$ (otherwise, we are done).
By Claim~\ref{claim:final-2}, $\maxk_{\sep}^2$ is an ancestor of $\maxk_{\sep^1}^2$.
But then, $\maxk_{\sep}^2$ should have been incident to an edge in $\superstrictedgeset{G}{\sep}$ after Phase 2 was completed.
The latter would contradict Claim~\ref{claim:final-3}.
\end{proofclaim}
\begin{myclaim}\label{claim:final-5}
From the time we considered $\sep$ during Phase 2 until the end of the algorithm, $\maxk_{\sep}^2$  is the father node of $\maxk_{\sep}^1$.
Moreover, $\maxk_{\sep}^1 \cap \maxk_{\sep}^2 = \sep$.
\end{myclaim}
\begin{proofclaim}
By Claim~\ref{claim:final-3}, $\maxk_{\sep}^2$ cannot be incident to any edge in $\superstrictedgeset{G}{\sep}$.
We also know that all edges in $\edgeset{G}{\sep}$ are incident to $\maxk_{\sep}^1$.
Thus, either $\maxk_{\sep}^1 = \maxk_{\sep}^2$ or $\maxk_{\sep}^2$ is the father node of $\maxk_{\sep}^1$.
But then in the former case, $\maxk_{\sep}^2$ would be incident to an edge in $\superstrictedgeset{G}{\sep}$ (Claim~\ref{claim:final-4}), a contradiction.
As a result, $\maxk_{\sep}^2$ is the father node of $\maxk_{\sep}^1$, and then we must have $\maxk_{\sep}^1 \cap \maxk_{\sep}^2 = \sep$.
\end{proofclaim}
Finally, since we did not modify $\edgeset{G}{\sep}$ during Phase 3 (cf. Invariant~\ref{invariant-2}), not all edges in $\superstrictedgeset{G}{\sep}$ can be incident to the same child node of $\maxk_{\sep}^2$.
Then, by Claim~\ref{claim:final-5}, not all edges in $\superstrictedgeset{G}{\sep}$ can be incident to $\maxk_{\sep}^1$.
By Claim~\ref{claim:final-4} we get that all the edges in $\superstrictedgeset{G}{\sep}$ must be incident to some child $\maxk$ of $\maxk_{\sep}^1$.
\begin{myclaim}\label{claim:final-6}
There must be a second largest minimal separator $\sep^2$ that strictly contains $\sep$.
\end{myclaim}
\begin{proofclaim}
Otherwise it would follow from Claim~\ref{claim:final-4} that at the time we considered $\sep$ during Phase 2, $\sep$ was $\cliquetree{G}$-convergent, thereby contradicting Claim~\ref{claim:final-1}.
\end{proofclaim}
Let $\sep^2$ be as in Claim~\ref{claim:final-6}.
Suppose for the sake of contradiction that we modified the edges in $\edgeset{G}{\sep^2}$ during Phase 3.
As we have $\maxk_{\sep^2}^3 = \maxk$, this implies that $\maxk_{\sep^2}^2$ was the father node of $\maxk$.
Then, since $\sep^1$ was considered before $\sep^2$, and that we have $\maxk \cap \maxk_{\sep}^1 = \sep^1$, we would get $\maxk_{\sep^2}^2 = \maxk_{\sep}^1$.
Recall that all edges in $\edgeset{G}{\sep^1}$ are incident to $\maxk_{\sep}^1$.
However by maximality of $|\sep^2|$, $\sep^2$ can only be contained into $\sep^1$. 
Hence, we should have not modified the edges in $\edgeset{G}{\sep^2}$ during Phase 3, that is a contradiction.
As a result, $\maxk_{\sep^2}^2 = \maxk$.

We now prove that at the time we considered $\sep^2$ during Phase 2, $\maxk$ was already incident to an edge in $\edgeset{G}{\sep^2}$.
Indeed, suppose for the sake of contradiction that this was not the case.
Then, there would exist $\sep' \supset \sep^2$ such that $\maxk$ was incident to an edge in $\edgeset{G}{\sep'}$.
Since we chose $\maxk$ instead of $\maxk_{\sep}^1$, by Claim~\ref{claim:final-4} we can always assume $\sep' \neq \sep^1$, a contradiction.

Overall, after Phase 1 was completed, $\maxk$ was incident to an edge in $\edgeset{G}{\sep^2}$.
We claim that more generally, $\maxk$ was incident to an edge in $\edgeset{G}{\sep^1}$ and $\edgeset{G}{\sep^2}$ after Phase 1 was completed.
Indeed, by Claim~\ref{claim:final-4}, we have $\maxk \neq \maxk_{\sep^1}^2$.
In this situation, since $\maxk$ was incident to an edge in $\edgeset{G}{\sep^1}$ after Phase 2, then it must be the case that $\maxk$ was already incident to such an edge after Phase 1 was completed. 

We observe that at the time that we considered $\sep$ during Phase 1, we already considered both $\sep^1$ and $\sep^2$.
So, it was already the case that $\maxk$ was incident to an edge in $\edgeset{G}{\sep^1}$ and $\edgeset{G}{\sep^2}$.
Furthermore we claim that after Phase 1, $\maxk_{\sep}^1$ was not incident to any edge in $\edgeset{G}{\sep^2}$. 
Indeed, this is because $\maxk_{\sep^2}^2$, and so $\maxk_{\sep^2}^1$, is a strict descendant of $\maxk_{\sep}^1$, and by Claim~\ref{claim:final-4} we have $\maxk_{\sep^2}^2 \cap \maxk_{\sep}^1 = \sep^1 \neq \sep^2$.
However, the latter contradicts the maximality of $\overrightarrow{v_{\sep}}(\maxk_{\sep}^1)$. 
\paragraph{Proof of Property~\ref{pty-fct:2}.}
Finally, for any $i > 0$, let us assume the existence of a minimal separator $\sep$ of $G_i$ that is contained into $\sep_i$.
In particular, $\maxk_{\sep}^3$ is a descendant of $\maxk_i$.
We only need to consider the following two cases:
\begin{itemize}
\item Case $\maxk_{\sep}^3 \neq \maxk_{\sep}^2$.
Then, as we modified $\edgeset{G}{\sep}$ during Phase 3, we must have $|\sep| \geq 3$ and $\sep$ is $\cliquetree{G}$-convergent (cf. Invariant~\ref{invariant-2}).
It implies $\maxk_{\sep}^3 = \maxk_i$ because we assume $\sep \subseteq \sep_i$.
In this situation, $\maxk_{\sep}^2 = \maxk_{p(i)}$.
Furthermore, we recall that at the time we considered $\sep$ during Phase 3, $\sep$ was not $\cliquetree{G}$-convergent. 
So there is at least one edge in $\superstrictedgeset{G}{\sep}$ to which $\maxk_{p(i)}$ is not incident.
In particular, such an edge is incident to $\maxk_i$, and so, it is labeled by a minimal separator of $G_i$.
Overall, we obtain as desired that $\sep$ is strictly contained into a minimal separator of $G_i$.
\item Case $\maxk_{\sep}^3 = \maxk_{\sep}^2$.
Note that $\maxk_{p(i)}$ is incident to an edge in $\edgeset{G}{\sep_i}$, and so to an edge in $\superedgeset{G}{\sep}$.
However, at the time we considered $\sep$ during Phase 2, we chose a $\maxk_{\sep}^2$ that was {\em not} an ancestor of $\maxk_{p(i)}$ (otherwise, this should have stayed so during Phase 3).
In particular, we chose a $\maxk_{\sep}^2$ that was not the least common ancestor of all the maximal cliques incident to an edge in $\superedgeset{G}{\sep}$.
As a result, $|\sep|\geq 3$ and $\sep$ was $\cliquetree{G}$-convergent.
We observe that in this situation, after Phase 2 was completed all the minimal separators $\sep'$ containing $\sep$ were $\cliquetree{G}$-convergent.
Therefore the set $\superedgeset{G}{\sep}$ was not modified during Phase 3.
Since $\sep \subseteq \sep_i$ and $\maxk_{\sep}^2 \neq \maxk_{p(i)}$, we so obtain $\maxk_{\sep}^2 = \maxk_{i}$.
Furthermore, $\maxk_{p(i)}$ is not incident to all edges in $\superstrictedgeset{G}{\sep}$ (otherwise, during Phase 2 we could have chosen $\maxk_{\sep}^2 = \maxk_{p(i)}$).
Overall, it implies that $\sep$ is $\cliquetree{G}$-convergent and, since $\maxk_{p(i)}$ is not incident to all edges in $\superstrictedgeset{G}{\sep}$, $\sep$ is strictly contained into a minimal separator of $G_i$.
\end{itemize}
\end{proof}

\section{Step~\ref{step-2}: A family of subtrees for the Clique-Intersections}\label{sec:ci-subtrees}

For every clique-intersection $\ci \in \MAXK{G} \cup \SEP{G}$, we aim at computing a polynomial-size representation of the family of all possible subtrees $\unlabeledsubtree{\ci}$ that we could encounter in a well-structured $4$-Steiner root of $G$. -- Our approach also works for weak minimal separators, however this is not needed for our algorithm. -- Correctness of this part mostly follows from Theorem~\ref{thm:x-free}.

In what follows, let $\cliquetree{G}$ be an arbitrary rooted clique-tree of $G$ ({\it i.e.}, not necessarily the one computed in Sec.~\ref{sec:clique-tree}).
We first focus on the case of minimal separators in Section~\ref{sec:minsep}, before extending our results to the maximal cliques that are either leaf-nodes or internal nodes of $\cliquetree{G}$ (Sections~\ref{sec:leaf-case} and~\ref{sec:super-selection}, respectively).

\subsection{Case of Minimal Separators}\label{sec:minsep}

We first prove that for any minimal separator $\sep$, the family ${\cal T}_{\sep}$ of all the potential subtrees $\unlabeledsubtree{\sep}$ has polynomial size.
Moreover, we can enumerate all possible subtrees $\unlabeledsubtree{\sep}$ in polynomial time. 
This result will be the cornerstone of all our other constructions in this part.

\begin{theorem}\label{thm:minsep}
Let $\SEP{G}$ be the set of all minimal separators in $G=(V,E)$. 
In ${\cal O}(n^5m)$-time we can construct a collection $({\cal T}_{\sep})_{\sep \in \SEP{G}}$ such that, for any well-structured $4$-Steiner root $\tree$ of $G$, and for any $\sep \in \SEP{G}$, $\unlabeledsubtree{\sep}$ is Steiner-equivalent to some subtree in ${\cal T}_{\sep}$.
\end{theorem}	

%This result will be further exploited in the next sections.
Before proving this above theorem, let us describe the main difficulty we had to face on.
Roughly, given $\sep \in \SEP{G}$ the difficulty in generating ${\cal T}_{\sep}$ comes from the bistars, as a brute-force generation of all possibilities would take time exponential in $|\sep|$.
In order to remedy to that issue we use the fact that in a well-structured $4$-Steiner root of $G$, $\sep$-free vertices are leaves of such a bistar with all of them, except maybe one, adjacent to the same central node.
For a fixed placement of the $\sep$-constrained vertices, this only gives us ${\cal O}(|\sep|)$ possibilities in order to place the $\sep$-free vertices.
Overall, we reduce the number of possible bistars to an ${\cal O}(|\sep|^5)$.

\begin{proof}
Let $\sep \in \SEP{G}$ be fixed. 
Up to some pre-processing we will construct ${\cal T}_{\sep}$ in ${\cal O}(|\sep|^6)$-time. 
Since $\max\{ \ |\sep| \mid \sep \in \SEP{G} \} = {\cal O}(n)$ and $\sum_{\sep \in \SEP{G}}|\sep| = {\cal O}(m)$~\cite{BlP93}, the latter will prove the result.

\medskip
\noindent
{\bf Case $diam(\unlabeledsubtree{\sep}) \leq 2$.}
Let us start with some easy cases. 
If $|\sep| = 1$ then, it suffices to add a single-node tree to ${\cal T}_{\sep}$.
Similarly, if $|\sep|=2$ then, by Theorem~\ref{thm:clique-intersection}, $\sep$ must induce a path of length at most $k-1=3$ in any $4$-Steiner root of $G$ with its two ends being the vertices of $\sep$.
This gives only ${\cal O}(1)$ possibilities to put into ${\cal T}_{\sep}$.
Thus, from now on assume $|\sep| \geq 3$.
Given any $4$-Steiner root $\tree$ of $G$, by Theorem~\ref{thm:clique-intersection} the subtree $\unlabeledsubtree{\sep}$ can only be a star or a bistar (but the latter only if $\sep$ is inclusion wise maximal in $\SEP{G}$).
Furthermore in the former case, all leaves in the star $\unlabeledsubtree{\sep}$ must be in $\sep$, and the center node can either be in $\sep$ or Steiner.
Overall, this gives ${\cal O}(|\sep|)$ possibilities of stars to put into ${\cal T}_{\sep}$, and so, this takes ${\cal O}(|\sep|^2)$-time. 

\medskip
\noindent
{\bf Case $diam(\unlabeledsubtree{\sep}) = 3$.}
We end up focusing on the case where $|\sep| \geq 3$ and $\unlabeledsubtree{\sep}$ may be a bistar.
In what follows, the two central nodes of such a bistar will be always denoted by $\centre{\unlabeledsubtree{\sep}} = \{ c_0, c_1 \}$.
We will introduce the following additional terminology.
A {\em heavy part} of $\sep$ is any clique-intersection $\ci \subset \sep$ such that $|\ci| \geq 3$.
A {\em light part} of $\sep$ is any clique-intersection $\ci \subset S$ such that $|\ci| = 2$.
We prove the following intermediate claim (also used in other parts of the paper):
\begin{myclaim}\label{claim:no-middle-edge}
If $\unlabeledsubtree{\sep}$ is a bistar and $\centre{\unlabeledsubtree{\sep}}$ is a light part then, there is a heavy part that strictly contains $\centre{\unlabeledsubtree{\sep}}$.
\end{myclaim}
\begin{proofclaim}
Suppose by contradiction $\ci = \{c_0,c_1\}$ is a light part and no heavy part contains it.
Let $\maxk_i$ be any maximal clique such that $\ci \subseteq \maxk_i$ but $\sep \not\subseteq \maxk_i$.
Such a $\maxk_i$ always exists since otherwise, taking the intersection of $\sep$ with all the maximal cliques that contains $\ci$, one would obtain $\sep = \ci$, a contradiction.
In this situation, $\ci \subseteq \maxk_i \cap \sep$, and so $\ci = \maxk_i \cap \sep$ since there is no heavy part containing $\ci$.
Furthermore we have $\maxk_i \not\subseteq \sep$.
Hence, there exists $j \in \{0,1\}$ such that $c_j$ has a neighbour in $\unlabeledsubtree{\maxk_i} \setminus \sep$ (possibly, a Steiner node).
By applying Lemma~\ref{lem:tree-intersection} to $N_{\tree}[c_j]$ and $\unlabeledsubtree{\maxk_i}$, we obtain:$$diam(N_{\tree}[c_j] \cup \unlabeledsubtree{\maxk_i}) = diam(\unlabeledsubtree{\maxk_i}) \leq 4.$$
In particular, $\maxk_i \cup Real(N_{\tree}[c_j])$ is a clique of $G$.
By maximality of $\maxk_i$, $Real(N_{\tree}[c_j]) \subseteq \maxk_i$.
However, there is at least one leaf in $Real(N_{\tree}[c_j]) \setminus \ci$, that implies $\ci \subset \maxk_i \cap \sep$, a contradiction.
\end{proofclaim}
Then, we divide the proof in two subcases.
We stress that both subcases only depend on the clique-arrangement of $G$, and are independent of any $4$-Steiner root $\tree$ --- we use the existence of such a root only for proving correctness of our construction.

\begin{itemize}
\item
We first consider the particular subcase when there exists a heavy part $\ci \subset \sep$.
In this situation, $\ci \subseteq N_{\tree}[c_0]$ or $\ci \subseteq N_{\tree}[c_1]$ for any well-structured $4$-Steiner root $\tree$ of $G$ where $\unlabeledsubtree{\sep}$ is a bistar.
Moreover, by Property~\ref{pty-ci:2} of Theorem~\ref{thm:clique-intersection} either $Real(N_{\tree}[c_0]) = \ci$ or $Real(N_{\tree}[c_1]) = \ci$.
Therefore, we can start choosing among ${\cal O}(|\ci|)$ possibilities the star induced by $\ci$ in $\tree$.
W.l.o.g., $c_0$ is the center of this star.
The other center $c_1$ must be either a Steiner node adjacent to $c_0$ (in $\tree$) or any vertex in $\ci \setminus \{c_0\}$.
Hence, there are also ${\cal O}(|\ci|)$ possibilities for $c_1$.
Finally, since we have $Real(N_{\tree}[c_0]) = \ci$ all the nodes in $\sep \setminus \ci$ must be leaves adjacent to $c_1$.
Overall, this gives ${\cal O}(|\sep|^2)$ possibilities of bistars to put into ${\cal T}_{\sep}$, and so, this takes ${\cal O}(|\sep|^3)$-time\footnote{
We will actually show in Lemma~\ref{lem:real-center-bistar} this number of potential bistars can be reduced. However, we choose not to include this improvement in this part of the proof in order to keep it as simple as possible.}.

\item
From now on we assume that there is no heavy part.
Following Definition~\ref{def:x-free} we process the vertices in $\sep$ according to the other clique-intersections in which they are contained.

\smallskip
\noindent
{\bf Processing the internally $\sep$-constrained vertices.}
For any well-structured $4$-Steiner root $\tree$ where $\unlabeledsubtree{\sep}$ is a bistar, by Claim~\ref{claim:no-middle-edge}, $\centre{\unlabeledsubtree{\sep}} = \{c_0,c_1\}$ is not a light part.
Furthermore, given any light part $\ci \subset \sep$, we can prove that either $\ci$ induces an edge of $\unlabeledsubtree{\sep}$, or $\unlabeledsubtree{\ci}$ is a non-edge star and $\ci = Real(N_{\tree}[c_j])$ for some $j$ (this also follows from Property~\ref{pty-ci:2} of Theorem~\ref{thm:clique-intersection}).
In this situation, we construct the subgraph $I_2$: of which the vertices are the internally $\sep$-constrained vertices, and the edges are the light parts.

\begin{myclaim}\label{claim:igraph-light}
If $G$ has a $4$-Steiner root $\tree$ where $\unlabeledsubtree{\sep}$ is a bistar then, $I_2$ has at most two connected components.
Specifically:
\begin{enumerate}
\item Either $I_2$ is connected, there is a unique light part $\ci$ such that $\unlabeledsubtree{\ci}$ is a non-edge star, and all other light parts must be edges incident to the unique central node in $\centre{\unlabeledsubtree{\sep}} \setminus \centre{\unlabeledsubtree{\ci}}$. 
\item Or the connected components of $I_2$ induce node-disjoint stars in $\unlabeledsubtree{\sep}$.
\end{enumerate} 
\end{myclaim}
\begin{proofclaim}
The result is a consequence of the following case analysis:
\begin{itemize}
\item Assume there is a light part $\ci \subset \sep$ such that $\unlabeledsubtree{\ci}$ is a non edge star and $Real(N_{\tree}[c_j]) = \ci$.
Two different situations might occur:
\begin{enumerate}
\item \underline{Situation \# 1} (see Fig.~\ref{fig:bistar-situation-1}): There is a light part $\ci' \neq \ci$ intersecting $\ci$. Since the vertex in $\ci' \cap \ci$ cannot be a leaf adjacent to $c_j$ (otherwise, we should have $\ci' \subseteq Real(N_{\tree}[c_j]) = \ci$, a contradiction), this must be $c_{1-j}$. In particular, for {\em any} light part $\ci' \neq \ci$, $\unlabeledsubtree{\ci'}$ must be an edge between $c_{1-j}$ and a leaf (thereby implying $\ci' \cap \ci \neq \emptyset$).

\begin{figure}[h!]
\centering
\includegraphics[width=.3\textwidth]{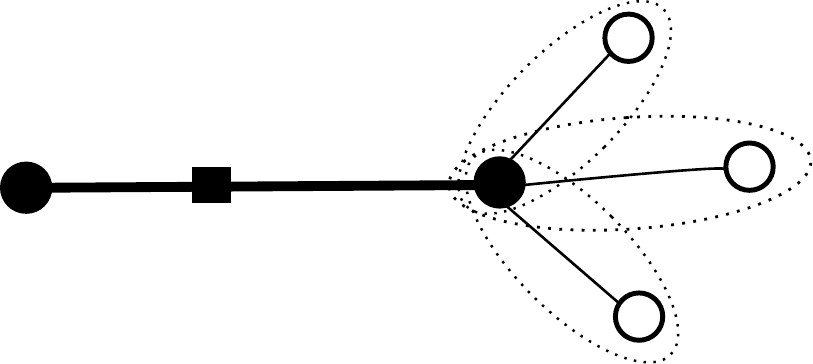}
\caption{Situation 1: the subtree $\unlabeledsubtree{\ci}$ is drawn in bold. There are 3 other light parts represented by dashed ellipses.}
\label{fig:bistar-situation-1}
\end{figure}

\item \underline{Situation \# 2} (see Fig.~\ref{fig:bistar-situation-2}): There is no other light part $\ci' \neq \ci$ intersecting $\ci$.
Since $\unlabeledsubtree{\sep} \setminus N_{\tree}[c_j]$ is an independent set, any light part $\ci' \subset \sep$ that does not intersect $\ci$ cannot be an edge.
We so deduce that if such a $\ci'$ exists then, $Real(N_{\tree}[c_{1-j}]) = \ci'$, and so, there are no other light part in $\sep$ than $\ci$ and $\ci'$.

\begin{figure}[h!]
\centering
\includegraphics[width=.25\textwidth]{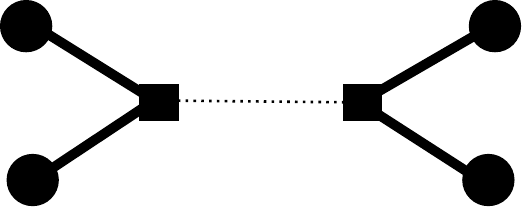}
\caption{Situation 2: the subtrees $\unlabeledsubtree{\ci}$ and $\unlabeledsubtree{\ci'}$ are drawn in bold.}
\label{fig:bistar-situation-2}
\end{figure}

\end{enumerate}

\item Otherwise, each light part is an edge of $\unlabeledsubtree{\sep}$ that contains either $c_0$ or $c_1$, but not both.
Therefore, there is a one-to-one mapping between the connected components of $I_2$ and the nonempty sets among $Real(N_{\tree}[c_0]) \setminus \{c_1\}, Real(N_{\tree}[c_1]) \setminus \{c_0\}$.
See Fig.~\ref{fig:bistar-only-edge} for an illustration.

\begin{figure}[h!]
\centering
\includegraphics[width=.4\textwidth]{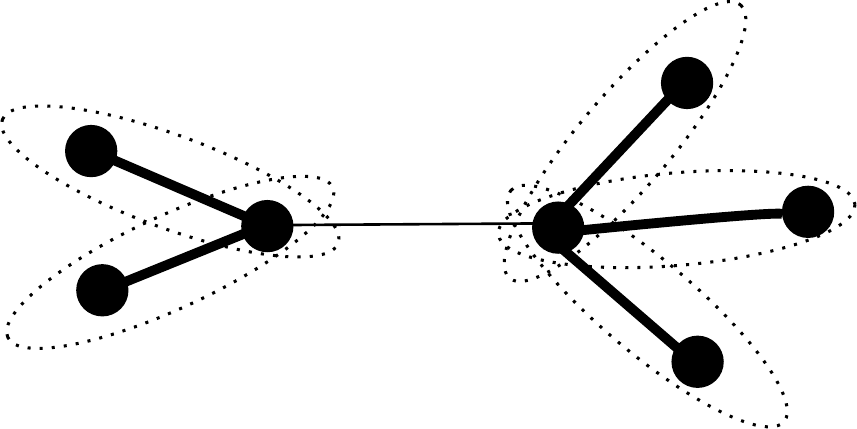}
\caption{A case where all light parts must be edges.}
\label{fig:bistar-only-edge}
\end{figure}

\end{itemize}
\end{proofclaim}
Overall, Claim~\ref{claim:igraph-light} reduces the placement of internally $\sep$-constrained vertices to the construction of two stars in parallel, thereby giving ${\cal O}(|\sep|^2)$ different possibilities.

\smallskip
\noindent
{\bf Processing the $(\sep,\ci_1,\ci_2)$-sandwiched vertices.}
Then, we consider all the other $\sep$-constrained vertices.
By Lemma~\ref{lem:bistar-center} in the previous Section, $\sep$ is strictly contained into exactly two maximal cliques, that we denote $\maxk$ and $\maxk'$.
In particular, if $v \in \sep$ is $(\sep,\ci_1,\ci_2)$-sandwiched then, either $\ci_1 = \maxk$ or $\ci_1 = \maxk'$.
In the former case we call $v$ {\em $\maxk$-dependent}, and in the latter case we call it {\em $\maxk'$-dependent}.

The following Claim~\ref{claim:exclusive-dependency} shows that it defines two equivalence classes:

\begin{myclaim}\label{claim:exclusive-dependency}
Assume that $G$ has a $4$-Steiner root $\tree$ where $\unlabeledsubtree{\sep}$ is a bistar, and let $v \in \sep$ be {\em not} internally $\sep$-constrained. 
If $v$ is $\maxk$-dependent then, either $v \in \centre{\unlabeledsubtree{\sep}} \cap \centre{\unlabeledsubtree{\maxk}}$, or $v$ is a leaf of $\unlabeledsubtree{\sep}$ that is adjacent in $\tree$ to the unique node in $\centre{\unlabeledsubtree{\maxk}}$.

In particular, a vertex $v \in \sep$ cannot be both $\maxk$-dependent and $\maxk'$-dependent (unless it is also internally $\sep$-constrained).
\end{myclaim}
\begin{proofclaim}
By the proof of Lemma~\ref{lem:bistar-center}, $\centre{\unlabeledsubtree{\sep}} = \{c_0,c_1\}$ where $\centre{\unlabeledsubtree{\maxk}} = \{c_0\}$ and $\centre{\unlabeledsubtree{\maxk'}} = \{c_1\}$.
Now, suppose by contradiction $v$ is either $c_1$ or a leaf adjacent to $c_1$.

Let $\ci_2 \in \CI{G}$ be such that $v$ is $(\sep,\maxk,\ci_2)$-sandwiched. 
As we have $\ci_2 \cap \sep = \{v\}$, $\ci_2$ is contained into at least one maximal clique $\maxk''$ that does not contain $\sep$.
Furthermore, we can assume w.l.o.g. $\maxk'' \cap \sep = \{v\}$ since otherwise $v$ is internally $\sep$-constrained, a contradiction.
Thus from now on we will assume $\ci_2 = \maxk''$ is a maximal clique that does not contain $\sep$.

Let $x \in (\maxk \cap \ci_2) \setminus \{v\}$.
On one hand, $x \notin \maxk'$ (because $x \notin \sep$) and so, $dist_{\tree}(x,c_1) \geq 3$.
On the other hand, $dist_{\tree}(x,c_0) \leq 2$ because $x \in \maxk$.
Altogether combined, $dist_{\tree}(x,c_0) = 2$ and the unique $xc_1$-path in $\tree$ goes by $c_0$.
Note that it also implies that the unique $xv$-path in $\tree$ goes by the edge $c_0c_1$, hence $dist_{\tree}(x,v) \in \{3,4\}$.

Furthermore, let $y \in \ci_2 \setminus \maxk$, that exists since we assume $\ci_2$ to be a maximal clique.
We prove as a subclaim $y \notin \maxk'$.
Indeed, by Theorem~\ref{thm:clique-tree-pties} there is an edge labeled by $\sep$ in any clique-tree of $G$.
This edge must be $\maxk\maxk'$, that implies  there can be no edge in $G$ between $\maxk \setminus \maxk'$ and $\maxk' \setminus \maxk$.
Since we have $xy \in E(G)$, this proves our subclaim that $y \notin \maxk'$.
In particular, we also have $dist_{\tree}(y,c_1) \geq 3$.
%Suppose for the sake of contradiction that there exist $v \in \sep$ which is both $\maxk$-dependent and $\maxk'$-dependent.
%Since there is no edge between $\maxk \setminus \maxk'$ and $\maxk' \setminus \maxk$, either $dist_{\tree}(v,\maxk \setminus \maxk') \geq 3$ or $dist_{\tree}(v,\maxk' \setminus \maxk) \geq 3$.
%W.l.o.g. we will assume $dist_{\tree}(v,\maxk \setminus \maxk') \geq 3$.
But then, since we must also have  $dist_{\tree}(x,y) \leq 4$, the $yv$-path in $\tree$ also goes by the edge $c_0c_1$ (otherwise, $dist_{\tree}(x,y) = dist_{\tree}(x,c_1) + dist_{\tree}(c_1,y) \geq 6$).
Since $dist_{\tree}(y,c_0) \geq 3$ (because $y \notin \maxk$) we so obtain that $dist_{\tree}(y,c_0) = 3$ and $c_1 = v$.

However, there exists $z \in \sep$ a leaf adjacent to $c_0$.
In particular, there exists a clique-intersection $\ci$ that contains all of $x,y,z,v$.
As we assume $v$ is not internally $\sep$-constrained, $\sep \subset \ci$.
This is a contradiction because we cannot have $\ci = \maxk$ nor $\ci = \maxk'$.
\end{proofclaim}
Given a fixed placement of internally $\sep$-constrained vertices, by Claim~\ref{claim:exclusive-dependency} we can generate all possible placements of the ``sandwiched'' vertices, as follows.
We choose the central node in $\centre{\unlabeledsubtree{\maxk}}$ and the central node in $\centre{\unlabeledsubtree{\maxk'}}$ (there are ${\cal O}(|\sep|^2)$ possibilities).
Then, all remaining $\maxk$-dependent vertices, resp. $\maxk'$-dependent, must be added as leaves adjacent to the unique node in $\centre{\unlabeledsubtree{\maxk}}$, resp. to the unique central node in $\centre{\unlabeledsubtree{\maxk'}}$.
Overall, we have ${\cal O}(|\sep|^2) \times {\cal O}(|\sep|^2) = {\cal O}(|\sep|^4)$ possibilities for positioning the $\sep$-constrained vertices.  

\smallskip
\noindent
{\bf Processing the $\sep$-free vertices.}
However, each such a possibility does not quite define a potential bistar for $\sep$ as we also need to position the $\sep$-free vertices.
By Theorem~\ref{thm:x-free}, we can always assume the $\sep$-free vertices to be leaf-nodes with all of them except maybe one adjacent to the same central node of $\unlabeledsubtree{\sep}$.
In particular, given a fixed placement of the $\sep$-constrained vertices, there are ${\cal O}(|\sep|)$ possibilities in order to place the $\sep$-free vertices (specifically, we choose among ${\cal O}(|\sep|)$ possibilities the unique $\sep$-free vertex that is not adjacent to the same central node as the others, if any, as well as the central node to which all other $\sep$-free vertices must be adjacent).
\end{itemize}
Summarizing, we only need to add ${\cal O}(|\sep|^5)$ different trees in ${\cal T}_{\sep}$, that takes ${\cal O}(|\sep|^6)$-time.

\smallskip
{\bf Final comments.}
A careful reader maybe observed that in our above analysis we ignored the complexity of several operations such as: computing the heavy parts and the light parts of $\sep$, and in the same way computing the $\maxk$-dependent vertices (resp. computing the $\maxk'$-dependent vertices).
By Theorem~\ref{thm:clique-arrangement} we can first compute the clique-arrangement of $G$ in polynomial time, and then all these above operations can be easily done in polynomial time as well.
However, in order to prove that our algorithm truly runs in ${\cal O}(n^5m)$-time, one needs to show that performing these above operations does not dominate the total running-time.  

\begin{myclaim}\label{claim:compute-parts}
Assume that $G$ has a $4$-Steiner root $\tree$ where $diam(\unlabeledsubtree{\sep}) \leq 3$.
Then, we have that: $$\sum_{\ci \in \CI{G}: \ci \subseteq \sep} |\ci| = {\cal O}(|\sep|).$$
\end{myclaim}
\begin{proofclaim}
Since $\unlabeledsubtree{\sep}$ has ${\cal O}(|\sep|)$ nodes and ${\cal O}(|\sep|)$ edges, we only need to consider the clique-intersections $\ci \subset \sep$ such that $\unlabeledsubtree{\ci}$ is a non-edge star. We are done as there can only be at most two such clique-intersections.
\end{proofclaim}
This first result above tells us how to compute the light parts and heavy parts of $\sep$.
Specifically, we can simply enumerate all the clique-intersections $\ci$ such that $\ci \subset \sep$, by using the clique-arrangement of $G$.
Furthermore if this enumeration takes more than ${\cal O}(|\sep|)$-time then we can stop as by Claim~\ref{claim:compute-parts}, $\sep$ cannot be mapped to a bistar in any $4$-Steiner root of $G$.
Overall, all the light parts and heavy parts of $\sep$ can be enumerated in total time ${\cal O}(|\sep|)$ if the clique-arrangement of $G$ is given, that is in ${\cal O}(n)$.

\begin{myclaim}\label{claim:compute-k-dependent}
All the $\maxk$-dependent vertices in $\sep$ can be computed in total ${\cal O}(nm\log{n})$-time.
\end{myclaim}
\begin{proofclaim}
For every clique-intersection $\ci \in \CI{G}$, we compute $\ci \cap \sep$ and $\ci \cap \maxk$.
This can be done in ${\cal O}(|\ci|)$-time assuming a trivial pre-processing in ${\cal O}(|\sep|+|\maxk|)$-time for marking all the vertices in $\sep$ and in $\maxk$, respectively.
Then, if $ 1 = |\ci \cap \sep| < |\ci \cap \maxk|$ the unique vertex of $\ci \cap \sep$ is $\maxk$-dependent.
Conversely, it follows from the definition that any $\maxk$-dependent vertex in $\sep$ can be computed this way.
The total running time is in ${\cal O}(\sum_{\ci \in \CI{G}}|\ci|) = {\cal O}(n |\CI{G}|)$.
Since $\CI{G}$ is exactly the vertex-set of the clique-arrangement of $G$, by Theorem~\ref{thm:clique-arrangement} we get ${\cal O}(n |\CI{G}|) = {\cal O}(nm\log{n})$.
\end{proofclaim}
We recall that there are ${\cal O}(n)$ minimal separators in $G$.
Therefore our above approach only requires a pre-processing in total ${\cal O}(n^2m\log{n})$-time, that is in $o(n^5m)$.
\end{proof}

\begin{remark}\label{rk:extension-minsep}
We only use the fact that $\sep$ is a minimal separator when we process the so-called ``sandwiched vertices''.
All other arguments in our proof stay valid for clique-intersections. 
Furthermore, as already observed for Lemma~\ref{lem:almost-simplicial-characterization}, such ``sandwiched vertices'' do not exist in maximal cliques.  
Therefore, we can also use the algorithm of Theorem~\ref{thm:minsep} in order to generate, for any maximal clique $\maxk_i$ without a $\maxk_i$-free vertex, all possible subtrees $\unlabeledsubtree{\maxk_i}$ of diameter at most $3$ in any well-structured $4$-Steiner root of $G$.
\end{remark}

\subsection{Case of a Leaf Node}\label{sec:leaf-case}

In this section, we generalize the construction of Theorem~\ref{thm:minsep} to the maximal cliques that can be leaves in a rooted clique-tree.
For that, we use a well-known decomposition of these maximal cliques into a unique minimal separator and a set of simplicial vertices.

\begin{theorem}\label{thm:leaf}
Given $G=(V,E)$ and a rooted clique-tree $\cliquetree{G}$ of $G$, let $\maxk_i \in \MAXK{G}$ be a leaf.
We can construct, in time polynomial in $|\maxk_i|$, a set ${\cal T}_i$ of $4$-Steiner roots for $G_i := G[\maxk_i]$ with the following additional property: In any {\em well-structured} $4$-Steiner root $\tree$ of $G$, there exists $\labeledsubtree[2]{i}{} \in {\cal T}_i$ Steiner-equivalent to $\unlabeledsubtree{\maxk_i}$.
\end{theorem}

\begin{proof}
Let $\maxk_{p(i)}$ be the father node of $\maxk_i$.
By Theorem~\ref{thm:clique-tree-pties}, $\sep_i := \maxk_i \cap \maxk_{p(i)}$ is a minimal separator.
We compute the family ${\cal T}_{\sep_i}$ given by Theorem~\ref{thm:minsep}.
Then, in order to compute a candidate subtree $\labeledsubtree{i}{}$, to be added into ${\cal T}_i$, we consider all the subtrees $\labeledsubtree{\sep_i}{} \in {\cal T}_{\sep_i}$ and we proceed as follows.
We select a node in $\labeledsubtree{\sep_i}{}$ that we assume to be closest to $\centre{\labeledsubtree{i}{}}$ (hence, ${\cal O}(|V(\labeledsubtree{\sep_i}{})|) = {\cal O}(|\sep_i|)$ possibilities), and we set its distance to the center (this can only be $0,1$ or $2$). 
In doing so, we can assume $\centre{\labeledsubtree{i}{}}$ to be added into $\labeledsubtree{\sep_i}{}$.
Finally, the vertices in $\maxk_i \setminus \sep_i$ are all simplicial, and so, we can connect them to $\centre{\labeledsubtree{i}{}}$ as explained in Lemma~\ref{lem:almost-simplicial-placement} (one possibility up to Steiner equivalence).
Note that in doing so, we may also obtain solutions $\labeledsubtree{i}{}$ such that $diam(\labeledsubtree{i}{}) \neq 4$, that we will need to discard. See Fig.~\ref{fig:keaf-construction} for an illustration. 
%or $\labeledsubtree{i}{}$ is not well-structured, that we will need to discard. See Fig.~\ref{fig:keaf-construction} for an illustration. 
Overall, $|{\cal T}_i| = {\cal O}(|\sep_i||{\cal T}_{\sep_i}|) = {\cal O}(|\sep_i|^6)$.
\end{proof}

\begin{figure}
\centering
\includegraphics[width=.4\textwidth]{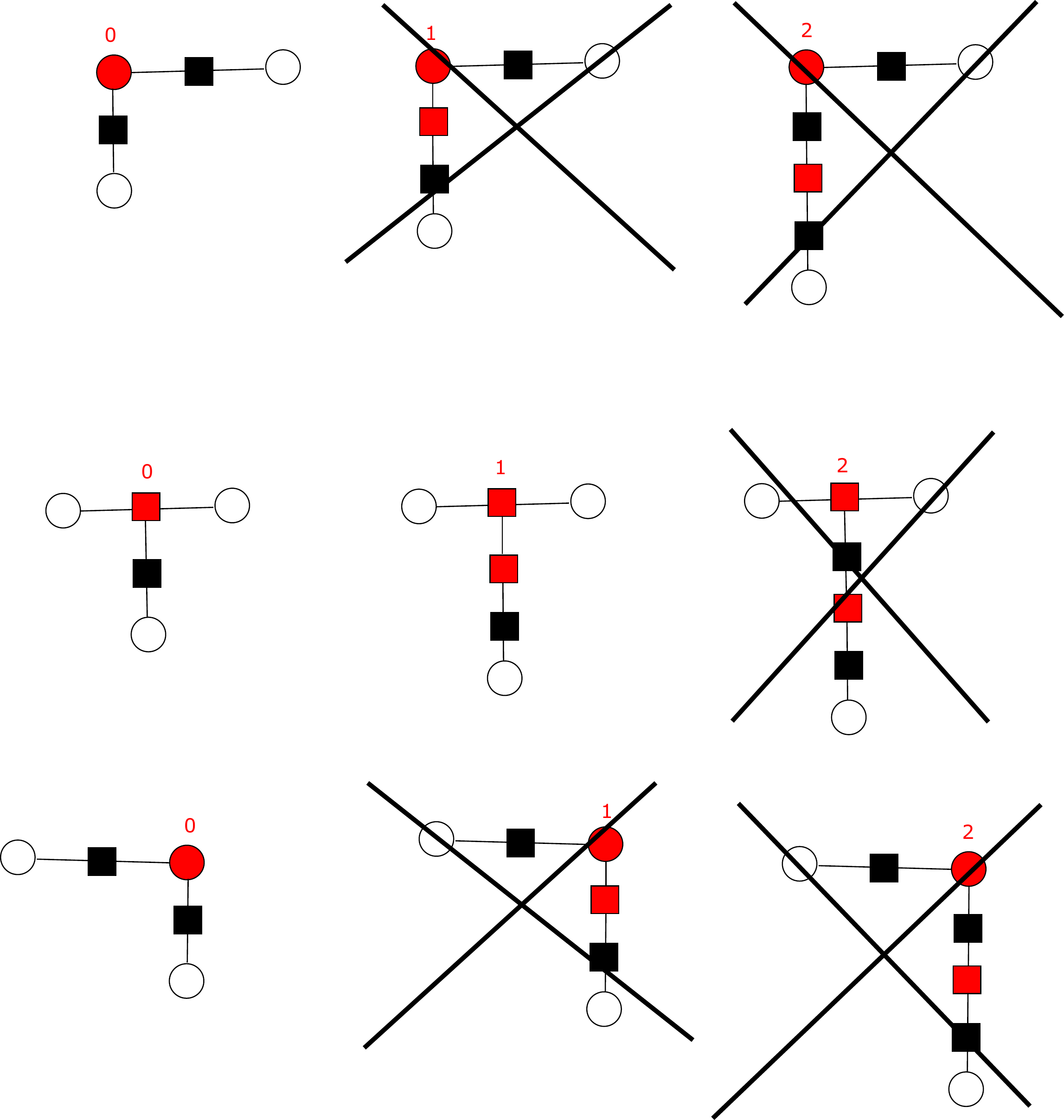}
\caption{An execution of the algorithm of Theorem~\ref{thm:leaf}. The minimal separator has size two and induces a star. There is one simplicial vertex to add.}
\label{fig:keaf-construction}
\end{figure}

\subsection{Case of an internal node}\label{sec:super-selection}

Given $G=(V,E)$ chordal, let $\cliquetree{G}$ be an arbitrary rooted clique-tree of $G$ and let $\maxk_i \in \MAXK{G}$ be an internal node of $\cliquetree{G}$.
We want to compute a polynomial-size representation for the family ${\cal F}_i$ of all possible subtrees $\unlabeledsubtree{\maxk_i}$ in a well-structured $4$-Steiner root of $G$.

\smallskip
\noindent
{\bf Some preliminary observations.}
By Lemma~\ref{lem:almost-simplicial-placement} there always exists a root $\tree$ where all $\maxk_i$-free vertices are leaves of $\unlabeledsubtree{\maxk_i}$ and connected to $\centre{\unlabeledsubtree{\maxk_i}}$ by a path of length two whose internal node is Steiner.
Thus, we can first assume for simplicity $\maxk_i$ does not contain any $\maxk_i$-free vertex --- such vertices, if they exist, will be added at the end of the construction.  
Furthermore, as noticed earlier (Remark~\ref{rk:extension-minsep}) we can use the algorithmic proof of Theorem~\ref{thm:minsep} in order to generate all the subtrees $\unlabeledsubtree{\maxk_i}$ of diameter at most three to be added in ${\cal F}_i$.
Summarizing, as a consequence of other results in this paper we are only interested in maximal cliques $\maxk_i$ with no $\maxk_i$-free vertex and in generating subtrees of diameter {\em exactly} four.

%Here the challenge is to restrict ourselves to some {\em polynomial-size} (polynomial-time computable) family of subtree.
\smallskip
\noindent
{\bf Organization of this part.}
Our main tool for this task is a careful analysis of the intersections between the minimal separators in $\maxk_i$ (Section~\ref{sec:minsep-intersections}).
Unfortunately, sometimes we cannot derive from this information a polynomial bound on the number of possible subtrees.
We identify the only degenerate case when this cannot be done, and show how to handle with it, in Section~\ref{sec:degenerate}.
Proposition~\ref{prop:internal} in Section~\ref{sec:compute-fi} will summarize our results for this part.

\subsubsection{Getting more from clique-intersections}\label{sec:minsep-intersections}

We will use the following lemma in order to prove our first result in this section:

\begin{lemma}\label{lem:real-center-bistar}
Let $G = (V,E)$ be strongly chordal, let $\sep \in \SEP{G}$, let $\maxk$ be a maximal clique containing $\sep$ and let $R,c$ be such that $R \subset S$ and either $c \in R$ or $c$ is Steiner.
We can compute in ${\cal O}(nm\log{n})$-time a node $c'$ with the following properties:
\begin{itemize}
\item For any well-structured $4$-Steiner root $\tree$ of $G$ such that: $\unlabeledsubtree{\sep}$ is a bistar, $c \in \centre{\unlabeledsubtree{\sep}} \setminus \centre{\unlabeledsubtree{\maxk}}$, and $Real(N_{\tree}[c]) = R$, there exists a well-structured root $\tree[2]$ with the same properties such that: $\centre{\unlabeledsubtree[2]{\maxk}} = \{c'\}$, and $dist_{\tree[2]}(u,v) \geq dist_{\tree}(u,v)$ for every $u,v \in V$. Moreover, either $\tree \equiv_G \tree[2]$, or $\sum_{u,v \in V} dist_{\tree[2]}(u,v) > \sum_{u,v \in V} dist_{\tree}(u,v)$. 
\end{itemize}
%Given $G=(V,E)$ strongly chordal, let $\sep \in \SEP{G}$.
%There exists a family ${\cal F}_{\sep}$ with the following two properties:
%\begin{enumerate}[label=\textbullet,ref=P-\ref{lem:real-center-bistar}.\theenumi]
%\item\label{pty:real-center:1}({\it Prop.~\ref{pty:real-center:1}.}) For any well-structured $4$-Steiner root $\tree$ of $G$, there exists a well-structured $\tree[2]$ such that $\left( \unlabeledsubtree[2]{\sep}, \maxk \right) \in {\cal F}_{\sep}$ for some maximal clique $\maxk \supset \sep$, and $dist_{\tree[2]}(u,v) \geq dist_{\tree}(u,v)$ for every $u,v \in V$. Moreover, either $\tree \equiv_G \tree[2]$, or $\sum_{u,v \in V} dist_{\tree[2]}(u,v) > \sum_{u,v} dist_{\tree}(u,v)$. 
%\item\label{pty:real-center:2}({\it Prop.~\ref{pty:real-center:2}.}) For any $R \subset \sep$ and $c \in R$, there is at most one bistar $\unlabeledsubtree{\sep} \in {\cal T}_{\sep}$ such that $c \in \centre{\unlabeledsubtree{\sep}}$ and $N[c] = R$.
%In the same way, there is at most one bistar $\unlabeledsubtree{\sep} \in {\cal T}_{\sep}$ such that $\alpha \in \centre{\unlabeledsubtree{\sep}}$ is Steiner and $N(\alpha) = R$.
%\end{enumerate}
%Moreover, ${\cal T}_{\sep}$ has size ${\cal O}(|\sep|^3)$ and it can be computed in polynomial time.
\end{lemma}

In order to better understand the significance of Lemma~\ref{lem:real-center-bistar}, assume that $\unlabeledsubtree{\sep}$ should be a bistar in the final solution we want to compute, and that we already identified one of its center node $c$ and the set of real nodes $R$ to which this node must be adjacent.
What this above property says is that there is essentially one canonical way to compute the bistar given $R$ and $c$.
The more technical condition $dist_{\tree[2]}(u,v) \geq dist_{\tree}(u,v)$ is simply there in order to ensure that by doing so, we cannot miss a solution of an intermediate problem we call {\sc Distance-Constrained Root} ({\it i.e.}, see Section~\ref{sec:encoding}).
Finally, our condition on the potential function $\sum_{u,v \in V} dist_{\tree[2]}(u,v)$ increasing ensures that we can repeatedly apply our ``canonical completion'' method for arbitrarily many minimal separators $\sep$.

\begin{proof}
%The result is obtained by applying some polynomial-time post-processing to the family ${\cal T}_{\sep}$ of Theorem~\ref{thm:minsep}.
%We consider all possible $R, c$ such that, for some $\unlabeledsubtree{S} \in {\cal T}_{\sep}$ we have $c \in \centre{\unlabeledsubtree{\sep}}$ and $Real(N[c]) = R$.
%W.l.o.g., for any fixed $R$ we keep at most one such a pair where $c$ is Steiner.
%There are only ${\cal O}(|{\cal T}_{\sep}|)$ possibilities for a fixed $\sep$, that is in ${\cal O}(|S|^3)$ by Theorem~\ref{thm:minsep}.
%
%\medskip
%\noindent
%{\bf Computation of the second central node.}
%For each such a pair, we define a node $c' \notin V \setminus R$ (either in $R$ or Steiner), as follows:
We define a node $c' \notin V \setminus R$ (either in $R$ or Steiner), as follows:
\begin{itemize}
\item If there exists a clique-intersection $\ci \subset \sep$ such that $\ci \not\subseteq R$ and $\ci \cap (R \setminus \{c\}) \neq \emptyset$ then, we pick $c' \in (\ci \cap R) \setminus \{c\}$;
\item Else, if there exists $c' \in R \setminus c$ that is $\maxk$-dependent ({\it i.e.}, $c'$ is $(\sep,\maxk,\ci_2)$-sandwiched for some $\ci_2 \in \CI{G}$) then, we output $c'$.
\item Otherwise, $c'$ is Steiner.
\end{itemize}
%Note that as $G$ is strongly chordal, we can compute the above $c'$ in polynomial time by using the clique-arrangement of $G$ (Theorem~\ref{thm:clique-arrangement}).
By using our previous Claims~\ref{claim:compute-parts} and~\ref{claim:compute-k-dependent}, we can compute the above $c'$ in ${\cal O}(nm\log{n})$-time.
%Furthermore, amongst all the bistars $\unlabeledsubtree{\sep} \in {\cal T}_{\sep}$ such that $c \in {\cal C}(\unlabeledsubtree{\sep})$ and $Real(N[c]) = R$, we only keep one such that either $\centre{\unlabeledsubtree{\sep}} = \{c,c'\}$ (if $c' \in R$) or the unique node in $\centre{\unlabeledsubtree{\sep}} \setminus \{c\}$ is Steiner.

\medskip
\noindent
{\bf Correctness.}
In what follows, let $\tree$ be any well-structured $4$-Steiner root of $G$ such that $\unlabeledsubtree{\sep}$ is a bistar, $c \in \centre{\unlabeledsubtree{\sep}} \setminus \centre{\unlabeledsubtree{\maxk}}$, and $Real(N_{\tree}[c]) = R$.
Furthermore, let $\centre{\unlabeledsubtree{\sep}} = \{c,c_2\}$.
We divide the proof into the following claims:
\begin{myclaim}\label{claim:op2-before}
If there exists a clique-intersection $\ci \subset \sep$ such that $\ci \not\subseteq R$ and $\ci \cap (R \setminus \{c\}) \neq \emptyset$ then, $c' = c_2$.
\end{myclaim}
\begin{proofclaim}
As we have $\ci \not\subseteq R$ and $Real(N_T[c]) = R$, $\ci$ cannot contain any leaf node of $\unlabeledsubtree{\sep}$ adjacent to $c$.
Thus, the only possible node in $(\ci \cap R) \setminus \{c\}$ must be $c_2$.
\end{proofclaim}
\begin{myclaim}\label{claim:op2-before-b}
If there exists $c' \in R \setminus c$ that is $\maxk$-dependent then, $c' = c_2$.
\end{myclaim}
\begin{proofclaim}
Suppose for the sake of contradiction $c' \neq c_2$.
In particular, since $Real(N_{\tree}[c]) = R$, $c'$ is a leaf of $\unlabeledsubtree{\sep}$ that is adjacent to $c$.
Let $\ci_2 \in \CI{G}$ be such that $c'$ is $(\sep,\maxk,\ci_2)$-sandwiched and let $x \in (\maxk \cap \ci_2) \setminus \{c'\}$.
As already proved for Claim~\ref{claim:exclusive-dependency}, $dist_{\tree}(x,c_2) = 2$ and $dist_{\tree}(x,c) = 3$.
In particular we have $dist_{\tree}(c',x) = 4$.
Since we have $c',x \in \ci_2$, we must have $diam(\unlabeledsubtree{\ci_2}) = 4$ and the central node of $\unlabeledsubtree{\ci_2}$ must be the central node onto the $c'x$-path in $\tree$, that is $c_2$.
However, since $\maxk \neq \ci_2$ (because $\sep \not\subseteq \ci_2$) the latter contradicts Lemma~\ref{lem:no-center-intersect}.
\end{proofclaim}
%In order to prove correctness of this post-processing, we fix any well-structured $4$-Steiner root $\tree$ of $G$ such that $\unlabeledsubtree{\sep}$ is a bistar, $c \in \centre{\unlabeledsubtree{\sep}}$ and $Real(N_T[c]) = R$.
%Let $\centre{\unlabeledsubtree{\sep}} = \{c,c_2\}$.
%If either $c' = c_2$ or both $c_2,c'$ are Steiner nodes then, we are done.
%So, we assume from now on this is not the case.
%\begin{myclaim}\label{claim:op2-before}
%A clique-intersection $\ci$ as defined above cannot exist.
%\end{myclaim}
%\begin{proofclaim}
%Suppose by contradiction such a $\ci$ exists.
%As we have $\ci \not\subseteq R$ and $Real(N_T[c]) = R$, $\ci$ cannot contain any leaf node of $\unlabeledsubtree{\sep}$ adjacent to $c$.
%Thus, the only possible node in $(\ci \cap R) \setminus \{c\}$ must be $c_2$, which contradicts our assumption that $c_2 \neq c'$.
%\end{proofclaim}
%Claim~\ref{claim:op2-before} implies $c_2 \in R$ but $c'$ is Steiner. 
By Claims~\ref{claim:op2-before} and~\ref{claim:op2-before-b} we may assume from now on $c'$ to be Steiner.
Assume $c_2 \in R$, {\it i.e.}, $c_2$ is not Steiner (otherwise, we are done).
We now reuse a transformation that we introduced in the proof of Theorem~\ref{thm:x-free} ({\it i.e.}, Operation~\ref{op:move} for $v = c_2$ and $c$).
Specifically, we define $Q_{c_2}$ as the subtree of $\tree$ induced by the union of $c_2$ with all the components of $\tree \setminus V(\unlabeledsubtree{\sep})$ that are adjacent to $c_2$ {\em and} do not contain any real node at a distance $\leq 4$ from $\sep \setminus \{c_2\}$ in $\tree$.   
We create a new tree $\tree[2]$ by first removing $V(Q_{c_2}) \setminus \{c_2\}$, then replacing $c_2$ by the Steiner node $c'$, and finally adding a copy of $Q_{c_2}$ and the edge $cc_2$.
In doing so, we can only increase the distances in $\tree[2]$ compared to $\tree$, and these distances strictly increase at least between $c_2$ and all the leaves of $\unlabeledsubtree{\sep}$ to which it was previously adjacent.

\begin{figure}[h!]
\centering
\includegraphics[width=.5\textwidth]{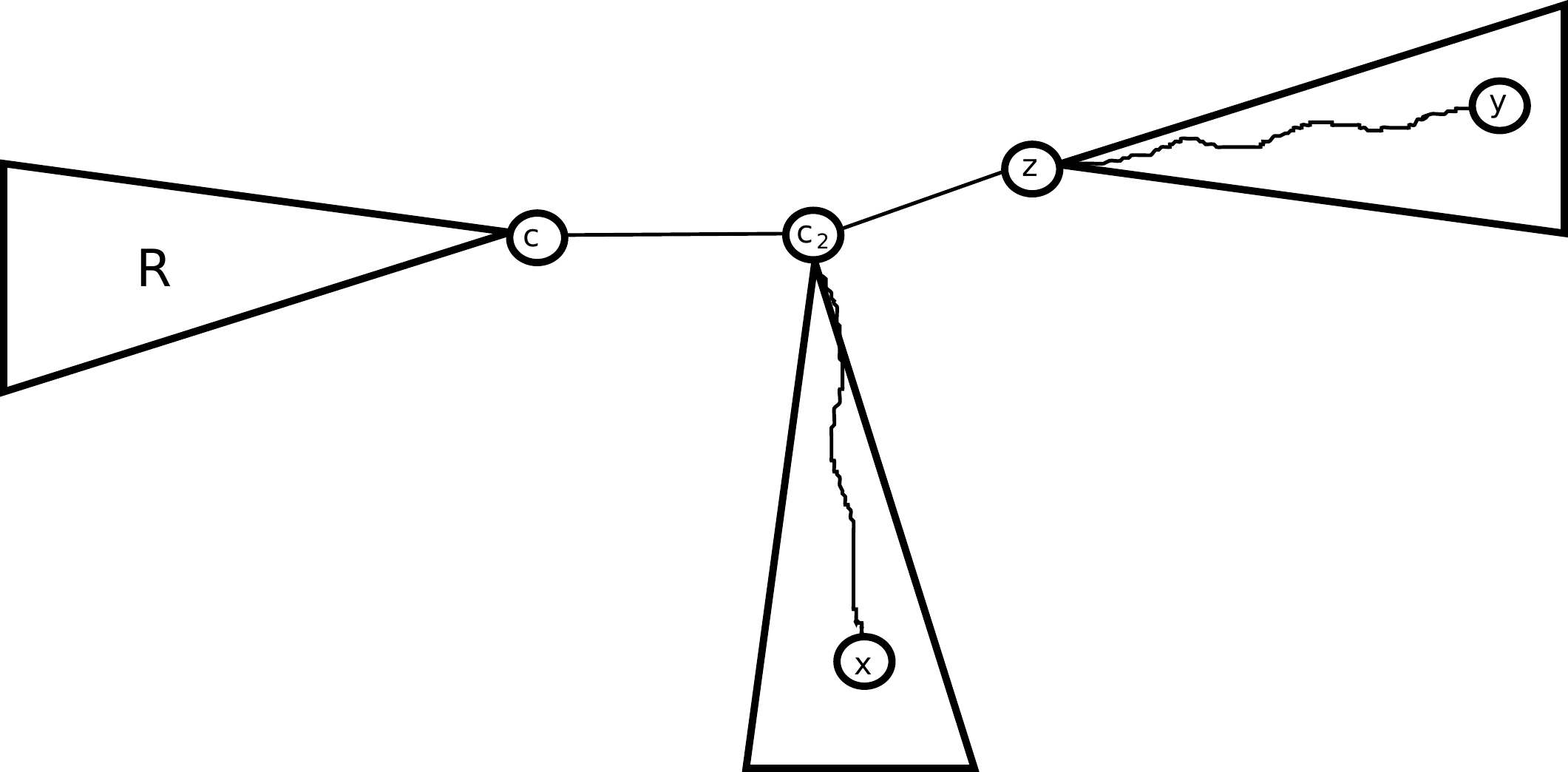}
\caption{To the proof of Lemma~\ref{lem:real-center-bistar}.}
\label{fig:canonical-bistar}
\end{figure}

\begin{myclaim}\label{claim:op-2}
$\tree[2]$ keeps the property of being a well-structured $4$-Steiner root of $G$
\end{myclaim}
Note that proving Claim~\ref{claim:op-2} will prove the lemma.

\begin{proofclaim}
This part of the proof closely follows Claim~\ref{claim:pties-op}.
First as our construction can only increase distances in the Steiner root, we can easily deduce that:
$$\forall uv \notin E, \ dist_{\tree[2]}(u,v) \geq dist_{\tree}(u,v) > 4.$$
However, we also need to check that conversely, $dist_{\tree[2]}(u,v) \leq 4$ for all $uv \in E$.
For that, we prove as an intermediate subclaim that:
$$dist_{\tree}(Real(Q_{c_2} \setminus \{c_2\}), V \setminus Q_{c_2}) > 4.$$
Indeed, suppose by contradiction there exist $x \in Real(Q_{c_2}) \setminus \{c_2\}, y \in V \setminus V(Q_{c_2})$ such that $dist_{\tree}(x,y) \leq 4$.
The unique $xy$-path in $\tree$ goes by $c_{2}$ (see Fig.~\ref{fig:canonical-bistar} for an illustration).
Furthermore by the definition of $Q_{c_2}$ we have $dist_{\tree}(x, \sep \setminus \{c_2\}) > 4$.
But then since $c_2 \in \centre{\unlabeledsubtree{\sep}}$ we have $dist_{\tree}(c_2,\sep \setminus \{c_2\}) =  1$, and so $dist_{\tree}(x,c_2) \geq 4$.
It implies $dist_{\tree}(x,y) > 4$, a contradiction.

This above subclaim implies that if $dist_{\tree[2]}(u,v) > 4$ for some $uv \in E$ then, $c_2 \in \{u,v\}$.
In order to prove that no such a pair $uv$ can exist, and that $\tree[2]$ is well-structured, we consider all the clique-intersections $\ci'$ that contain $c_2$.
There are three cases:
\begin{itemize}
\item Case $\ci' \subset \sep$. Then, either $\ci' = \{c_2\}$ and we are done, or $|\ci'| \geq 2$. Furthermore in the latter subcase we have $\ci' \subseteq R$ (otherwise, this would imply the existence of a $\ci$ as earlier defined, that would contradict Claim~\ref{claim:op2-before}). As a result we have $\unlabeledsubtree{\ci'} = \labeledsubtree[2]{}{\ci'}$.

\item Case $\ci' \supseteq \sep$. We can observe that $\unlabeledsubtree{\ci'} \cap Q_{c_2} = \{c_2\}$ since we proved above that we have $dist_{\tree}(Real(Q_{c_2} \setminus \{c_2\}), V \setminus Q_{c_2}) > 4$.
In particular, $\labeledsubtree[2]{}{\ci'}$ is obtained from $\unlabeledsubtree{\ci'}$ by replacing $c_2$ by a Steiner node (only if it were an internal node of $\unlabeledsubtree{\ci'}$) then, making of $c_2$ a leaf. 
Furthermore, since there is at least one leaf adjacent to each center node in $\unlabeledsubtree{S}$, we can prove as for Subclaim~\ref{subclaim:pties-op-b} that $diam(\labeledsubtree[2]{}{\ci'}) = diam(\unlabeledsubtree{\ci'})$.
As already observed in Claim~\ref{claim:pties-op}, it implies that all properties of Theorem~\ref{thm:x-free} are preserved provided that $c_2$ is not $\ci'$-free. 
So we only need to prove this is always the case. On one hand, $c_2$ is not $\sep$-free because $c_2$ is not a leaf of $\unlabeledsubtree{\sep}$ and we assume $\tree$ is well-structured.
On the other hand, $c_2$ is not $\ci'$-free for any $\ci' \supset \sep$, because $c_2 \in \sep $ and $|\sep| \geq 3$. 

\item Otherwise, in all other cases we prove $\labeledsubtree[2]{}{\ci'} = \unlabeledsubtree{\ci'}$.
To see that, first note that it may not be the case only if $\unlabeledsubtree{\ci'}$ is not fully contained into $Q_{c_2}$ (and so, $\unlabeledsubtree{\ci'} \cap Q_{c_2} = \{c_2\}$).

We prove as a subclaim that if $\labeledsubtree[2]{}{\ci'} \neq \unlabeledsubtree{\ci'}$ then, $\unlabeledsubtree{\ci'}$ must intersect $\unlabeledsubtree{\sep} \setminus \{c_2\}$.
Indeed, otherwise $\unlabeledsubtree{\ci'}$ must intersect some subtree $\tree[sub]$ of $\tree \setminus \unlabeledsubtree{\sep}$ that is not in $Q_{c_2}$.
Let $x \in \ci' \cap \tree[sub]$.
By the definition of $Q_{c_2}$, there exists $x' \in Real(\tree[sub])$ such that $dist_{\tree}(x', \sep \setminus \{c_2\}) \leq 4$ and there exists a $xx'$-path in $\tree$ that does not go by $c_2$.
In particular, $x'$ must be adjacent in $G$ to all the vertices in $\sep$ (otherwise, we could derive the existence of a clique-intersection $\ci$ as defined above, thereby contradicting Claim~\ref{claim:op2-before}).
Note that it implies $dist_{\tree}(x',c_2) \leq 2$ because the unique path in $\tree$ between $x'$ and any leaf adjacent to $c$ must go by $c_2$.
Then, $x' \in \maxk$.
Furthermore since $dist_{\tree}(x,c_2) \leq 4$ we obtain (by considering the median node of the triple $x,x',c_2$) $dist_{\tree}(x,x') \leq 4$.

Let $\maxk''$ be a maximal clique containing all of $x,x',c_2$.
There are two cases.
\begin{itemize}
\item Assume first $dist_{\tree}(x,\sep \setminus \{c_2\}) > 4$.
Then, $\sep \not\subseteq \maxk''$, and more specifically $\maxk'' \cap \sep = \{c_2\}$.
It implies $\maxk'' \neq \maxk$, and so $c_2$ is $(\sep,\maxk,\maxk'')$-sandwiched, that contradicts Claim~\ref{claim:op2-before-b}.
\item Otherwise, $dist_{\tree}(x,\sep \setminus \{c_2\}) \leq 4$ and we may further assume $x=x'$.
However, by the hypothesis $\sep \not\subseteq \ci'$, $\ci' \not\subseteq \sep$ resp., and so $\sep \cap \ci' = \{c_2\}$ (otherwise, $\ci = \sep \cap \ci'$ would falsify Claim~\ref{claim:op2-before}).
It implies $v$ is $(\sep,\maxk,\ci')$-sandwiched, that contradicts Claim~\ref{claim:op2-before-b}.
\end{itemize}
Therefore, we proved our subclaim that $\unlabeledsubtree{\ci'}$ must intersect $\unlabeledsubtree{\sep} \setminus \{c_2\}$.
In fact, our proof shows more generally that $\unlabeledsubtree{\ci'}$ cannot intersect any subtree of $\tree \setminus \unlabeledsubtree{\sep}$ with a node adjacent to $c_2$. 
Then, $\unlabeledsubtree{\ci'}$ must intersect a real node in $N_{\tree}(c_2) \setminus \{c\} = \sep \setminus R$ since otherwise, $\unlabeledsubtree{R} = \labeledsubtree[2]{}{R}$ would imply $\unlabeledsubtree{\ci'} = \labeledsubtree[2]{}{\ci'}$, a contradiction.
However in this situation, $\ci := \sep \cap \ci'$ satisfies $|\ci| \geq 2, c_2 \in \ci \ \text{and} \ \ci \not\subseteq R$, that contradicts Claim~\ref{claim:op2-before}.
\end{itemize}
\end{proofclaim}
The above case analysis ends up proving the claim, and so, the lemma.
\end{proof}

We will also need the following useful result which we keep using throughout most of the remaining proofs in this paper:

\begin{lemma}\label{lem:center-in-star}
Given $G=(V,E)$ and $\tree$ any $4$-Steiner root of $G$, let $\ci \in \CI{G}$ and let $\sep \subset \ci$ be a minimal separator.
If $\unlabeledsubtree{\sep}$ is a non-edge-star then, there exists $c \in N_{\tree}[\centre{\unlabeledsubtree{\ci}}]$ such that $Real(N_{\tree}[c]) = \sep$.
\end{lemma}

\begin{proof}
Write $\centre{\unlabeledsubtree{\sep}} = \{c\}$.
By Theorem~\ref{thm:clique-intersection}, $Real(N_{\tree}[c]) = \sep$.
Furthermore since by the hypothesis $\unlabeledsubtree{\sep}$ has at least two leaves then, the unique path between at least one such a leaf and $\centre{\unlabeledsubtree{\ci}}$ must pass by $c$.
Since $rad(\unlabeledsubtree{\ci}) \leq 2$, this implies $dist_{\tree}(c, \centre{\unlabeledsubtree{\ci}}) \leq 1$. 
\end{proof}

We now explain how to construct an important subfamily of ${\cal F}_i$:

\begin{lemma}\label{lem:with-bistar-sep}
Let $\maxk_i$ be a maximal clique of $G=(V,E)$ with no $\maxk_i$-free vertex.
In ${\cal O}(|\maxk_i|^6 \cdot n^3m\log{n})$-time, we can compute a family ${\cal B}_i$ with the following special property:
For any well-structured $4$-Steiner root $\tree$ of $G$ where for at least one minimal separator $\sep \subset \maxk_i$, $\unlabeledsubtree{\sep}$ is a bistar, there is a $\tree[2]$ such that $\unlabeledsubtree[2]{\sep_i} \equiv_G \unlabeledsubtree{\sep_i}$, $\unlabeledsubtree[2]{\maxk_i} \in {\cal B}_i$ and $dist_{\tree[2]}(r,V_i \setminus \sep_i) \geq dist_{\tree}(r,V_i \setminus \sep_i)$ for every $r \in V(\unlabeledsubtree{\sep_i})$.
\end{lemma}

Note that we do not capture {\em all} well-structured roots with this above lemma, but only those maximizing certain distances' conditions.

\begin{proof}
Let $\border{\maxk_i} \subseteq \SEP{G}$ contain all the minimal separators in $\maxk_i$.
By Theorem~\ref{thm:minsep}, for any $\sep \in \border{\maxk_i}$ we can construct a family ${\cal T}_{\sep}$ such that, in any $\labeledsubtree{\maxk_i}{} \in {\cal F}_i$, we should have $\labeledsubtree{\maxk_i}{\sep}$ is Steiner-equivalent to some tree in ${\cal T}_{\sep}$.
This takes total time ${\cal O}(|\maxk_i|^6|\border{\maxk_i}|) = {\cal O}(n|\maxk_i|^6)$. 
Fix $\sep \in \border{\maxk_i}$ (there are ${\cal O}(n)$ possibilities) and a bistar $\unlabeledsubtree{\sep} \in {\cal T}_{\sep}$ (by Theorem~\ref{thm:minsep}, there are ${\cal O}(|\sep|^5)$ possibilities, that is in ${\cal O}(|\maxk_i|^5)$). 
Note that in particular for $\sep_i \subseteq \sep$, this will also generate all possibilities for $\unlabeledsubtree{\sep_i}$.

Roughly we show that except in a few particular cases easy to solve, for every $\sep' \in \border{\maxk_i}$ there is only {\em one} canonical solution in ${\cal T}_{\sep'}$ that is compatible with $\unlabeledsubtree{\sep}$; moreover, this canonical solution can be computed in ${\cal O}(nm\log{n})$-time. 
For proving that, let us assume the existence of a $4$-Steiner root $\tree$ of $G$ that contains $\unlabeledsubtree{\sep}$ as a subtree.
We may only consider those $\sep' \in \border{\maxk_i}$ that are {\em not} contained into any other $\sep'' \in \border{\maxk_i}$.
Indeed if $\sep' \subseteq \sep''$ then, trivially $\unlabeledsubtree{\sep'}$ is forced by $\unlabeledsubtree{\sep''}$.
Thus from now on, we assume $\sep'$ is inclusion wise maximal in $\border{\maxk_i}$.

In what follows is a simple observation for the case $\sep \cap \sep' = \emptyset$ (see also Fig.~\ref{fig:bistar-not-intersect} for an illustration).

\begin{myclaim}\label{claim:no-empty-intersection}
If $\sep \cap \sep' = \emptyset$ then, $\unlabeledsubtree{\sep'}$ is a star with a Steiner central node.
Moreover, the central node of $\unlabeledsubtree{\maxk_i}$ must be Steiner.
\end{myclaim}

\begin{figure}[h!]
\centering
\includegraphics[width=.3\textwidth]{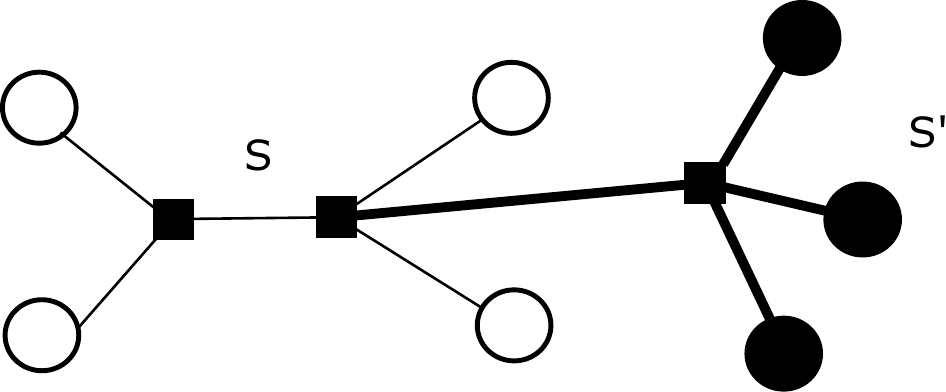}
\caption{An example where $\sep \cap \sep' = \emptyset$.}
\label{fig:bistar-not-intersect}
\end{figure}

\begin{proofclaim}
By Lemma~\ref{lem:bistar-center}, $\centre{\unlabeledsubtree{\maxk_i}} \subset \centre{\unlabeledsubtree{\sep}}$, and so by Theorem~\ref{thm:clique-intersection}, $Real(N_{\tree}[\centre{\unlabeledsubtree{\maxk_i}}]) \subset \sep$. 
But then, $\unlabeledsubtree{\maxk_i} \setminus N_{\tree}[\centre{\unlabeledsubtree{\maxk_i}}]$ is a collection of isolated leaves.
The latter proves either $\sep' = \{v\}$ is a cut-vertex or $\unlabeledsubtree{\sep'}$ is a star with a Steiner central node in $N_{\tree}[\centre{\unlabeledsubtree{\maxk_i}}]$.
In the former case we so conclude that $v$ is $\maxk_i$-free by inclusion wise maximality of $\sep'$ and by Lemma~\ref{lem:almost-simplicial-characterization}.
Since we assume there is no $\maxk_i$-free vertex, this case cannot happen.
Therefore, $\unlabeledsubtree{\sep'}$ is a star with a Steiner central node in $N_{\tree}[\centre{\unlabeledsubtree{\maxk_i}}]$.
Finally, the center of $\unlabeledsubtree{\maxk_i}$ must be also Steiner (otherwise by Lemma~\ref{lem:center-in-star}, this vertex should be in $\sep'$).
\end{proofclaim}
If $\sep \cap \sep' = \emptyset$ then, by combining Claim~\ref{claim:no-empty-intersection} and Lemma~\ref{lem:center-in-star} there is essentially one way to insert $\sep'$ in $\unlabeledsubtree{\maxk_i}$ ({\it i.e.}, we construct a star $\unlabeledsubtree{\sep'}$ with one Steiner central node, then we make this central node adjacent to the Steiner central node of $\unlabeledsubtree{\maxk_i}$).

For the remaining cases, we assume $\sep' \cap \sep \neq \emptyset$ for any inclusion wise maximal $\sep' \in \border{\maxk_i}$.
Several cases may arise: 

\begin{itemize}

\item \underline{Case there exist $u,v \in \sep \cap \sep'$ nonadjacent in $\unlabeledsubtree{\sep}$} (see Fig.~\ref{fig:bistar-not-adjacent} for an illustration). 
We prove $\unlabeledsubtree{\sep'}$ is a bistar.
Indeed, suppose by contradiction $\unlabeledsubtree{\sep'}$ is a star. 
Then, since $u,v \in \sep \cap \sep'$ are non adjacent, the center of $\unlabeledsubtree{\sep'}$ must be in $\centre{\unlabeledsubtree{\sep}}$. This implies $\sep' \subset \sep$, a contradiction. 
Therefore, we proved as claimed $\unlabeledsubtree{\sep'}$ must be a bistar.

In this situation we must have $\sep \cap \sep' = Real(N_{\tree}[\centre{\unlabeledsubtree{\maxk_i}}])$ ({\it e.g.}, see the proof of Lemma~\ref{lem:bistar-center}).
Since there exist $u,v \in \sep \cap \sep'$ nonadjacent in $\unlabeledsubtree{\sep}$, the central node of $\unlabeledsubtree{\maxk_i}$ can be uniquely defined as the central node $c_i \in \centre{\unlabeledsubtree{\sep}}$ such that $u,v \in N_{\tree}[c_i]$. 
Finally, since $Real(N_{\tree}[c_i])$ is fixed by $\unlabeledsubtree{\sep}$, by Lemma~\ref{lem:real-center-bistar} this leaves at most {\em one} canonical possibility for the second central node in $\centre{\unlabeledsubtree{\sep'}}$, and so, at most one possibility for $\unlabeledsubtree{\sep'}$.

\begin{figure}[h!]
\centering
\includegraphics[width=.3\textwidth]{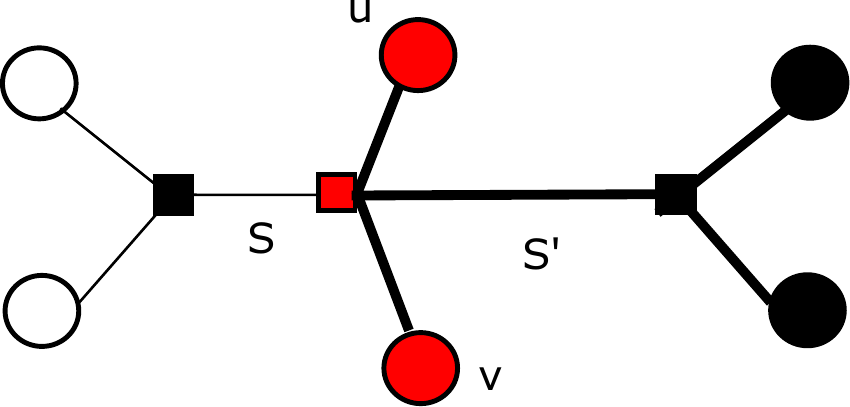}
\caption{Two bistars intersecting.}
\label{fig:bistar-not-adjacent}
\end{figure}

\smallskip
Note that we always fall in this case provided $|\sep \cap \sep'| \geq 3$.
So, we are left to study when $|\sep \cap \sep'| \in \{1,2\}$.

\item \underline{Case $\sep \cap \sep' = \{u,v\}$.}
We further assume $uv \in E(\unlabeledsubtree{\sep})$ (otherwise, we fall in the previous subcase).
Recall that we assume $\sep' \not\subseteq \sep$.
In particular, we must have $\sep \cap \sep' \subseteq N_{\tree}[\centre{\unlabeledsubtree{\maxk_i}}]$.
W.l.o.g., $u \in \centre{\unlabeledsubtree{\sep}}$ (or equivalently, $u$ must be the central node of $\unlabeledsubtree{\maxk_i}$).
Since $\sep' \not\subseteq \sep$ and there is at least one leaf-node of $\unlabeledsubtree{\sep}$ that is a real node adjacent to $u$, we so deduce that $v$ is a leaf of $\unlabeledsubtree{\sep}$ (otherwise, we would have $\sep \cap \sep' \neq \{u,v\}$).
Several situations force $\unlabeledsubtree{\sep'}$ to be a star, for instance if:
\begin{itemize}
\item $\sep'$ is strictly contained into another minimal separator of $G$;
%\item $S' \cap S'' = \emptyset$ for some $S'' \in \Omega(X_i)$ inclusion wise maximal;
\item $G \setminus \sep'$ has at least three full components\footnote{Equivalently, given any clique-tree $\cliquetree{G}$ of $G$ we have $|\edgeset{G}{\sep'}| \geq 2$.};
\item or $Real(N_{\tree}[u]) \neq \{u,v\}$.
\end{itemize} 
If such a situation occurs then, by Lemma~\ref{lem:center-in-star}, $v$ must be the center of the star $\unlabeledsubtree{\sep'}$, thereby leaving only one possibility for $\unlabeledsubtree{\sep'}$ ({\it i.e.}, see Fig.~\ref{fig:bistar-star}).

\begin{figure}[h!]
\centering
\includegraphics[width=.3\textwidth]{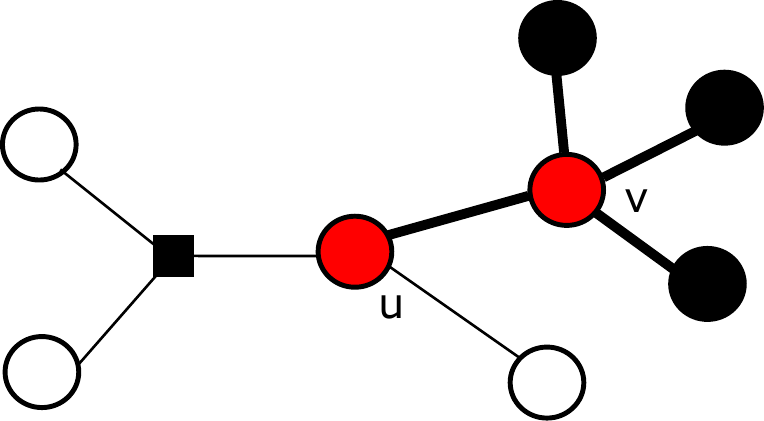}
\caption{Star intersecting a bistar.}
\label{fig:bistar-star}
\end{figure}

\smallskip
From now on assume that no minimal separator strictly contains $\sep'$, $G \setminus \sep'$ has exactly two full components and $Real(N_{\tree}[u]) = \{u,v\}$.
The subtree $\unlabeledsubtree{\sep'}$ is forced to be a bistar if there exists at least one $\sep'' \in \border{\maxk_i}$ inclusion wise maximal such that: $\sep \cap \sep'' = \{v\}$ (otherwise, $\unlabeledsubtree{\sep''}$ should be an edge and, since we assume $|\sep'| \geq 3$ this would imply $\sep'' \subseteq \sep'$ by Lemma~\ref{lem:center-in-star}, a contradiction).
Furthermore as explained in the previous case there is at most one canonical possibility for the bistar $\unlabeledsubtree{\sep'}$.

\smallskip
If no $\sep''$ as above exists then, $\unlabeledsubtree{\sep'}$ may be either a star or a bistar.
We can bipartition all the remaining minimal separators $\sep'' \in \border{\maxk_i}$ that are inclusion wise maximal (including $\sep'$) as follows: those containing $v$, and those that do not.
Note that in the former subcase (which includes $\sep'$) we have $\sep \cap \sep'' = \{u,v\}$, whereas in the latter subcase $\sep \cap \sep'' = \{u\}$.
Furthermore if $\sep \cap \sep'' = \{u\}$ then, $\unlabeledsubtree{\sep''}$ must always be a star with a Steiner central node that is adjacent to $u$ (to see that, recall that $Real(N_{\tree}[u]) = \{u,v\}$, and so, $\unlabeledsubtree{\sep''}$ cannot be a bistar).
In particular, there is only {\em one} possibility for such a $\sep''$.
However, the same as $\sep'$, for all other $\sep''$ such that $\sep \cap \sep'' = \{u,v\}$, $\unlabeledsubtree{\sep''}$ may be either a star or a bistar.
The key observation here is that $\unlabeledsubtree{\sep''}$ can be a star for at most {\em one} such a $\sep''$ (otherwise, by Lemma~\ref{lem:center-in-star} there would be two non-edge stars with the same center node $v$, that contradicts Property~\ref{pty-ci:2} of Theorem~\ref{thm:clique-intersection}).
Summarizing, since all these sets $\sep'' \setminus \sep$ are pairwise disjoint, we are left with ${\cal O}(|\maxk_i|)$ possibilities for the unique such $\sep''$ for which $\unlabeledsubtree{\sep''}$ is a star (if any); this choice fixes the corresponding subtree for all the remaining minimal separators. 
See Fig.~\ref{fig:bistar-two-choices-1}.

\begin{figure}[h!]
\centering
\includegraphics[width=.3\textwidth]{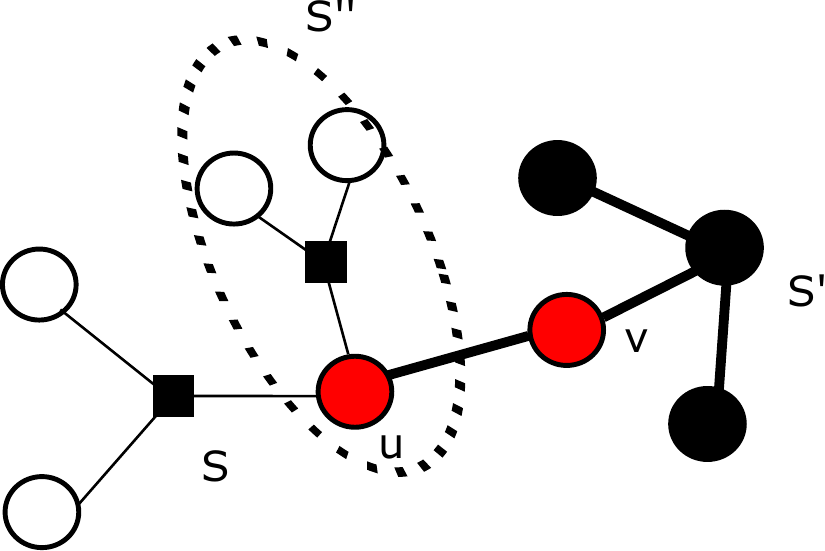}
\caption{A case where there are two possibilities for $\unlabeledsubtree{\sep'}$.}
\label{fig:bistar-two-choices-1}
\end{figure}

\item \underline{Case $\sep \cap \sep' = \{v\}$.}
If $v \in \centre{\unlabeledsubtree{\sep}}$ then, $v \in \centre{\unlabeledsubtree{\maxk_i}}$ because we assume $\sep' \not\subseteq \sep$. In particular, the only possibility for $\unlabeledsubtree{\sep'}$ is a star with a Steiner central node that is adjacent to $v$ (recall that $v$ is adjacent to at least one leaf in $\unlabeledsubtree{\sep}$, and so, $\unlabeledsubtree{\sep'}$ cannot be a bistar).
Assume for the remaining of the case $v$ is a leaf of $\unlabeledsubtree{\sep}$.
As in the previous case, several situations force $\unlabeledsubtree{\sep'}$ to be a star, like if:
\begin{itemize}
\item $\sep'$ is strictly contained into another minimal separator of $G$;
\item $G \setminus \sep'$ has at least three full components;
\item or $Real(N_{\tree}[\centre{\unlabeledsubtree{\maxk_i}}]) \neq \{v\}$.
\end{itemize}
Furthermore if such a situation occurs then, $v$ must be a center node of the star $\unlabeledsubtree{\sep'}$ (possibly, $\unlabeledsubtree{\sep'}$ is an edge), and so, there is only one possibility for $\unlabeledsubtree{\sep'}$.

\smallskip
From now on assume no minimal separator strictly contains $\sep'$, $G \setminus \sep'$ has exactly two full components and $Real(N_{\tree}[\centre{\unlabeledsubtree{\maxk_i}}]) = \{v\}$.
In particular, the unique central node of $\unlabeledsubtree{\maxk_i}$ is some Steiner node $\alpha_i$.
We also consider all the other minimal separators $\sep'' \in \border{\maxk_i}$ in the same situation as $\sep'$.
As we only consider inclusion wise maximal elements $\sep'' \in \border{\maxk_i}$ intersecting $\sep$, we must have $\sep \cap \sep'' = \{v\}$.
However, unlike the previous subcase, in an {\em arbitrary} well-structured $\tree$ there may be several such $\sep''$ for which $\unlabeledsubtree{\sep''}$ is a star.
We now prove that up to local modifications of $\tree$, we can always assume there is at most {\em one} such $\sep''$ for which $\unlabeledsubtree{\sep''}$ is a star.  
Note that by doing so, we can conclude as for the previous subcase about the number of possibilities for $\unlabeledsubtree{\maxk_i}$.

\begin{figure}[h!]
\centering
\includegraphics[width=.6\textwidth]{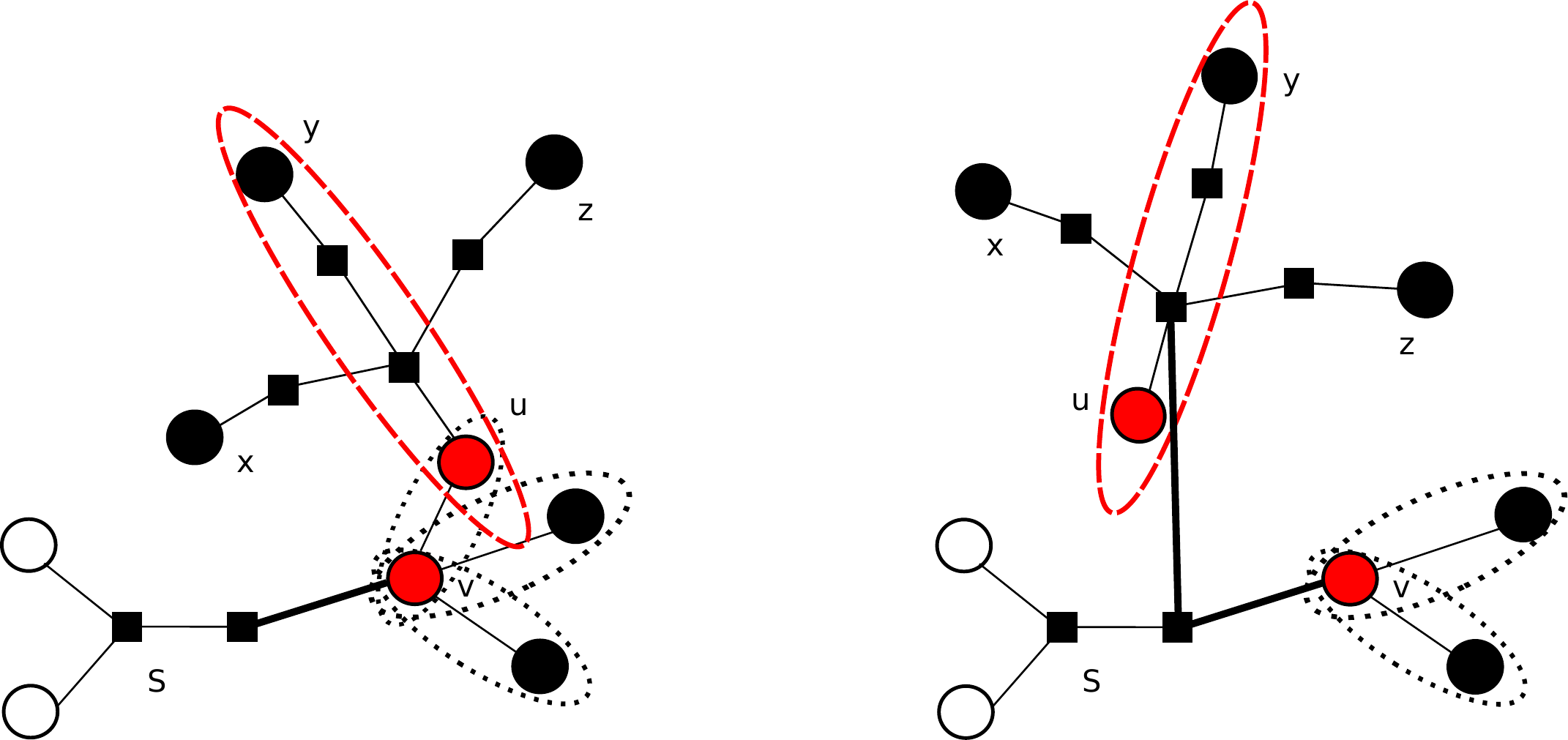}
\caption{The transformation of an edge into a bistar.}
\label{fig:bistar-two-choices-2}
\end{figure}

Assume there exist $\sep_{i_j},\sep_{i_k}$ such that $\unlabeledsubtree{\sep_{i_j}},\unlabeledsubtree{\sep_{i_k}}$ are stars.
Then, $\unlabeledsubtree{\sep_{i_j}},\unlabeledsubtree{\sep_{i_k}}$ must be edges with a common end $v$ (otherwise, one should be a non-edge star and so by Lemma~\ref{lem:center-in-star}, either $\sep_{i_j} \subset \sep_{i_k}$ or $\sep_{i_k} \subset \sep_{i_j}$, a contradiction). 
By inclusion wise maximality of $\sep_{i_j}$ and $\sep_{i_k}$, there is at least one of these two separators whose intersection with $\sep_i$ is either empty or reduced to $\{v\}$.
Assume w.l.o.g. this is the case for $\sep_{i_j}$ and write $\sep_{i_j} = \{u,v\}$.
We first gain more insights on the structure of $\unlabeledsubtree{\maxk_{i_j}}$.
For that, let $W_{i_j} := V_{i_j} \setminus \sep_{i_j}$.
Since $v$ has a neighbour in $V_i \setminus V_{i_j}$ we must have $dist_{\tree}(v,W_{i_j}) = 4$.
This implies $diam(\unlabeledsubtree{\maxk_{i_j}}) = 4$, $dist_{\tree}(u,W_{i_j}) = 3$, and all other real vertices of $\unlabeledsubtree{\maxk_{i_j}}$ must be leaves at distance two from $\centre{\unlabeledsubtree{\maxk_{i_j}}}$.
See Fig.~\ref{fig:bistar-two-choices-2} for an illustration.

We connect the unique node $\alpha_{i_j} \in \centre{\unlabeledsubtree{\maxk_{i_j}}}$ (which is Steiner) to $\alpha_i$ and then, we remove the edge $uv$.
In doing so, we obtain a tree $\tree[2]$ such that $Real(\tree[2]) = V$ and $\unlabeledsubtree[2]{\sep_{i_j}}$ is a bistar.
Note that in particular, $\unlabeledsubtree{\sep_i} = \unlabeledsubtree[2]{\sep_i}$, that follows from $\sep_i \cap \sep_{i_j} \subseteq \{v\}$.
Furthermore we claim that $\tree[2]$ keeps the property of being a $4$-Steiner root of $G$.
There are two cases:
\begin{itemize}
\item Suppose by contradiction $dist_{\tree[2]}(x,y) \leq 4$ for some $xy \notin E$. In particular, the unique $xy$-path in $\tree[2]$ must go by the edge $\alpha_{i_j}\alpha_i$. However since all neighbours of $\alpha_i$ except $v$ and all neighbours of $\alpha_{i_j}$ except $u$ are Steiner nodes, we obtain that $\{u,v\} \cap \{x,y\} \neq \emptyset$.
In particular if $x \in \{u,v\}$ then, $y \in \maxk_i \cup \maxk_{i_j}$, and so $xy \in E(G)$.
A contradiction.
\item Conversely, suppose by contradiction $dist_{\tree[2]}(x,y) > 4$ for some $xy \in E$. In particular, the unique $xy$-path in $\tree$ must go by the edge $uv$. Therefore, there must be a maximal clique $\maxk$ containing all of $x,y,u,v$. We have $\maxk \notin \{\maxk_i,\maxk_{i_j}\}$ since $dist_{\tree[2]}(x,y) > 4$. But then by maximality of $\sep_{i_j}$, $\maxk \cap \maxk_i = \maxk \cap \maxk_{i_j} = \maxk_i \cap \maxk_{i_j} = \sep_{i_j}$ and there are at least three full components in $G \setminus \sep_{i_j}$. A contradiction.
\end{itemize} 
So, we proved as claimed that $\tree[2]$ keeps the property of being a $4$-Steiner root of $G$.
We can prove in the same way that we have $\unlabeledsubtree{\ci} = \unlabeledsubtree[2]{\ci}$ for every clique-intersection $\ci \notin \{\sep_{i_j},\maxk_i,\maxk_{i_j}\}$, that implies $\tree[2]$ is well-structured.
We end up observing $dist_{\tree[2]}(r,W_i) \geq dist_{\tree}(r,W_i)$ for every $r \in \unlabeledsubtree{\sep_i}$ by construction, where $W_i := V_i \setminus \sep_i$.

Then, we obtain the desired property by repeating this above transformation until there is at most one $\sep''$ such that $\unlabeledsubtree{\sep''}$ is a star.
\end{itemize}
Overall, given a fixed $\unlabeledsubtree{\sep}$ we have at most ${\cal O}(|\maxk_i|)$ possibilities for $\unlabeledsubtree{\maxk_i}$.
By Lemma~\ref{lem:real-center-bistar}, every such a possibility can be computed in time ${\cal O}(|\border{\maxk_i}|) \times {\cal O}(nm\log{n})$, that is in ${\cal O}(n^2m\log{n})$.
Overall, the total running-time is in ${\cal O}(n|\maxk_i|^5) \times {\cal O}(|\maxk_i|) \times {\cal O}(n^2m\log{n}) = {\cal O}(|\maxk_i|^6 \cdot n^3m\log{n})$.
\end{proof}

\subsubsection{A degenerate case}\label{sec:degenerate}

If there is no minimal separator $\sep \subset \maxk_i$ such that $\unlabeledsubtree{\sep}$ is a bistar then, we get much less information on the structure of $\unlabeledsubtree{\maxk_i}$.
We identify the following as our main obstruction for bounding the number of possible subtrees:

\begin{definition}\label{def:thin-leg}
Given $G=(V,E)$ and $\tree$ a $4$-Steiner root of $G$, let $\maxk_i \in \MAXK{G}$ and let $\sep \subset \maxk_i$ be a minimal separator of size $|\sep| \geq 2$.
We call $\unlabeledsubtree{\sep}$ a {\em thin branch} of $\unlabeledsubtree{\maxk_i}$ if we have: 
\begin{itemize}
\item $\unlabeledsubtree{\sep} \setminus \centre{\unlabeledsubtree{\maxk_i}}$ is a connected component of $\unlabeledsubtree{\maxk_i} \setminus \centre{\unlabeledsubtree{\maxk_i}}$;
\item and there is no other $\unlabeledsubtree{\sep'}, \sep' \subset \maxk_i$ which both intersects $\centre{\unlabeledsubtree{\maxk_i}}$ and $\unlabeledsubtree{\sep} \setminus \centre{\unlabeledsubtree{\maxk_i}}$.
\end{itemize} 
%A {\em branch} of $\tree{X_i}$ is a component $C$ of $\tree{X_i} \setminus {\cal C}(\tree{X_i})$.
%It is called {\em thin} if there exists a minimal separator $S \subset X_i$ such that:
%\begin{itemize}
%\item $|S|=2$;
%\item no minimal separator in $X_i$ can strictly contain $S$;
%\item either $S = Real(C)$ or $S = Real(C) \cup {\cal C}(\tree{X_i})$.
%\end{itemize}  
The {\em head} of a thin branch is the node of $\unlabeledsubtree{\sep}$ that is the closest to $\centre{\unlabeledsubtree{\maxk_i}}$. 
\end{definition}

\begin{figure}[h!]
\centering
\includegraphics[width=.4\textwidth]{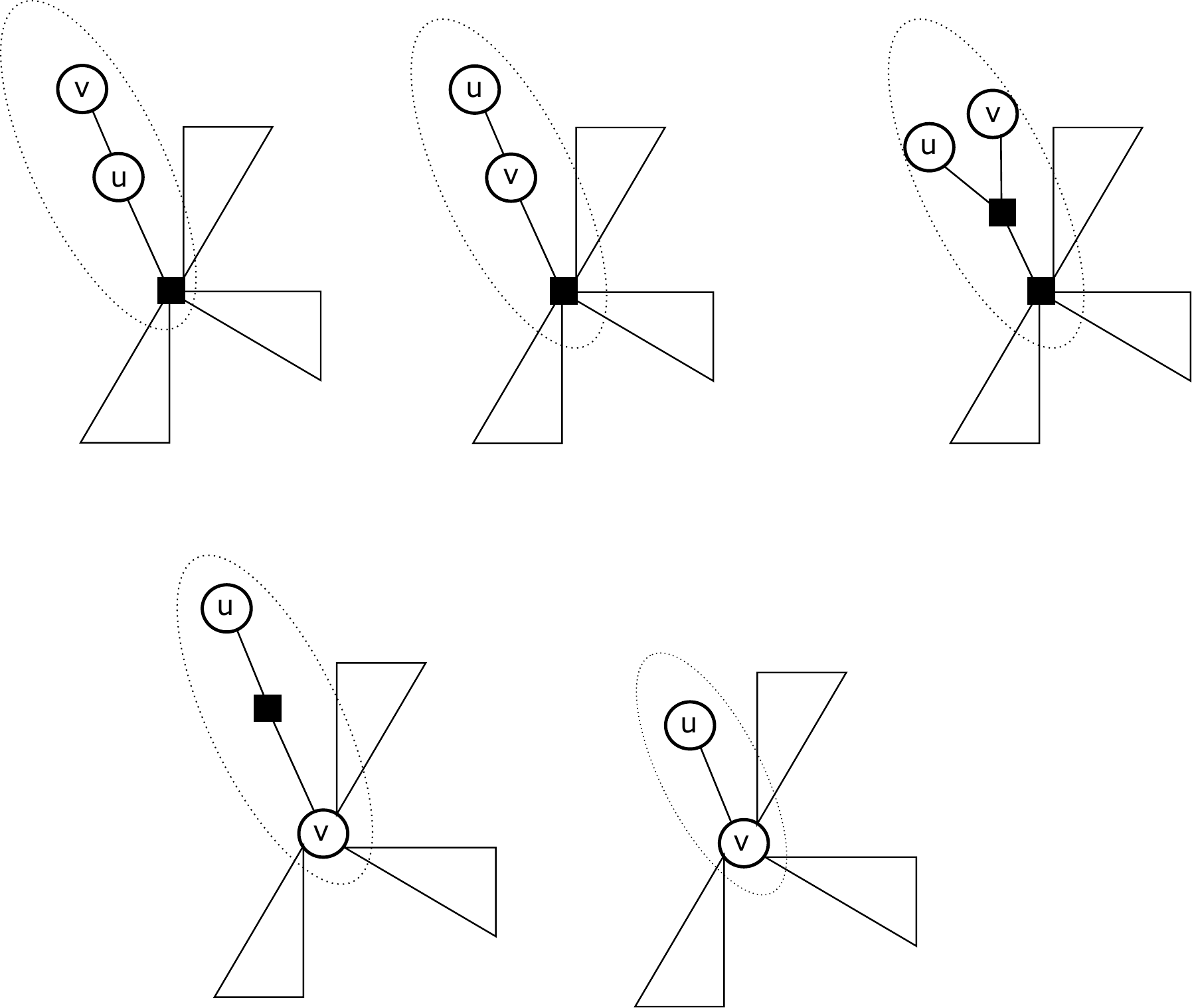}
\caption{Examples of thin branches (represented by a dashed ellipse).}
\label{fig:thin-branches}
\end{figure}

In order to understand the difficulties we met, assume on the way to construct $\unlabeledsubtree{\maxk_i}$ we correctly identified $\centre{\unlabeledsubtree{\maxk_i}}$ and a minimal separator $\sep$ for which $\unlabeledsubtree{\sep}$ must be a thin branch.
We can prove that $\unlabeledsubtree{\sep}$ must be a star (possibly, an edge).
However, without any additional information, there would be at least $|\sep| \geq 2$ possibilities for $\centre{\unlabeledsubtree{\sep}}$ ({\it e.g.}, see Figure~\ref{fig:thin-branches}).
If there are $p$ such minimal separators $\sep^1,\sep^2,\ldots,\sep^p$ for which $\unlabeledsubtree{\sep^j}$ must be a thin branch then, the number of possibilities for $\unlabeledsubtree{\maxk_i}$ goes up to $2^p$ at least.

\medskip
Intuitively, our choice for $\unlabeledsubtree{\sep^j}$ does not really matter as long as this does not violate any distance's constraints in the final solution we get. 
Guided by this intuition, we will sketch in Section~\ref{sec:greedy} a way to process all these $\sep^j$'s -- except maybe one -- independently from each other.
In particular, for now we do not really need to ``guess' what will be exactly $\unlabeledsubtree{\sep^j}$ in our final solution but just to correctly certify it has to be a thin branch.
Specifically, we prove the following result:

\begin{lemma}\label{lem:thin:branch}
Let $\maxk_i$ be a maximal clique of $G=(V,E)$ with no $\maxk_i$-free vertex.
There exists a family ${\cal D}_i$ that can be computed in ${\cal O}(n|\maxk_i|^4)$-time, and such that the following hold for any well-structured $4$-Steiner root $\tree$ of $G$:
\begin{enumerate}
\item
If $diam(\unlabeledsubtree{\maxk_i}) = 4$, and there is no minimal separator $\sep \subset \maxk_i$ such that $\unlabeledsubtree{\sep}$ is a bistar then, we have: $(\unlabeledsubtree[2]{Y_i \cup \centre{\unlabeledsubtree[2]{\maxk_i}}}, \ \centre{\unlabeledsubtree[2]{\maxk_i}}) \in {\cal D}_i$ for some $\tree[2] \equiv_G \tree$ and $Y_i \subseteq X_i$; 
\item
Moreover, $\sep_i \subseteq Y_i$, and for any $v \in \maxk_i \setminus Y_i$ there is a minimal separator $\sep \subseteq (\maxk_i \setminus Y_i) \cup \centre{\unlabeledsubtree{\maxk_i}}$ such that: $v \in \sep$, and $\unlabeledsubtree{\sep}$ is a thin branch.   
\end{enumerate}
\end{lemma}

\begin{proof}
By the hypothesis we are left to compute the diameter-four subtrees where, for every minimal separator $\sep \subset \maxk_i$, $\unlabeledsubtree{\sep}$ has diameter at most two.
For that, we only need to consider the subset ${\cal S}_i$ of all minimal separators $\sep \subset \maxk_i$ that are not strictly contained into any other minimal separator in $\maxk_i$.
Furthermore, we recall that by the hypothesis there is no $\maxk_i$-free vertex.
In particular, every $\sep \in {\cal S}_i$ has size at least two.
This implies $\unlabeledsubtree{\sep}$ must be either an edge or a star.
We now divide the proof into several cases:

\begin{itemize}
\item \underline{Case there is a $\sep \in {\cal S}_i$ such that $\centre{\unlabeledsubtree{\maxk_i}} = \centre{\unlabeledsubtree{\sep}}$.}
Note that in this case, $\unlabeledsubtree{\sep}$ must be a non-edge star.
Fix $\sep \in {\cal S}_i$ (there are ${\cal O}(n)$ possibilities) and one non-edge star $\unlabeledsubtree{\sep}$ such that $Real(\unlabeledsubtree{\sep}) = \sep$ (there are ${\cal O}(|\sep|) = {\cal O}(|\maxk_i|)$ possibilities).
By Lemma~\ref{lem:center-in-star}, $Real(N_{\tree}[\centre{\unlabeledsubtree{\sep}}]) = \sep$.
Therefore, there is at most one compatible solution for any other $\sep' \in {\cal S}_i$, namely: if $v \in \sep \cap \sep'$ is a leaf of $\unlabeledsubtree{\sep}$ then, $\unlabeledsubtree{\sep'}$ must be a star with $v$ as a central node (possibly, an edge); otherwise, $\unlabeledsubtree{\sep'}$ must be a non-edge star with a Steiner central node $\alpha \in N_{\tree}(\centre{\unlabeledsubtree{\sep}})$.
See Fig.~\ref{fig:ugly-case-2} for an illustration of that case.

\begin{figure}[h!]
\centering
\includegraphics[width=.3\textwidth]{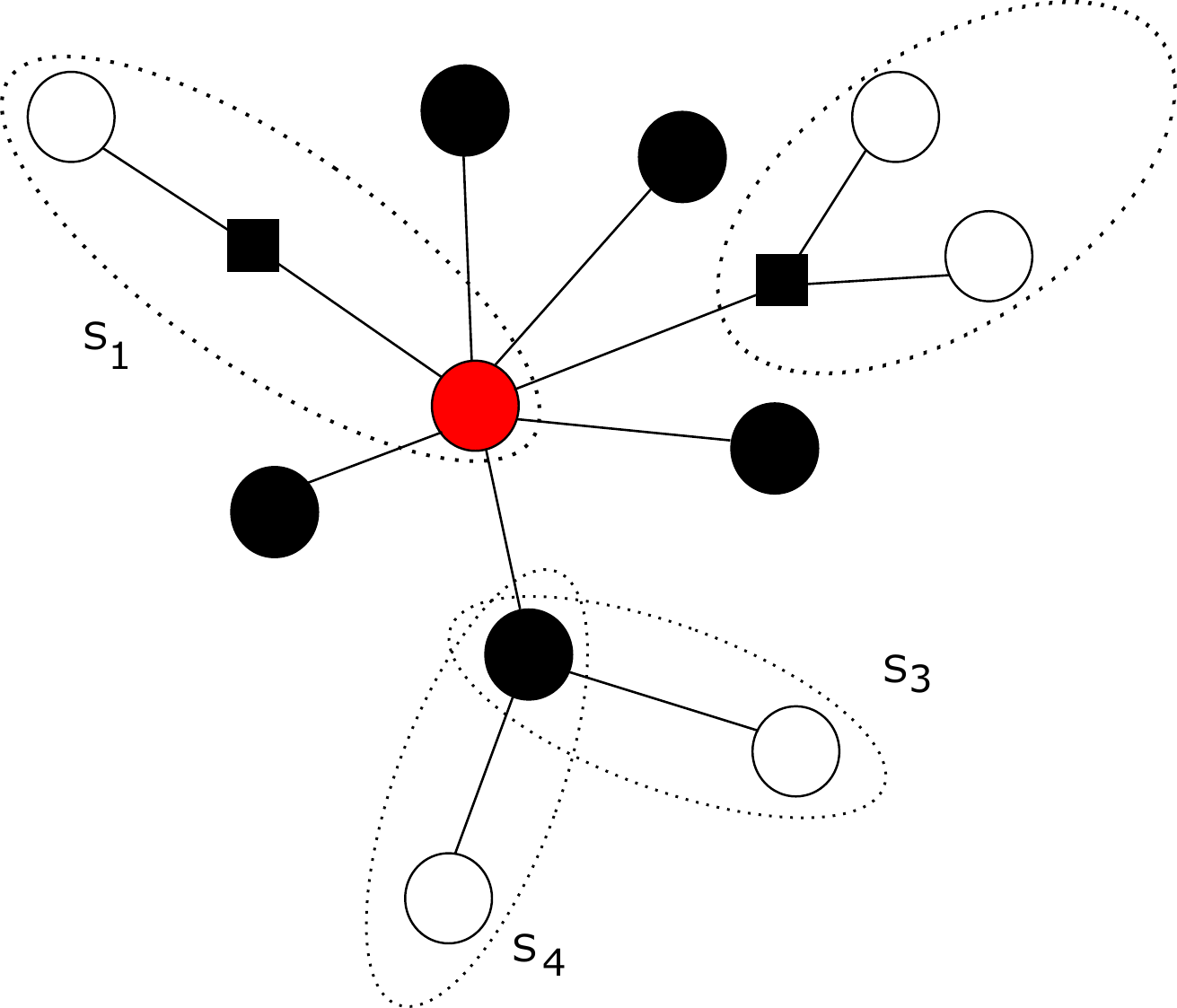}
\caption{Case 1 of Lemma~\ref{lem:thin:branch}.}
\label{fig:ugly-case-2}
\end{figure}

\item \underline{Case there is no $\sep \in {\cal S}_i$ such that $\centre{\unlabeledsubtree{\maxk_i}} = \centre{\unlabeledsubtree{\sep}}$.}
Fix $c_i \in \centre{\unlabeledsubtree{\maxk_i}}$ as being any vertex of $\maxk_i$ or Steiner (this gives ${\cal O}(|\maxk_i|)$ possibilities).
There are several subcases:
\begin{itemize}
\item
For $\sep \in {\cal S}_i$ of size $|\sep| \geq 3$, the only possibility for $\unlabeledsubtree{\sep}$ is to be a non-edge star such that (by Lemma~\ref{lem:center-in-star}) $Real(N_{\tree}[\centre{\unlabeledsubtree{\sep}}]) = \sep$, and the center of $\unlabeledsubtree{\sep}$ must be adjacent to $c_i$.
If in addition, there is a clique-intersection $\ci \subset \sep, |\ci| = 2$ and $c_i \in \ci$ then, the center of $\unlabeledsubtree{\sep}$ must be the unique vertex in $\ci \setminus \{c_i\}$.
Otherwise, we claim that $\unlabeledsubtree{\sep}$ must be a thin branch.
%as we assume that there can be no $\sep' \in {\cal S}_i$ such that: either $\unlabeledsubtree{\sep'}$ is a bistar or $\centre{\unlabeledsubtree{\maxk_i}} = \centre{\unlabeledsubtree{\sep'}}$, we obtain that $\sep$ is a thin branch. 
Indeed, suppose by contradiction the existence of a $\sep' \subset \maxk_i$ such that $\unlabeledsubtree{\sep'}$ intersects both $\centre{\unlabeledsubtree{\maxk_i}}$ and $\unlabeledsubtree{\sep} \setminus \centre{\unlabeledsubtree{\maxk_i}}$. 
As we assume $\sep \not\subseteq \sep'$, we should have $\unlabeledsubtree{\sep} \cap \unlabeledsubtree{\sep'}$ that is contained into the edge between $c_i$ and the central node of $\unlabeledsubtree{\sep}$. 
In particular, $\sep' \not\subseteq \sep$ (otherwise, this would contradict the non-existence of a clique-intersection $\ci$ as defined above).
This implies $\unlabeledsubtree{\sep'}$ should be either a bistar or a non-edge star such that $\centre{\unlabeledsubtree{\maxk_i}} = \centre{\unlabeledsubtree{\sep'}}$.
However, in both cases this would contradict our hypothesis that no such $\sep'$ exist.
Therefore, we proved as claimed that $\unlabeledsubtree{\sep}$ must be a thin branch.
\item
Let $\sep^1,\sep^2,\ldots,\sep^q$ be minimal separators of size exactly two that are pairwise intersecting into some vertex $u \neq c_i$.
Then, their union must be a star: where the center is the unique vertex $u$ in $\bigcap \sep^j$, and $Real(N_{\unlabeledsubtree{\maxk_i}}[u]) = \{c_i\} \cup \left( \bigcup_j \sep^j \right)$.
In particular, $\unlabeledsubtree{\sep^j}$ must be an edge for every $j$.
\item
So, the only remaining subcase is a minimal separator $\sep \in {\cal S}_i$ such that: $|\sep| = 2$, and the intersection of $\sep$ with any other minimal separator of ${\cal S}_i$ is either empty or reduced to $c_i$. 
Then, $\unlabeledsubtree{\sep}$ must be a thin branch. 
\end{itemize}
See Fig.~\ref{fig:ugly-case-1} for an illustration of these subcases.

\begin{figure}[h!]
\centering
\includegraphics[width=.3\textwidth]{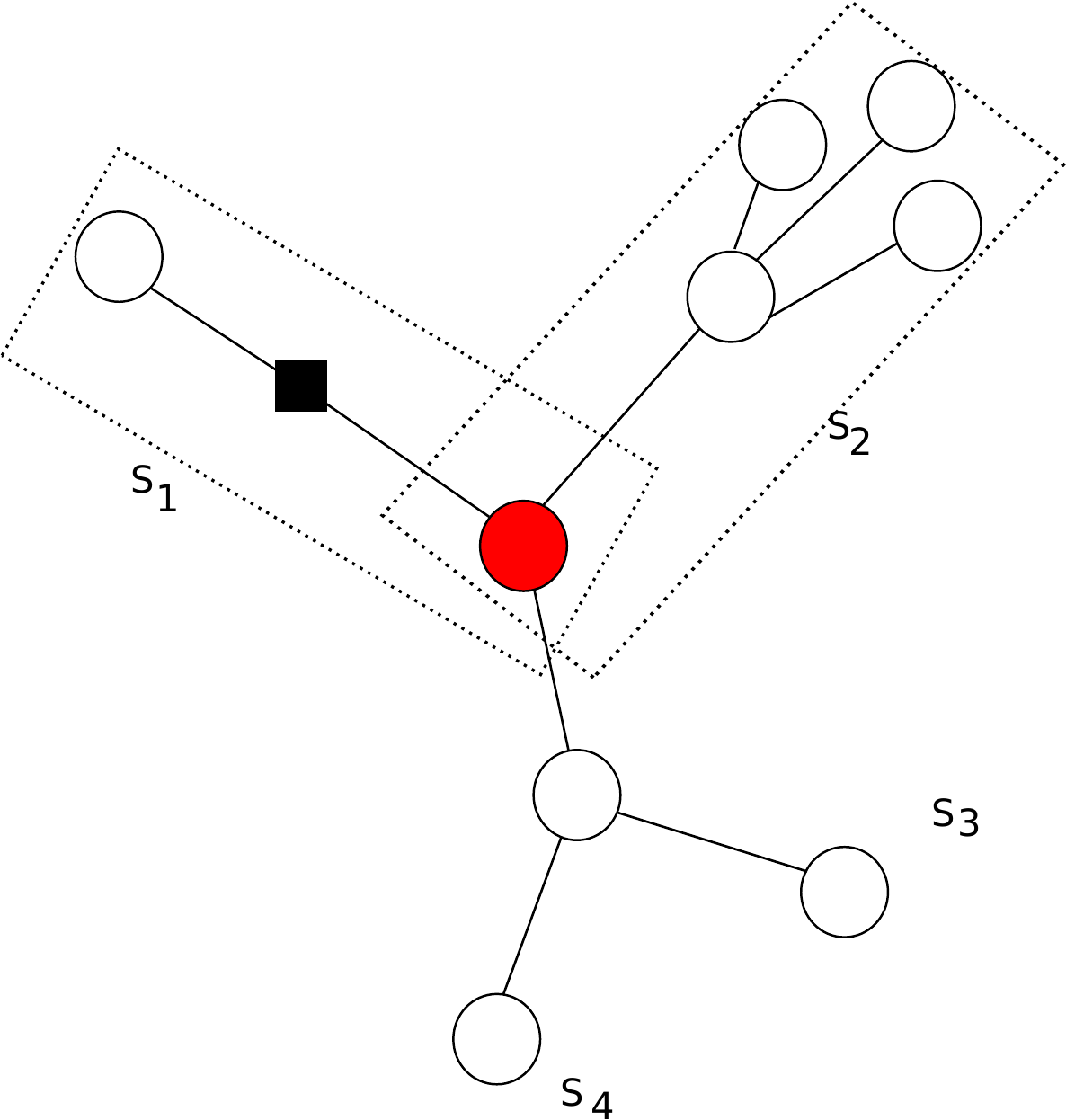}
\caption{Case 2 of Lemma~\ref{lem:thin:branch}. Thin branches are identified by dotted rectangles.}
\label{fig:ugly-case-1}
\end{figure}

\smallskip
Finally, according to Definition~\ref{def:thin-leg}, there may be at most one $\sep$ such that $\sep_i \cap (\sep \setminus \{c_i\}) \neq \emptyset$ and $\unlabeledsubtree{\sep}$ must be a thin branch.
Only for this $\sep$ we generate all possibilities for $\unlabeledsubtree{\sep}$, thereby generating ${\cal O}(|\sep|)$ different pairs $(\labeledsubtree{Y_i}{},c_i)$ to add in the family.
\end{itemize}
\end{proof}

\subsubsection{The polynomial-time computation}\label{sec:compute-fi}

Summarizing this section we get:

\begin{proposition}\label{prop:internal}
Let $\maxk_i$ be a maximal clique of $G=(V,E)$.
In ${\cal O}(|\maxk_i|^7 \cdot n^3m\log{n})$-time, we can compute a family ${\cal F}_i$ with the following special property.
For any well-structured $4$-Steiner root $\tree$ of $G$, there exists a $\tree[2]$ and a (not necessarily maximal) clique $Y_i \subseteq \maxk_i$ such that $(\unlabeledsubtree[2]{Y_i \cup \centre{\unlabeledsubtree[2]{\maxk_i}}}, \ \centre{\unlabeledsubtree[2]{\maxk_i}}) \in {\cal F}_i$, and we have:
\begin{itemize}
\item $\sep_i \subseteq Y_i$ and $\unlabeledsubtree[2]{\sep_i} \equiv_G \unlabeledsubtree{\sep_i}$;
\item $dist_{\tree[2]}(r,V_i \setminus \sep_i) \geq dist_{\tree}(r,V_i \setminus \sep_i)$ for any $r \in V(\unlabeledsubtree{\sep_i})$;
\item For any $v \in \maxk_i \setminus Y_i$ there is a minimal separator $\sep \subseteq (\maxk_i \setminus Y_i) \cup \centre{\unlabeledsubtree{\maxk_i}}$ such that: $v \in \sep$ and $\unlabeledsubtree{\sep}$ is a thin branch. 
\end{itemize}
Moreover, $Y_i = \maxk_i$ if either $diam(\unlabeledsubtree[2]{\maxk_i}) < 4$ or there exists a minimal separator $\sep \subset \maxk_i$ such that $\unlabeledsubtree[2]{\sep}$ is a bistar.
\end{proposition}

\begin{proof}
If there are $\maxk_i$-free vertices then, by using Lemma~\ref{lem:almost-simplicial-placement}, there is essentially one canonical way to add these vertices at the end of the construction.
For that, it suffices to fix the center of $\unlabeledsubtree{\maxk_i}$ which, as explained in the proof of Theorem~\ref{thm:leaf}, can only increase the total runtime by a multiplicative factor in ${\cal O}(|\maxk_i|)$. 
Thus from now on we may assume that $\maxk_i$ has no $\maxk_i$-free vertex.
Furthermore, we may also assume we already computed all the diameter-three subtrees to add in ${\cal F}_i$ ({\it i.e.}, see Remark~\ref{rk:extension-minsep}).
We explain in Lemma~\ref{lem:with-bistar-sep} how to compute the subfamily ${\cal B}_i$ of all diameter-four subtrees where at least one minimal separator $\sep \subset \maxk_i$ has $\unlabeledsubtree{\sep}$ being a bistar.
Finally, Lemma~\ref{lem:thin:branch} completes the construction of the family ${\cal F}_i$.
\end{proof}

\section{Step~\ref{step-3}: Deciding the partial solutions to store}\label{sec:encoding}

The next two Sections are devoted to the main loop of our algorithm proving Theorem~\ref{thm:main-steiner-power}.
Specifically, let $\cliquetree{G}$ be the rooted clique-tree we defined in Section~\ref{sec:clique-tree}.
We consider all the maximal cliques $\maxk_i$ that are {\em internal nodes} in $\cliquetree{G}$, starting from the twigs ({\it a.k.a.}, internal nodes whose all children nodes are leaves).
For every such a $\maxk_i$ we must execute a same procedure: that corresponds to Step~\ref{step-3} of the algorithm, and is the main focus of this section.

During this procedure, we activate each child node $\maxk_{i_j}$ of $\maxk_i$ sequentially, by sending a message.
This activation message contains a series of constraints on the $4$-Steiner roots of $G_{i_j}$ that we want to compute.
Overall, after it has been activated, we complete Step~\ref{step-4} for $\maxk_{i_j}$ (detailed in the next Section~\ref{sec:greedy}) using the series of constraints it has received from its father node $\maxk_i$. 
The output is a set ${\cal T}_{i_j}$ of partial solutions whose size is polynomial in $|\sep_{i_j}|$.
Finally, after we received ${\cal T}_{i_j}$, we can compute and send the activation message for the next child of $\maxk_i$ to be activated (if any).
%In what follows, let $\maxk_{i_j}$ be a fixed child of $\maxk_i$ in $\cliquetree{G}$.
%Recall that in the next Step, we will compute a subset ${\cal T}_{i_j}$ of $4$-Steiner roots for $G_{i_j}$.
%During this current Step, we compute a series of ``indications'' to be transmitted to $X_{i_j}$ in order to enforce the number of partial solutions that we will store in ${\cal T}_{i_j}$ to stay polynomial in $|S_{i_j}|$.
%For that, we introduce the following problem:

As a way to formalize the constraints that we need to include in activation messages, we now introduce the following problem:

\begin{center}
	\fbox{
		\begin{minipage}{.95\linewidth}
			\begin{problem}[\textsc{Distance-Constrained Root}]\
				\label{prob:distances-constraints} 
					\begin{description}
					\item[Input:] a graph $G=(V,E)$ and a rooted clique-tree $\cliquetree{G}$, a maximal clique $\maxk_{i_j}$, a tree $\labeledsubtree{\sep_{i_j}}{}$ s.t. $Real(\labeledsubtree{\sep_{i_j}}{}) = \sep_{i_j}$, and a sequence $(d_r)_{r \in V(\labeledsubtree{\sep_{i_j}}{})}$ of positive integers.
					\item[Output:] Either a $4$-Steiner root $\labeledsubtree{i_j}{}$ of $G_{i_j}$ s.t. $\labeledsubtree{\sep_{i_j}}{} \equiv_G \labeledsubtree{i_j}{\sep_{i_j}}$ and, $\forall r \in V(\labeledsubtree{\sep_{i_j}}{})$: $dist_{\labeledsubtree{i_j}{}}(r, V_{i_j} \setminus \sep_{i_j}) \geq d_r$; Or $\bot$ if there is no such a partial solution which can be extended to some well-structured $4$-Steiner root $\tree$ of $G$.
				\end{description}
			\end{problem}     
		\end{minipage}
	}
\end{center}

\begin{theorem}\label{thm:encoding}
Given $G=(V,E)$ chordal and a rooted clique-tree $\cliquetree{G}$ as in Theorem~\ref{thm:final-clique-tree}, let $\maxk_i$ be an internal node with children $\maxk_{i_1},\maxk_{i_2},\ldots,\maxk_{i_p}$.
If we can solve \textsc{Distance-Constrained Root} in time $P(n,|\sep_{i_j}|)$ for some polynomial $P$ then, we can compute in time ${\cal O}(n|\maxk_i|^5P(n,|\maxk_i|))$ a family ${\cal T}_{i_1}, {\cal T}_{i_2}, \ldots, {\cal T}_{i_p}$ of $4$-Steiner roots for $G_{i_1}, G_{i_2}, \ldots, G_{i_p}$, respectively, such that:
\begin{enumerate}
\item For any $j \in \{1,2,\ldots,p\}$, $|{\cal T}_{i_j}| = {\cal O}(|\sep_{i_j}|^5)$;
\item For any well-structured $4$-Steiner root $\tree$ of $G$, there exists a $\tree[2]$ such that: $\unlabeledsubtree{\maxk_i} \equiv_G \unlabeledsubtree[2]{\maxk_i}$, $\unlabeledsubtree[2]{V_{i_j}} \in {\cal T}_{i_j}$ for any $j \in \{1,2,\ldots,p\}$, and (only if $\maxk_i \neq \maxk_0$) $dist_{\tree[2]}(r,V_i \setminus \sep_i) \geq dist_{\tree}(r,V_i \setminus \sep_i)$ for any node $r \in V(\unlabeledsubtree{\sep_i})$.
\end{enumerate}
\end{theorem}

We postpone the proof that we can solve {\sc Distance-Constrained Root} in polynomial time to Section~\ref{sec:greedy}.
This above result can be seen as a pre-processing phase for $\maxk_i$, that is crucial in order to bound the runtime of our algorithm by a polynomial.
Note that the technical condition on the nodes in $\unlabeledsubtree{\sep_i}$ is simply there to ensure that when later in the algorithm, we will need to solve \textsc{Distance-Constrained Root} at $\maxk_i$, we cannot miss a solution.

\medskip
The remaining of this subsection is now devoted to the proof of Theorem~\ref{thm:encoding}.
We will use some additional terminology that we define next:
\begin{definition}\label{def:op}
Given $G=(V,E)$, let $A,B,\sep \subset V$ satisfy $A \cup B = V$ and $A \cap B = \sep$.
Two trees $\labeledsubtree{A}{},\labeledsubtree{B}{}$, where $Real(\labeledsubtree{A}{}) = A$ and $Real(\labeledsubtree{B}{}) = B$, are {\em compatible} if $\labeledsubtree{A}{\sep} \equiv_G \labeledsubtree{B}{\sep}$.
Then, $\labeledsubtree{A}{} \odot \labeledsubtree{B}{}$ is the tree obtained from $\labeledsubtree{A}{},\labeledsubtree{B}{}$ by the identification of $\labeledsubtree{A}{\sep}$ with $\labeledsubtree{B}{\sep}$. 

\smallskip
In particular, assume $G$ to be chordal and let $\cliquetree{G}$ be a rooted clique-tree of $G$.
For any $\maxk_i \in \MAXK{G}$, let $\sep_i := \maxk_i \cap \maxk_{p(i)}$, let $V_i := V(G_i)$ and let $W_i := V_i \setminus \sep_i$.
Given $\tree,\tree[2]$ $4$-Steiner roots of $G$, we say that $\tree[2]$ is {\em $i$-congruent} to $\tree$ if $\tree[2] \equiv_G \unlabeledsubtree{V \setminus W_i} \odot \labeledsubtree[2]{i}{}$, for some $4$-Steiner root $\labeledsubtree[2]{i}{}$ of $G_i$.
\end{definition}

Note that in particular, any two Steiner roots of $G$ are trivially $0$-congruent ({\it i.e.}, assuming $\sep_0 = \emptyset$ by convention). 
Finally in what follows we also use $d_{\labeledsubtree{i_j}{}}(r)$ as a shorthand for $dist_{\labeledsubtree{i_j}{}}(r,W_{i_j})$.
We observe that for any $r \in V(\unlabeledsubtree{\sep_{i_j}})$ we have $dist_{\labeledsubtree{i_j}{}}(r,W_{i_j}) \leq dist_{\labeledsubtree{i_j}{}}(r,\maxk_{i_j} \setminus \sep_{i_j}) \leq 4$.

\paragraph{Outline of the proof.}
We process the children nodes $\maxk_{i_j}$ sequentially by increasing size of the minimal separators $\sep_{i_j}$.
For that, we start constructing the family ${\cal T}_{\sep_{i_j}}$ of Theorem~\ref{thm:minsep}, and we consider the subtrees $\labeledsubtree{\sep_{i_j}}{} \in  {\cal T}_{\sep_{i_j}}$ sequentially.
We divide the proof into several cases depending on $|\sep_{i_j}|$ and on $diam(\labeledsubtree{\sep_{i_j}}{})$.
\begin{itemize}
\item If $|\sep_{i_j}| \leq 2$ then, there can only be ${\cal O}(1)$ different possibilities for the pair $\labeledsubtree{\sep_{i_j}}{}, (d_r)_{r \in V(\labeledsubtree{\sep_{i_j}}{})}$.
We can solve \textsc{Distance-Constrained Root} for all these possibilities, thereby obtaining the family ${\cal T}_{i_j}$.
However, for some reasons that will become clearer in Section~\ref{sec:greedy}, we only keep in ${\cal T}_{i_j}$ the solutions which satisfy some local optimality criteria.
See Section~\ref{sec:small}. 

\item The processing of the minimal separators $\sep_{i_j}$ with at least three elements is more intricate (Sections~\ref{sec:star} and~\ref{sec:bistar}).
For a fixed $\labeledsubtree{\sep_{i_j}}{}$ we define an encoding with only $|\sep_{i_j}|^{{\cal O}(1)}$ possibilities, that essentially summarizes at ``guessing'' the central nodes of $\unlabeledsubtree{\maxk_i}$ and $\unlabeledsubtree{\maxk_{i_j}}$.
Then, we show that only one solution per possible encoding needs to be stored in ${\cal T}_{i_j}$.
The correctness of this part crucially depends on some additional distances' constraints that are derived from the smaller separators contained into $\sep_{i_j}$, and on Theorem~\ref{thm:final-clique-tree}.
Indeed, our approach could not work with an arbitrary rooted clique-tree.
\end{itemize}

\subsection{Case $diam(\labeledsubtree{\sep_{i_j}}{}) \leq 1$.}\label{sec:small}

In this situation, $|\sep_{i_j}| \leq 2$ and so, there can only be ${\cal O}(1)$ possibilities for the distances' constraints $(d_r)_{r \in \sep_{i_j}}$.
We can solve \textsc{Distance-Constrained Root} for all possible values, thereby obtaining the family ${\cal T}_{i_j}$.
However, for reasons which will become clearer in the remaining of this Section and in the proofs of Section~\ref{sec:greedy}, storing all these possibilities would increase the runtime of our algorithm.
We confront this issue with a local optimality criterion.
Specifically: 

\begin{myclaim}\label{claim:cut-vertex}
Assume $\sep_{i_j} = \{v\}$ and let $\labeledsubtree[M]{i_j}{} \in {\cal T}_{i_j}$ maximize $d_{\labeledsubtree[M]{i_j}{}}(v)$.
If $\tree$ is a $4$-Steiner root of $G$ and $\unlabeledsubtree{V_{i_j}} \in {\cal T}_{i_j}$ then, $\unlabeledsubtree{V \setminus W_{i_j}} \odot \labeledsubtree[M]{i_j}{}$ is also a $4$-Steiner root of $G$.
\end{myclaim}\qedclaim

By Claim~\ref{claim:cut-vertex}, we so conclude that if $\sep_{i_j}$ is a cut-vertex then, we can keep exactly one solution in the family ${\cal T}_{i_j}$.

\begin{myclaim}\label{claim:edge}
Assume $\sep_{i_j} = \{u,v\}$.
Let $\tree$ be a $4$-Steiner root of $G$ such that $\unlabeledsubtree{\sep_{i_j}}$ is an edge and $dist_{\tree}(v, W_{i_j}) \geq dist_{\tree}(u, W_{i_j})$.
Then, $\tree[2] := \unlabeledsubtree{V \setminus W_{i_j}} \odot \labeledsubtree[v]{i_j}{}$ is also a $4$-Steiner root of $G$, where $\labeledsubtree[v]{i_j}{} \in {\cal T}_{i_j}$ is, among all solutions in this set such that $\labeledsubtree{i_j}{\sep_{i_j}}$ is an edge and $d_{\labeledsubtree{i_j}{}}(v)$ is maximized, one maximizing $d_{\labeledsubtree{i_j}{}}(u)$.
Moreover, $dist_{\tree[2]}(v, W_{i_j}) \geq dist_{\tree}(v, W_{i_j})$ and $dist_{\tree[2]}(u, W_{i_j}) \geq dist_{\tree}(u, W_{i_j})$.
\end{myclaim}

\begin{proofclaim}
It suffices to prove $dist_{\tree[2]}(u, W_{i_j}) \geq dist_{\tree}(u, W_{i_j})$.
We first observe $d_{\labeledsubtree[v]{i_j}{}}(v) - d_{\labeledsubtree[v]{i_j}{}}(u) \leq 1$ because $\labeledsubtree[v]{i_j}{\sep_{i_j}}$ is an edge.
Therefore, there are two cases. 
Either $dist_{\tree}(v, W_{i_j}) = d_{\labeledsubtree[v]{i_j}{}}(v)$ was already maximized, and so we have $d_{\labeledsubtree[v]{i_j}{}}(u) \geq dist_{\tree}(u, W_{i_j})$. Or $d_{\labeledsubtree[v]{i_j}{}}(v) \geq dist_{\tree}(v, W_{i_j}) + 1$, and so $d_{\labeledsubtree[v]{i_j}{}}(u) \geq d_{\labeledsubtree[v]{i_j}{}}(v) -1 \geq dist_{\tree}(v, W_{i_j}) \geq dist_{\tree}(u, W_{i_j})$.  
\end{proofclaim}

By Claim~\ref{claim:edge}, if $\sep_{i_j} = \{u,v\}$ and $\labeledsubtree{\sep_{i_j}}{}$ is an edge then, we only need to keep two solutions, namely: among all those maximizing $d_{\labeledsubtree{i_j}{}}(v)$ ($d_{\labeledsubtree{i_j}{}}(u)$, resp.) the one maximizing $d_{\labeledsubtree{i_j}{}}(u)$ ($d_{\labeledsubtree{i_j}{}}(v)$, resp.). 

\subsection{Case $\labeledsubtree{\sep_{i_j}}{}$ is a non-edge star.}\label{sec:star}
If $|\sep_{i_j}| = 2$ then, as already observed in Section~\ref{sec:small}, there can only be ${\cal O}(1)$ different possibilities for the constraints.
We can solve \textsc{Distance-Constrained Root} for all possible values, thereby obtaining the family ${\cal T}_{i_j}$.
Thus from now on we assume $|\sep_{i_j}| \geq 3$.

\smallskip
Although there may be exponentially many possible sets of constraints in this case, we show that only a few of the distances' constraints we impose truly need to be considered by our algorithm.
Specifically, write $\centre{\labeledsubtree{\sep_{i_j}}{}} = \{c\}$.
We consider all possible pairs of nodes $c_i,c_{i_j}$ that are either in $\sep_{i_j}$ or Steiner.
For every such a fixed pair, we are interested in the existence of (well-structured) $4$-Steiner roots $\tree$ of $G$ such that: $c_i \in \centre{\unlabeledsubtree{\maxk_i}} \cap N_{\tree}[c]$ and $c_{i_j} \in \centre{\unlabeledsubtree{\maxk_{i_j}}} \cap N_{\tree}[c]$. 
In order for our algorithm to decide whether such Steiner roots exist, we will prove that we only need to consider the distances between ${\cal O}(1)$ nodes in $\labeledsubtree{\sep_{i_j}}{}$ and $W_{i_j}$.
-- In particular, we will prove that we only need to store ${\cal O}(1)$ partial solutions in ${\cal T}_{i_j}$. --
On our way to prove this result, we also use various properties of $4$-Steiner powers in order to impose additional distances' constraints on the solutions in ${\cal T}_{i_j}$ which we prove to be necessary in order to extend such a partial solution to all of $G$.
This second phase is crucial in proving correctness of our approach.

Overall, by Lemma~\ref{lem:center-in-star}, we can always relate {\em any} $4$-Steiner root of $G$ with a pair $c_i,c_{i_j}$ as defined above.
Since there are ${\cal O}(|\sep_{i_j}|)$ possibilities for every of $\labeledsubtree{\sep_{i_j}}{}, c_i,c_{i_j}$, we are left with only ${\cal O}(|\sep_{i_j}|^3)$ different possibilities to store for stars.

\medskip
Recall that $\cliquetree{G}$ is a rooted clique-tree of $G$ as stated in Theorem~\ref{thm:final-clique-tree}.
Before starting our analysis, we need to derive a few properties from $\cliquetree{G}$.
Indeed, our approach could not work with an arbitrary rooted clique-tree.

\begin{myclaim}\label{claim:not-a-container}
If $\sep_{i_j}$ strictly contains a minimal separator of $G_{i_j}$ then, in any $4$-Steiner root $\tree$ of $G$, $\unlabeledsubtree{\sep_{i_j}}$ is a bistar. 
\end{myclaim}

\begin{proofclaim}
Suppose by contradiction $\unlabeledsubtree{\sep_{i_j}}$ is a star. 
Any minimal separator of $G_{i_j}$ that is strictly contained into $\sep_{i_j}$ should have size at most two. However, Property~\ref{pty-fct:2} of Theorem~\ref{thm:final-clique-tree} ensures that all such minimal separators should have size at least three.
A contradiction.
\end{proofclaim}
By Claim~\ref{claim:not-a-container}, we may assume from now on that $\sep_{i_j}$ does not strictly contain any minimal separator of $G_{i_j}$ (otherwise, $\labeledsubtree{\sep_{i_j}}{}$ cannot be a star, and we are done with this case).
We are now left with two possibilities:

\subsubsection{Subcase no minimal separator of $G_{i_j}$ contains $\sep_{i_j}$.}\label{sec:star-uncontained}
Given any $4$-Steiner root $\labeledsubtree{i_j}{}$ of $G_{i_j}$ where $\labeledsubtree{i_j}{\sep_{i_j}} \equiv_G \labeledsubtree{\sep_{i_j}}{}$, we define an encoding as follows:
$${\tt short-encode}(\labeledsubtree{i_j}{}) :=  \langle c, c_{i_j}, d_{\labeledsubtree{i_j}{}}(c), d_{\labeledsubtree{i_j}{}}(c_{i_j}) \rangle,$$
where $c_{i_j} \in \centre{\labeledsubtree{i_j}{\maxk_{i_j}}} \cap N_{\labeledsubtree{i_j}{}}[c]$ is arbitrary. 
The relationship between short encodings and {\sc Distance-Constrained Root} is discussed at the end of the section.
First we prove the following result:

\begin{myclaim}\label{claim:star-1}
If ${\tt short-encode}(\labeledsubtree{i_j}{}) = {\tt short-encode}(\labeledsubtree[2]{i_j}{})$ and $\tree := \tree[0] \odot \labeledsubtree{i_j}{}$ is a $4$-Steiner root of $G$ then, $\tree[2] := \tree[0] \odot \labeledsubtree[2]{i_j}{}$ is also a $4$-Steiner root of $G$.
\end{myclaim}

\begin{proofclaim}
It suffices to prove $d_{\labeledsubtree{i_j}{}}(v) = d_{\labeledsubtree[2]{i_j}{}}(v)$ for every $v \in \sep_{i_j}$.
For that, we need to analyze the possible intersections between $\sep_{i_j}$ and the minimal separators in $\maxk_{i_j}$.
Recall that $\sep_{i_j}$ is not a minimal separator of $G_{i_j}$ by the hypothesis.
\begin{itemize}
\item Moreover, assume $c_{i_j} \neq c$.
By Lemma~\ref{lem:center-in-star} we have $Real(N_{\tree}[c]) = \sep_{i_j}$.
Combined with the fact that a minimal separator of $G_{i_j}$ can neither contain $\sep_{i_j}$ nor be strictly contained into $\sep_{i_j}$ (Claim~\ref{claim:not-a-container}), this implies all the paths between $\sep_{i_j}$ and $W_{i_j}$ must pass by $c,c_{i_j}$ (see Fig.~\ref{fig:star-encoding-1} for an illustration).
In this situation, our partial encoding already contains all the distances' information we need.

\begin{figure}[h!]
\centering
\includegraphics[width=.25\textwidth]{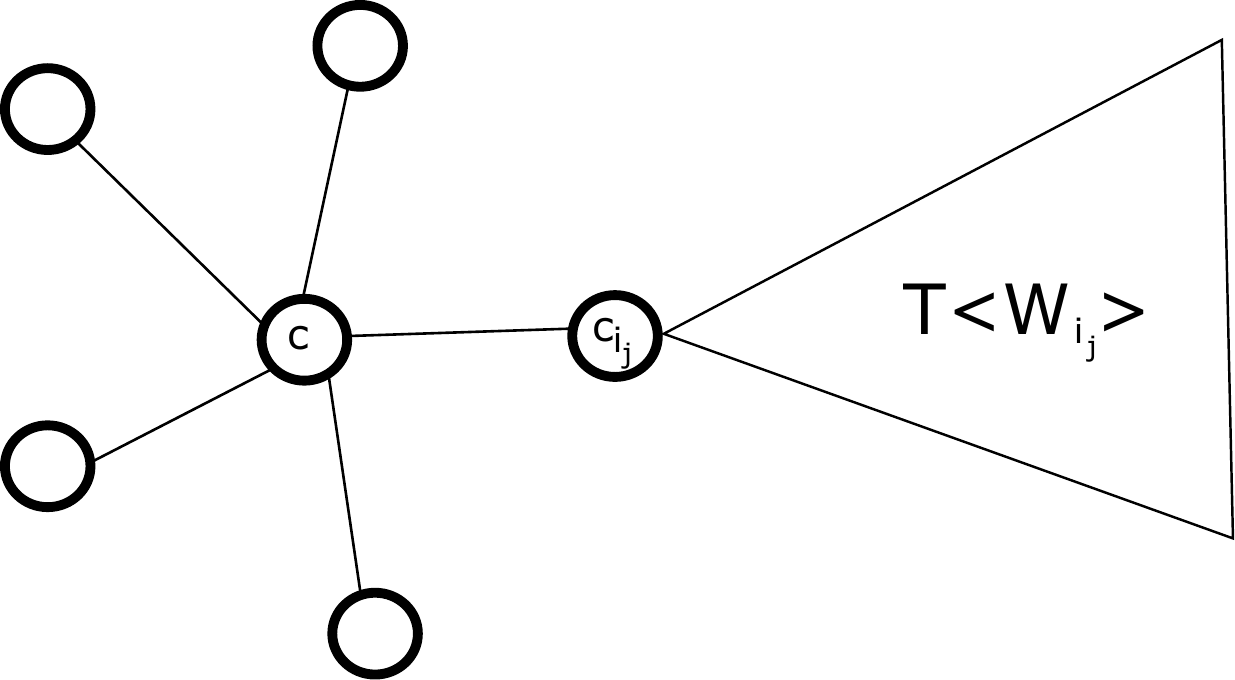}
\caption{A schematic view of $\unlabeledsubtree{V_{i_j}}$.}
\label{fig:star-encoding-1}
\end{figure}

\item Otherwise, $c_{i_j} = c$.
A simple transformation of the construction proposed in Lemma~\ref{lem:almost-simplicial-placement} shows that we can always assume the simplicial vertices among $\sep_{i_j} \setminus \{c\}$ (in $G_{i_j}$) to be leaves adjacent to $c$ in $T$.
Namely, we can make all these vertices leaves of $\unlabeledsubtree{\maxk_{i_j}}$ in such a way that they are connected to $c$ via a path with one Steiner node.
We complete this construction by contracting each such simplicial vertex with its Steiner neighbour. 

\begin{figure}[h!]
\centering
\includegraphics[width=.25\textwidth]{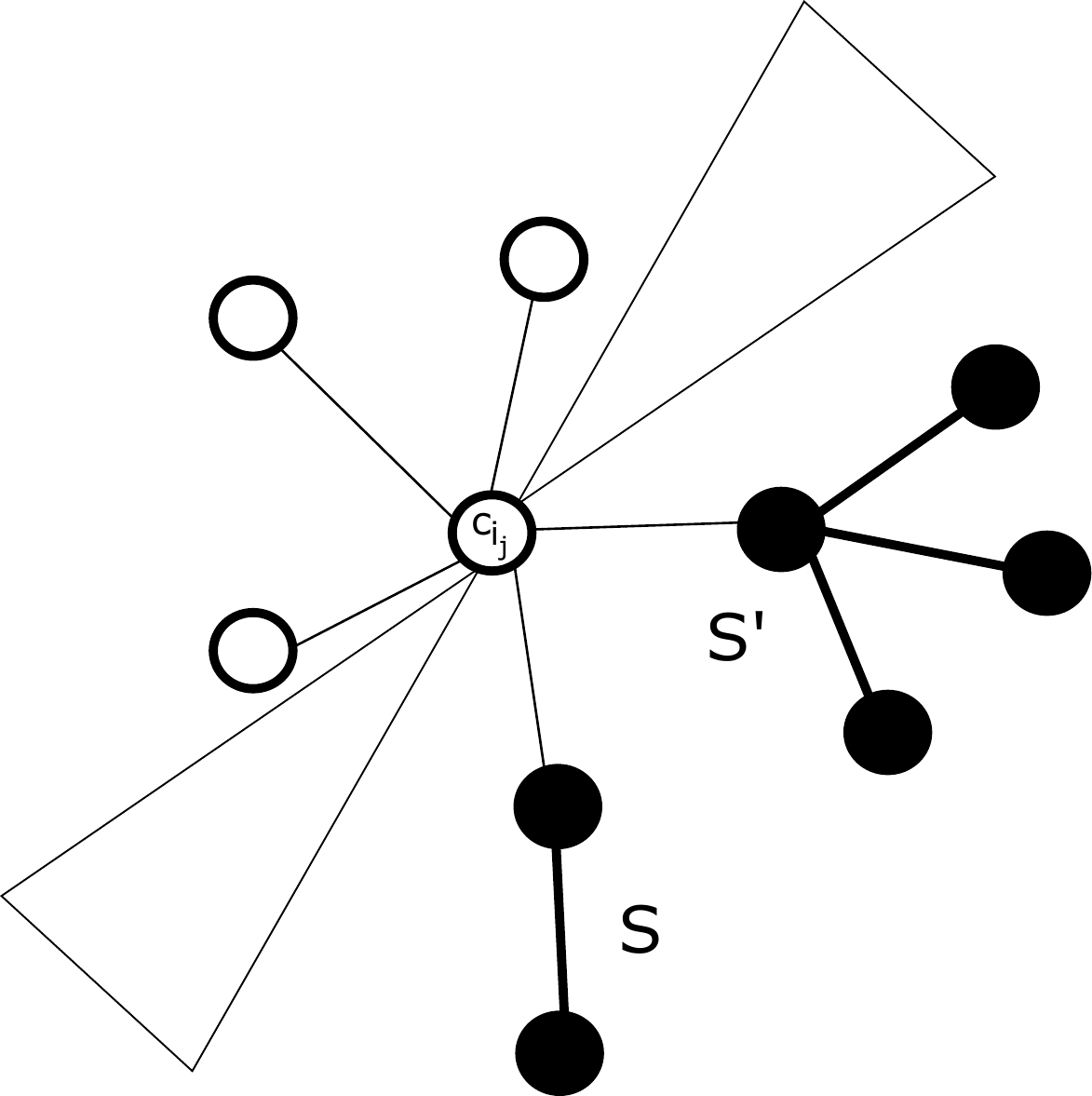}
\caption{The case $c = c_{i_j}$. Two minimal separators $\sep,\sep'$ overlapping $\sep_{i_j}$ are drawn in bold.}
\label{fig:star-encoding-2}
\end{figure}

\smallskip
We end up showing that {\em all} the vertices in $\sep_{i_j} \setminus \{c\}$ that are contained into another minimal separator $\sep$ of $G_{i_j}$ are adjacent in $\tree$ to some vertex in $W_{i_j}$ (hence, their distance to this set is known implicitly and does not need to be stored in the encoding).
Indeed, since $\sep$ and $\sep_{i_j}$ overlap, we cannot have $\unlabeledsubtree{\sep}$ is a bistar (otherwise, $\sep_{i_j} \subseteq \sep$ by Lemma~\ref{lem:bistar-center}).
In particular: either $\unlabeledsubtree{\sep}$ is an edge with exactly one end in $\sep_{i_j}$; or $\unlabeledsubtree{\sep}$ is a non edge star and, by Lemma~\ref{lem:center-in-star}, the unique vertex in $(\sep \cap \sep_{i_j}) \setminus \{c\}$ is its center.
See Fig.~\ref{fig:star-encoding-2} for an illustration.
\end{itemize}
\end{proofclaim}

Finally given ${\tt short-encode}(\labeledsubtree{i_j}{})$, we can transform such a short encoding into the constraints $(d_r)_{r \in V(\labeledsubtree{i_j}{\sep_{i_j}})}$ where:
\begin{itemize}
\item $d_c = d_{\labeledsubtree{i_j}{}}(c)$
%\item If $c_i$ is a real node that is different than $c$ then, $d_{c_i} = d_{T_{i_j}}(c_i)$;
\item If $c_{i_j}$ is a real node that is different than $c$ then, $d_{c_{i_j}} = d_{\labeledsubtree{i_j}{}}(c_{i_j})$;
\item For all other nodes $r \in \sep_{i_j}$: $$d_r = \begin{cases}
d_c + 1 \ \text{if} \ c \neq c_{i_j} \ \text{or} \ r \ \text{is simplicial in} \ G_{i_j} \\
1 \ \text{otherwise.}
\end{cases}.$$
\end{itemize}
Note that in doing so, $d_c \in \{2,3\} \ \text{and when it is defined} \ d_{c_{i_j}} \in \{1,2\}$.
Overall, there are at most $2^2 = 4$ possibilities for a fixed pair $\labeledsubtree{\sep_{i_j}}{}, c_{i_j}$.
Furthermore, this above transformation is not injective, and we can so obtain the same constraints for different short encodings (thereby further reducing the size of ${\cal T}_{i_j}$).
The reason why this does not matter is that assuming we made a correct guess for ${\tt short-encode}(\labeledsubtree{i_j}{})$, we proved in Claim~\ref{claim:star-1} that we have $d_{\labeledsubtree{i_j}{}}(r) = d_r$ for any $r \in V(\labeledsubtree{\sep_{i_j}}{})$.
In particular, if $\labeledsubtree{i_j}{}$ can be extended to a $4$-Steiner root of $G$ then, so could be any partial solution $\labeledsubtree[2]{i_j}{}$ that would satisfy these above constraints as we would have $d_{\labeledsubtree[2]{i_j}{}}(r) \geq d_{\labeledsubtree{i_j}{}}(r)$ for any $r \in V(\labeledsubtree{\sep_{i_j}}{})$.
%Summarizing, we are left solving several instances of \textsc{Distance-Constrained Root} for $X_{i_j}$ with additional distance constraints that are derived from $c,c_{i_j}$.
%In particular if $c_{i_j} \neq c$ is Steiner then, all the distances' information we need are encoded by the constraint $d_c$.
%Overall, this only leaves us with ${\cal O}(1)$ different sets of distances' constraints to be considered.

\subsubsection{Subcase a minimal separator of $G_{i_j}$ contains $\sep_{i_j}$.}\label{sec:star-contained}

As for the previous subcase, we start introducing a short encoding then, we explain its relationship with {\sc Distance-Constrained Root} at the end of this section.
The novelty here is that we need to complete our encoding with additional distances' conditions, that we will also use in order to define our distances' constraints in the input.
Specifically, let us fix a pair $c_i,c_{i_j}$ and let us only consider the partial solutions $\labeledsubtree{i_j}{}$ where $\labeledsubtree{i_j}{\sep_{i_j}} \equiv_G \labeledsubtree{\sep_{i_j}}{}$ and $c_{i_j} \in \centre{\labeledsubtree{i_j}{\maxk_{i_j}}} \cap N_{\labeledsubtree{i_j}{}}[c]$.
We set:
$${\tt short-encode-2}(\labeledsubtree{i_j}{}) = \left[  \langle c, c_{i}, d_{\labeledsubtree{i_j}{}}(c), d_{\labeledsubtree{i_j}{}}(c_{i}) \rangle \,\middle\vert\, d_{\labeledsubtree{i_j}{}}(v), \ \forall v \in \sep_i \cap \sep_{i_j} \right].$$

In order to bound the number of possible such encodings, we prove that: 

\begin{myclaim}
$|\sep_i \cap \sep_{i_j}| = {\cal O}(1)$.
\end{myclaim}

\begin{proofclaim}
Recall that by Lemma~\ref{lem:star-center}, $\sep_{i_j}$ must be weakly $\cliquetree{G}$-convergent.
Since we assume $|\sep_{i_j}| \geq 3$, by Property~\ref{pty-fct:1} of Theorem~\ref{thm:final-clique-tree}, $\sep_{i_j}$ is $\cliquetree{G}$-convergent.
Moreover since there is a minimal separator of $G_{i_j}$ that contains $\sep_{i_j}$, the maximal clique incident to all edges in $\superedgeset{G}{\sep_{i_j}}$ must be $\maxk_{i_j}$.
This implies that $\sep_i := \maxk_i \cap \maxk_{p(i)}$ cannot contain $\sep_{i_j}$.
In particular, $|\sep_i \cap \sep_{i_j}| \leq 2$.
\end{proofclaim}

This new encoding above may not be informative enough in some cases.
We complete it with additional {\em distances' conditions}.
%That is, for any other $v \in S_{i_j}$, we still fix the value of the constraint $d_v$ (that is necessary in order to keep the number of considered sets of constraints to a constant), however this value may depend on some other factors such as the minimal separators contained into $S_{i_j}$.
Specifically, 
%we group the solutions ${\cal T}_{i_j}$ that leads to the same new encoding, thereby obtaining the subsets ${\cal T}_{i_j}^0, {\cal T}_{i_j}^1, \ldots, {\cal T}_{i_j}^{r_{i_j}}$. 
%Then, for every fixed $t \in \{0, 1, \ldots, r_{i_j}\}$ we proceed as follows.
we consider all the other minimal separators $\sep_{i_k} := \maxk_i \cap \maxk_{i_k}$ between $\maxk_i$ and one of its children nodes such that $\sep_{i_k} \subset \sep_{i_j}$.
Note that since we assume $\labeledsubtree{\sep_{i_j}}{}$ to be a star, we must have $|\sep_{i_k}| \leq 2$.
There are two possibilities:
\begin{itemize}
\item If $\sep_{i_k} = \{v_{i_k}\}$ then, by Claim~\ref{claim:cut-vertex} there is only one solution left in ${\cal T}_{i_k}$.
Specifically, this solution $\labeledsubtree{i_k}{} \in {\cal T}_{i_k}$ maximizes $d_{i_k} := d_{\labeledsubtree{i_k}{}}(v_{i_k})$.
We are left ensuring $d_{\labeledsubtree{i_j}{}}(v_{i_k}) > 4 - d_{i_k}$. 
%We only keep in ${\cal T}_{i_j}^t$ the partial solutions that satisfy this new constraint.
\item Otherwise, $\sep_{i_k} = \{u_{i_k},v_{i_k}\}$.
Then, $\labeledsubtree{i_j}{\sep_{i_k}}$ must be an edge and we may assume w.l.o.g. $c_{i_j} = u_{i_k}$.
We are left to ensure that $d_{\labeledsubtree{i_j}{}}(v_{i_k}) \geq 2$.
\end{itemize}

\begin{myclaim}\label{claim:star-2}
For any $\labeledsubtree{i_j}{}$ that satisfies the above distances' conditions, one of the following properties must be true:
\begin{enumerate}
\item $\labeledsubtree{i_j}{}$ can be extended to a $4$-Steiner root of $G$;
\item For any $4$-Steiner root $\labeledsubtree[2]{i_j}{}$ of $G_{i_j}$ such that {\small${\tt short-encode-2}(\labeledsubtree{i_j}{}) = {\tt short-encode-2}(\labeledsubtree[2]{i_j}{})$}, we cannot extend $\labeledsubtree[2]{i_j}{}$ to a {\em well-structured} $4$-Steiner root of $G$.
\end{enumerate}
\end{myclaim}

\begin{figure}
\centering
\includegraphics[width=.4\textwidth]{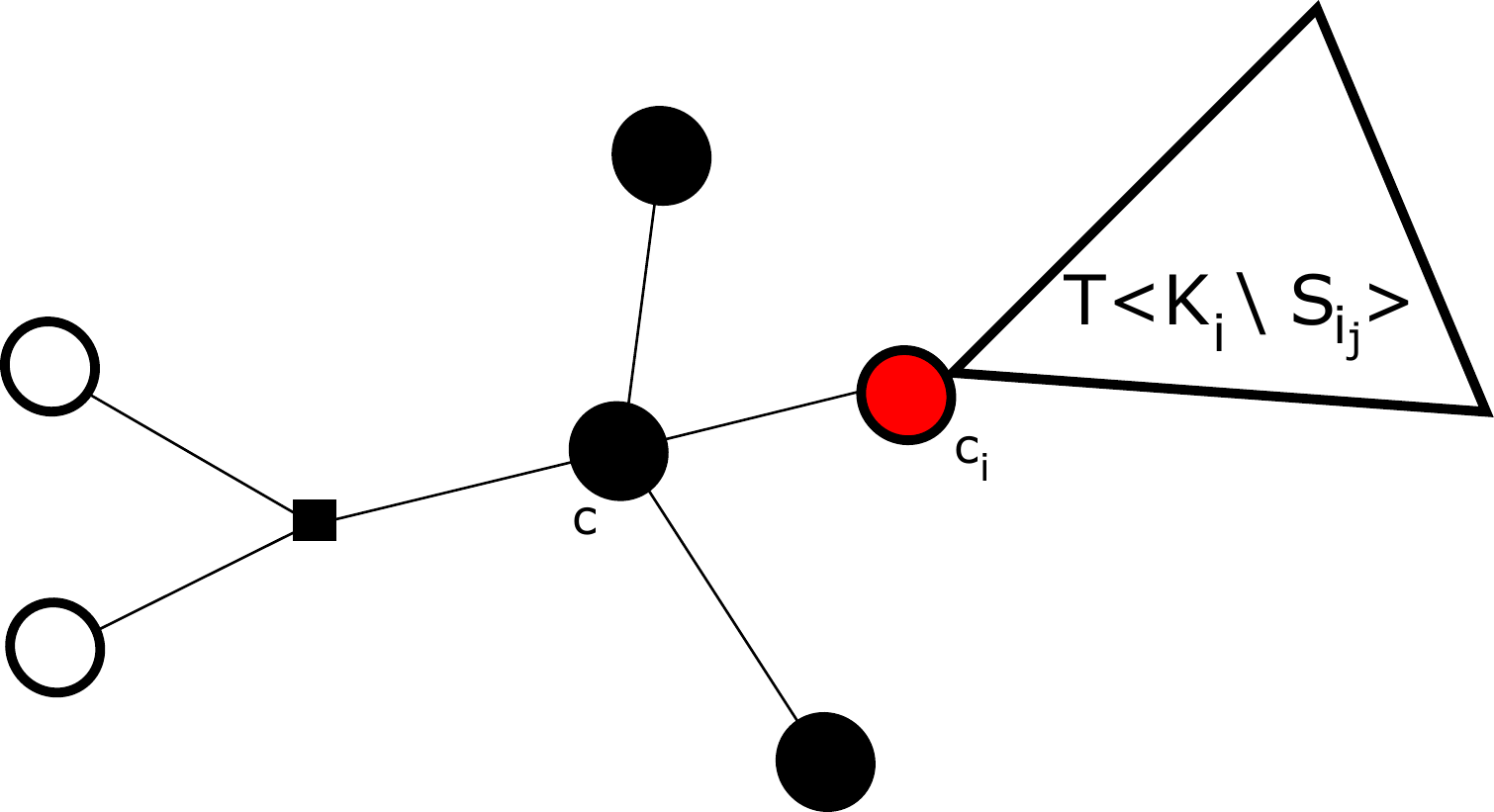}
\caption{To the proof of Claim~\ref{claim:star-2}.}
\label{fig:star-encoding-3}
\end{figure}

\begin{proofclaim}
See Fig.~\ref{fig:star-encoding-3} for an illustration.
Assume that there is a $\labeledsubtree[2]{i_j}{}$ as stated in the claim that can be extended to a well-structured $4$-Steiner root $\tree[2]$ of $G$.
In order to prove the claim, it suffices to prove that $\tree := \unlabeledsubtree[2]{V \setminus W_{i_j}} \odot \labeledsubtree{i_j}{}$ (or a slight modification of it) is also a $4$-Steiner root of $G$.
For that, since $E(G \setminus W_{i_j}) \cup E(G_{i_j}) = E(G)$ (all edges are covered), it suffices to prove $dist_{\tree}(V \setminus V_{i_j},W_{i_j}) > 4$.
We start observing that according to Property~\ref{pty-fct:2} of Theorem~\ref{thm:final-clique-tree}, there is a minimal separator $\sep'$ of $G_{i_j}$ that strictly contains $\sep_{i_j}$.
Such a minimal separator $\sep'$ must induce a bistar in $\tree[2]$ and $\tree$.
Then, as $c$ must be the center of one of the two maximal cliques containing $\sep'$ ({\it i.e.}, see Lemmata~\ref{lem:bistar-center} and~\ref{lem:star-center}), we so obtain $c \neq c_i$.
In particular, all the paths between $\sep_{i_j}$ and $\maxk_i \setminus \sep_{i_j}$ will need to pass by $c_i$, and so: 
$$dist_{\tree}(\maxk_i \setminus \sep_{i_j},W_{i_j}) = dist_{\tree[2]}(\maxk_i \setminus \sep_{i_j},W_{i_j}) > 4.$$
We now prove as a subclaim that we also have $dist_{\tree}(V \setminus V_i,W_{i_j}) = dist_{\tree[2]}(V \setminus V_i,W_{i_j}) > 4$.
Indeed, in order to have $dist_{\tree}(u,v) \leq 4$ for some $u \notin V_i, v \in W_{i_j}$, we need an intermediate node $r \in \unlabeledsubtree{\sep_i} \cap \unlabeledsubtree{\sep_{i_j}}$ onto the unique $uv$-path in $\tree$.
However, such a node $r$ should be either $c$ or in $\sep_i \cap \sep_{i_j}$.
By definition of our short encoding, this would imply $dist_{\tree[2]}(u,v) \leq 4$, a contradiction.

We finally consider the other minimal separators $\sep_{i_k} := \maxk_i \cap \maxk_{i_k}$ between $\maxk_i$ and one of its children nodes that intersect $\maxk_{i_j}$.
%%On the way, we will discard some partial solutions in ${\cal T}_{i_j}$.
We have $\sep_{i_j} \not\subseteq \sep_{i_k}$ (otherwise, $\sep_{i_j}$ could not be $\cliquetree{G}$-convergent, thereby contradicting Property~\ref{pty-fct:1} of Theorem~\ref{thm:final-clique-tree}).
Thus, $|\sep_{i_k} \cap \sep_{i_j}| \leq 2$.
%%In particular, $\tree{S_{i_k}}$ is not a bistar (otherwise, $S_{i_j} \subset S_{i_k}$ by Lemma~\ref{lem:star-center}).
Furthermore since we have $c \neq c_i$, Lemma~\ref{lem:center-in-star} implies that either $\sep_{i_k} \subset \sep_{i_j}$ or $\sep_{i_k} \cap \sep_{i_j} \subseteq \{c_i,c\}$.
In the latter case: $$dist_{\tree}(W_{i_k},W_{i_j}) = dist_{\tree[2]}(W_{i_k},W_{i_j}) > 4.$$
%%Observe this is always the case provided $|S_{i_k}| \geq 3$.
Thus from now on assume $\sep_{i_k} \subset \sep_{i_j}$ (and so, $|\sep_{i_k}| \leq 2$).

If $\sep_{i_k}$ is a cut-vertex then, it follows from Claim~\ref{claim:cut-vertex} and the distances' conditions satisfied by $\labeledsubtree{i_j}{}$ that we can always assume $d_{\labeledsubtree{i_j}{}}(\sep_{i_k}) > 4 - dist_{\tree[2]}(\sep_{i_k},W_{i_k})$.
In particular, $dist_{\tree}(W_{i_k},W_{i_j}) > 4$.
We end up with the case $\sep_{i_k} = \{u_{i_k},v_{i_k}\}$.
Then, $\unlabeledsubtree{\sep_{i_k}}$ must be an edge and we may assume w.l.o.g. $c = u_{i_k}$.
Since we assume $\unlabeledsubtree{\sep_{i_j}}$ is a non-edge star, $c$ is adjacent to some other leaf than $v_{i_k}$.
In other words, $dist_{\tree}(c,V \setminus W_{i_k}) = 1$ is minimized.
Then, since we have $dist_{\tree[2]}(W_{i_k},W_{i_j}) > 4$ we must have $dist_{\tree[2]}(c,W_{i_k}) = 4$ and so, $dist_{\tree[2]}(v_{i_k},W_{i_k}) = 3$.
It follows from Claim~\ref{claim:edge} and the distances' conditions satisfied by $\labeledsubtree{i_j}{}$ that we can always assume $dist_{\tree}(W_{i_k},W_{i_j}) = dist_{\tree[2]}(v_{i_k},W_{i_k}) + d_{\labeledsubtree{i_j}{}}(v_{i_k}) \geq 3 + 2 > 4$. 
\end{proofclaim}

Finally given ${\tt short-encode-2}(\labeledsubtree{i_j}{})$, we can transform such a short encoding into the constraints $(d_r)_{r \in V(\labeledsubtree{\sep_{i_j}}{})}$ where:
\begin{itemize}
\item $d_c = d_{\labeledsubtree{i_j}{}}(c)$
\item For any $v \in \sep_{i_j} \cap \sep_i$, $d_v = d_{\labeledsubtree{i_j}{}}(v)$.
\item If $c_i$ is a real node that is not in $\sep_i \cup \{c\}$ then, $d_{c_i} = d_{\labeledsubtree{i_j}{}}(c_i)$;
%\item If $c_{i_j}$ is a real node that is different than $c$ then, $d_{c_{i_j}} = d_{T_{i_j}}(c_{i_j})$;
\item If $v_{i_k} \in \sep_{i_k} \subset \sep_{i_j}$ has a distance-condition then, $d_{v_{i_k}}$ is set to the largest such condition.
\item Finally, for all other nodes $r \in V(\labeledsubtree{\sep_{i_j}}{})$: $d_r = 1$ (trivial constraint).
\end{itemize}

For any fixed $\labeledsubtree{\sep_{i_j}}{}$ the mapping $\varphi : {\tt short-encode-2}(\labeledsubtree{i_j}{}) \to \langle c_i, (d_r)_{r \in V(\labeledsubtree{\sep_{i_j}}{})} \rangle$ is injective.
Moreover we proved in Claim~\ref{claim:star-2} that, in any $4$-Steiner root $\tree$ of $G$, all the paths between $\unlabeledsubtree{\sep_{i}}$ and $W_{i_j}$ must go by $\{c,c_i\} \cup (\sep_i \cap \sep_{i_j})$.
Therefore, our short encodings always preserve a yes-instance of {\sc Distance-Constrained Root} at $\maxk_i$ provided one exists.

%
%%We call ${{\cal T}_{i_j}^t}'$ the resulting subset of solutions.
%For any other $v \in S_{i_j}$, we may impose the trivial constraint .
%Unlike the previous subcase we may have $d_{T_{i_j}}(v) > d_v$, and so, it may be not completely clear why the passing from the short encoding to distances' constraints can be made w.l.o.g.
%
%
%However what really matters here is that we can always conclude that $v \neq c_i$ ({\it i.e.}, because if $c_i$ is not a Steiner node then, $d_{c_i} \geq 3$).
%In any solution $T_{i_j'}$ where this is true (possibly, $T_{i_j'}$ has a different short encoding) we prove in the following claim that the path between this vertex $v$ and any vertex in $V_i \setminus V_{i_j}$ will always pass by $c$.
%Hence, we only need to correctly guess $d_c$. 
%We solve \textsc{Distance-Constrained Root} with these additional distances' constraints above.

\subsection{Case $\labeledsubtree{\sep_{i_j}}{}$ is a bistar}\label{sec:bistar}
We follow the same approach as in Section~\ref{sec:star}.
In fact, the proof is a bit simpler in this case.
For instance by Lemma~\ref{lem:bistar-center}, in {\em any} $4$-Steiner root $\tree$ of $G$ such that $\unlabeledsubtree{\sep_{i_j}} \equiv_G \labeledsubtree{\sep_{i_j}}{}$, we must have $\centre{\labeledsubtree{\sep_{i_j}}{}} = \{c_i,c_{i_j}\}$ where $c_i$ and $c_{i_j}$ are the unique central nodes of $\unlabeledsubtree{\maxk_i}$ and $\unlabeledsubtree{\maxk_{i_j}}$, respectively.

\paragraph{Preliminaries.}
Before choosing our short encoding, we will need to prove some new properties of the rooted clique-tree $\cliquetree{G}$ (given by Theorem~\ref{thm:final-clique-tree}).

\begin{myclaim}\label{claim:bistar-1}
Let $\labeledsubtree{i_j}{}$ be a $4$-Steiner root of $G_{i_j}$ such that $\labeledsubtree{i_j}{\sep_{i_j}} \equiv_G \labeledsubtree{\sep_{i_j}}{}$ is a bistar.
If $\labeledsubtree{i_j}{}$ can be extended to a $4$-Steiner root of $G$ then, all the vertices in $N_{\labeledsubtree{i_j}{}}(c_i) \setminus \{c_{i_j}\}$ are simplicial in $G_{i_j}$ 

(hence, their distance to $W_{i_j}$ is always equal to $d_{\labeledsubtree{i_j}{}}(c_i) + 1$).
\end{myclaim}

\begin{proofclaim}
This may not be the case only if some of these vertices are contained into a minimal separator $\sep$ of $G_{i_j}$.
Then, since $N_{\labeledsubtree{i_j}{}}[c_i] \setminus \{c_{i_j}\}$ is a connected component of $\labeledsubtree{i_j}{\maxk_{i_j}} \setminus \{c_{i_j}\}$, we should have $\sep \subseteq \sep_{i_j}$.
By Property~\ref{pty-fct:2} of Theorem~\ref{thm:final-clique-tree} we obtain that $|\sep| \geq 3$.
Furthermore, $\sep \subset \sep_{i_j}$ because, by Lemma~\ref{lem:bistar-center}, $G \setminus \sep_{i_j}$ has exactly two full components.
This implies that the only possibility for $\labeledsubtree{i_j}{\sep}$ is a star; moreover, $Real(N_{\labeledsubtree{i_j}{}}[c_i]) = \sep$ by Lemma~\ref{lem:star-center}. 
However, by Property~\ref{pty-fct:2} of Theorem~\ref{thm:final-clique-tree}, there should exist a minimal separator of $G_{i_j}$ that strictly contains $\sep$.
As proved in Lemma~\ref{lem:star-center}, $\sep$ must be weakly $\cliquetree{G}$-convergent.
Hence, all edges in $\superstrictedgeset{G}{\sep}$ must be incident to $\maxk_{i_j}$.
This implies that in fact, the bistars $\labeledsubtree{i_j}{\sep'}, \ \sep' \supset \sep$ can only intersect in $N_{\labeledsubtree{i_j}{}}[c_{i_j}]$.
So, $Real(N_{\labeledsubtree{i_j}{}}[c_{i_j}]) = \sep$, that is a contradiction since $c_{i_j} \neq c_i$.
\end{proofclaim}

\begin{myclaim}\label{claim:bistar-2}
Let $\labeledsubtree{i_j}{}$ be a $4$-Steiner root of $G_{i_j}$ such that $\labeledsubtree{i_j}{\sep_{i_j}} \equiv_G \labeledsubtree{\sep_{i_j}}{}$ is a bistar, and let $\tree$ be a $4$-Steiner root of $G$ extending $\labeledsubtree{i_j}{}$.
One of the following conditions is true:
\begin{enumerate}
\item 
Either every real node in $N_{\tree}(c_{i_j})$ is simplicial in $G_{i_j}$ or is adjacent in $\tree$ to a vertex of $W_{i_j}$.
\item 
Or $|\sep_i \cap N_{\tree}[c_{i_j}]| \leq 2$, and in the same way $|\sep_{i_k} \cap N_{\tree}[c_{i_j}]| \leq 2$ for any other child $\maxk_{i_k}$ of $\maxk_i$.
\end{enumerate}
\end{myclaim}

\begin{proofclaim}
Recall that according to Lemma~\ref{lem:bistar-center}, $\sep_{i_j}$ is only contained into the maximal cliques $\maxk_i$ and $\maxk_{i_j}$.
In particular if we prove $\maxk \cap \maxk' \subseteq \sep_{i_j}$ for some $\{\maxk,\maxk'\} \neq \{\maxk_i,\maxk_{i_j}\}$ then, this inclusion is strict.
We implicitly use this fact in what follows.
Specifically if $\sep_i \subseteq \sep_{i_j}$ (resp., $\sep_{i_k} \subseteq \sep_{i_j}$ for some child $\maxk_{i_k}$ of $\maxk_i, \ k \neq j$, or $\sep \subseteq \sep_{i_j}$ for some minimal separator $\sep$ of $G_{i_j}$) then, we must have $\sep_i \subset \sep_{i_j}$ (resp., $\sep_{i_k} \subset \sep_{i_j}$ or $\sep \subset \sep_{i_j}$).

For the remaining of the proof we may assume $|N_{\tree}[c_{i_j}]| \geq 3$ (otherwise we are done).
Our proof is based on the following two observations:
\begin{enumerate}
\item By Property~\ref{pty-ci:2} of Theorem~\ref{thm:clique-intersection}, having $|\sep_i \cap N_{\tree}[c_{i_j}]| \geq 3$ ($|\sep_{i_k} \cap N_{\tree}[c_{i_j}]| \geq 3$, resp.) would imply $Real(N_{\tree}[c_{i_j}]) \subseteq \sep_i$ ($Real(N_{\tree}[c_{i_j}]) \subseteq \sep_{i_k}$, resp.).
Since $N_{\tree}[c_{i_j}] \setminus \{c_i\}$ is a connected component of $\unlabeledsubtree{\maxk_i} \setminus \{c_i\}$, this would in turn imply $\sep_i \subset \sep_{i_j}$ ($\sep_{i_k} \subset \sep_{i_j}$, resp.).
As a result, we would have $\sep_i = Real(N_{\tree}[c_{i_j}])$ (resp., $\sep_{i_k} = Real(N_{\tree}[c_{i_j}])$).
\item By Property~\ref{pty-fct:2} of Theorem~\ref{thm:final-clique-tree} any minimal separator of $G_{i_j}$ that is strictly contained into $\sep_{i_j}$ must have size at least $3$, be $\cliquetree{G}$-convergent, and be strictly contained into a minimal separator of $G_{i_j}$.
So, the only possibility for such a separator is also $\sep = Real(N_{\tree}[c_{i_j}])$.
\end{enumerate}
If $\sep_{i_j}$ contains a minimal separator $\sep$ of $G_{i_j}$ then, by the above two observations we cannot have $\sep = \sep_i$, neither $\sep = \sep_{i_k}$ for any other child $\maxk_{i_k}$ of $\maxk_i$ (otherwise, $\sep$ could not be $\cliquetree{G}$-convergent).
In particular, we cannot have $|\sep_i \cap N_{\tree}[c_{i_j}]| \geq 3$, neither $|\sep_{i_k} \cap N_{\tree}[c_{i_j}]| \geq 3$, and so we always fall in Case 2 of the claim.

From now on we assume that $\sep_{i_j}$ does not strictly contain any minimal separator of $G_{i_j}$.
We prove as a subclaim that either $Real(N_{\tree}[c_{i_j}]) \neq \sep_i$ and $Real(N_{\tree}[c_{i_j}]) \neq \sep_{i_k}$ for any child $\maxk_{i_k}$, or there is no separator of $G_{i_j}$ that contains $Real(N_{\tree}[c_{i_j}])$.
Indeed, suppose by contradiction $Real(N_{\tree}[c_{i_j}]) = \sep_i$ and there exists a separator $\sep'$ of $G_{i_j}$ that contains $Real(N_{\tree}[c_{i_j}])$.
%Then, since we assume $N_T[c_{i_j}] \geq 3$ and it is not minimal separator of $G_{i_j}$, $\tree{S'}$ must be a bistar. 
It implies by Lemma~\ref{lem:star-center} $\sep_i$ is weakly $\cliquetree{G}$-convergent but not $\cliquetree{G}$-convergent, thereby contradicting Property~\ref{pty-fct:1} of Theorem~\ref{thm:final-clique-tree}.
The proof for $\sep_{i_k}$ is identical as the one above, thereby proving the subclaim.

Finally, let us assume either $Real(N_{\tree}[c_{i_j}]) = \sep_i$ or $Real(N_{\tree}[c_{i_j}]) = \sep_{i_k}$ for some child $\maxk_{i_k}$ (otherwise, according to our first observation above, we always fall in Case 2 of the claim).
Since no minimal separator of $G_{i_j}$ can contain $Real(N_{\tree}[c_{i_j}])$, we can reuse the same proof as for Claim~\ref{claim:star-1} (Case $c = c_{i_j}$) in order to prove that we always fall in Case 1 of the claim.
Specifically, we end up showing that {\em all} the vertices in $N_{\tree}(c_{i_j})$ that are contained into another minimal separator $\sep$ of $G_{i_j}$ are adjacent in $\tree$ to some vertex in $W_{i_j}$.
Indeed, $\sep$ and $Real(N_{\tree}[c_{i_j}])$ must overlap because we assume $\sep \not\subseteq \sep_{i_j}$ and $Real(N_{\tree}[c_{i_j}]) \not\subseteq \sep$. 
Then, we cannot have $\unlabeledsubtree{\sep}$ is a bistar (otherwise, $Real(N_{\tree}[c_{i_j}]) \subseteq \sep$ by Lemma~\ref{lem:bistar-center}).
In particular: either $\unlabeledsubtree{\sep}$ is an edge with exactly one end in $N_{\tree}[c_{i_j}]$; or $\unlabeledsubtree{\sep}$ is a non edge star and, by Lemma~\ref{lem:center-in-star}, the unique vertex in $(\sep \cap \sep_{i_j}) \setminus \{c_{i_j}\}$ is its center.
\end{proofclaim}

\paragraph{The encoding.}
For any $4$-Steiner root $\labeledsubtree{i_j}{}$ of $G_{i_j}$ such that $\labeledsubtree{i_j}{\sep_{i_j}} \equiv_G \labeledsubtree{S_{i_j}}{}$ is a bistar, we include in our short encoding the following information: 
%$$c_i \ \text{and} \ \begin{cases} \text{if} \ |N_{T_{i_j}}[c_{i_j}]| \leq 3 \ \text{then,} \ d_{T_{i_j}}(v) \ \text{for any} \ v \in N_{T_{i_j}}[c_{i_j}] \\
%\text{otherwise,} \ d_{T_{i_j}}(v) \ \text{for any} \ v \in \{c_{i_j}\} \cup (N_{T_{i_j}}[c_{i_j}] \cap S_i)\end{cases}.$$
\begin{flalign*}
\left[ c_i, d_{\labeledsubtree{i_j}{}}(c_i),d_{\labeledsubtree{i_j}{}}(c_{i_j}) \right] \ \text{and} \
\begin{cases}
\text{only if} \ |Real(N_{\labeledsubtree{i_j}{}}[c_{i_j}])| \leq 2 \text{)} \ d_{\labeledsubtree{i_j}{}}(r), \ \forall r \in N_{\labeledsubtree{i_j}{}}[c_{i_j}] \\
\text{or (only if} \ |\sep_i \cap N_{\labeledsubtree{i_j}{}}[c_{i_j}]| \leq 2 \text{)} \ d_{\labeledsubtree{i_j}{}}(v), \ \forall v \in N_{\labeledsubtree{i_j}{}}[c_{i_j}] \cap \sep_i
\end{cases}
\end{flalign*}
As usual, the relationship between this above encoding and {\sc Distance-Constrained Root} is made explicit at the end of the section.
There are only ${\cal O}(1)$ possibilities for a fixed $\labeledsubtree{\sep_{i_j}}{}$. By Theorem~\ref{thm:minsep}, we so obtain ${\cal O}(|\sep_{i_j}|^5)$ different encodings.
However, we need to complete this case with similar distances' conditions as for the star case (Section~\ref{sec:star}).

\paragraph{Additional conditions.}
Specifically, assume $|Real(N_{\labeledsubtree{\sep_{i_j}}{}}[c_{i_j}])| \geq 3$ and $|N_{\labeledsubtree{\sep_{i_j}}{}}[c_{i_j}] \cap \sep_{i_k}| \leq 2$ for any other child $\maxk_{i_k}$ of $\maxk_i$ (otherwise, no additional constraint is needed).
%We group the solutions ${\cal T}_{i_j}$ that leads to the same new encoding, thereby obtaining the subsets ${\cal T}_{i_j}^0, {\cal T}_{i_j}^1, \ldots, {\cal T}_{i_j}^{r_{i_j}}$. 
%Then, for every fixed $t \in \{0, 1, \ldots, r_{i_j}\}$ we proceed as follows.
We consider all the other minimal separators $\sep_{i_k} := \maxk_i \cap \maxk_{i_k}$ between $\maxk_i$ and one of its children nodes such that $\sep_{i_k} \subseteq N_{\labeledsubtree{\sep_{i_j}}{}}[c_{i_j}]$.
In particular, $|\sep_{i_k}| \leq 2$.
There are two possibilities:
\begin{itemize}
\item If $\sep_{i_k} = \{v_{i_k}\}$ then, by Claim~\ref{claim:cut-vertex} there is only one solution left in ${\cal T}_{i_k}$.
Specifically, this solution $\labeledsubtree{i_k}{} \in {\cal T}_{i_k}$ maximizes $d_{i_k} := d_{\labeledsubtree{i_k}{}}(v_{i_k})$.
We are left ensuring $d_{\labeledsubtree{i_j}{}}(v_{i_k}) > 4 - d_{i_k}$. 
%We only keep in ${\cal T}_{i_j}^t$ the partial solutions that satisfy this new constraint.
\item Otherwise, $\sep_{i_k} = \{u_{i_k},v_{i_k}\}$.
We have $\sep_{i_k} \neq Real(N_{\labeledsubtree{i_j}{}}[c_{i_j}])$.
Then, the tree $\labeledsubtree{\sep_{i_j}}{\sep_{i_k}}$ must be an edge and we may assume w.l.o.g. $c_{i_j} = u_{i_k}$.
We are left to ensure that $d_{\labeledsubtree{i_j}{}}(v_{i_k}) \geq 2$.
\end{itemize}
%We call ${{\cal T}_{i_j}^t}'$ the resulting subset of solutions.
%We finally solve \textsc{Distance-Constrained Root} with these additional distances' constraints that we defined above.

\begin{myclaim}\label{claim:bistar-3}
For any $\labeledsubtree{i_j}{}$ that satisfies all of the above distances' conditions, one of the following properties is true:
\begin{enumerate}
\item Either $\labeledsubtree{i_j}{}$ can be extended to a $4$-Steiner root of $G$;
\item Or for any $4$-Steiner root $\labeledsubtree[2]{i_j}{}$ of $G_{i_j}$ with the same short encoding as $\labeledsubtree{i_j}{}$, we cannot extend $\labeledsubtree[2]{i_j}{}$ to a {\em well-structured} $4$-Steiner root of $G$.
\end{enumerate}
\end{myclaim}

\begin{proofclaim}
Assume there is a $\labeledsubtree[2]{i_j}{}$ as stated in the claim that can be extended to a well-structured $4$-Steiner root $\tree[2]$ of $G$.
In order to prove the claim, it suffices to prove that $\tree := \unlabeledsubtree[2]{V \setminus W_{i_j}} \odot \labeledsubtree{i_j}{}$ (or a slight modification of it) is also a $4$-Steiner root of $G$.
Equivalently, we are left proving that $dist_{\tree}(V \setminus V_{i_j},W_{i_j}) > 4$.
By Theorem~\ref{thm:x-free}, we have $dist_{\tree}(v,W_{i_j}) = 2 + dist_{\labeledsubtree{i_j}{}}(c_{i},W_{i_j})$ and in the same way $dist_{\tree[2]}(v,W_{i_j}) = 2 + dist_{\labeledsubtree[2]{i_j}{}}(c_{i},W_{i_j})$ for any simplicial vertex $v \in \maxk_i$.
In particular:
$$dist_{\tree}(v,W_{i_j}) = dist_{\tree[2]}(v,W_{i_j}) > 4. $$
So, we are left to consider the other minimal separators $\sep_{i_k} := \maxk_i \cap \maxk_{i_k}$ between $\maxk_i$ and any other node (including its father node $\maxk_{p(i)}$).
Note that $\sep_{i_k}$ cannot both intersect $N_{\tree}(c_{i_j})$ and $N_{\tree}(c_{i})$ (otherwise, $\sep_{i_k} = \sep_{i_j}$, thereby contradicting Lemma~\ref{lem:bistar-center}).
Moreover by Claim~\ref{claim:bistar-1}, any real node in $N_{\tree}(c_i) \setminus \{c_{i_j}\}$ is simplicial in $G_{i_j}$.

Let us first assume $\sep_{i_k} \cap \sep_{i_j} \subseteq N_{\tree}[c_i]$. 
If $i_k \neq i$ then, we have in this situation: 
\begin{align*}
dist_{\tree}(W_{i_j},W_{i_k}) = \min\{ &dist_{\tree}(W_{i_k},c_{i}) + dist_{\labeledsubtree{i_j}{}}(c_{i},W_{i_j}), \\
&dist_{\tree}(W_{i_k},c_{i_j}) + dist_{\labeledsubtree{i_j}{}}(c_{i_j},W_{i_j})\},
\end{align*} 
and in the same way: 
\begin{align*}
dist_{\tree[2]}(W_{i_j},W_{i_k}) = \min\{ &dist_{\tree[2]}(W_{i_k},c_{i}) + dist_{\labeledsubtree[2]{i_j}{}}(c_{i},W_{i_j}), \\
&dist_{\tree[2]}(W_{i_k},c_{i_j}) + dist_{\labeledsubtree[2]{i_j}{}}(c_{i_j},W_{i_j})\}.
\end{align*} 
In particular:
$$dist_{\tree}(W_{i_j},W_{i_k}) = dist_{\tree[2]}(W_{i_j},W_{i_k}) > 4. $$
Otherwise, $i_k = i$, and we also obtain that:
$$dist_{\tree}(W_{i_j},V \setminus V_i) = dist_{\tree[2]}(W_{i_j},V \setminus V_{i}) > 4. $$

As a result, we are only interested in the situation $\sep_{i_k} \cap N_{\tree}(c_{i_j}) \neq \emptyset$ -- that implies $\sep_{i_k} \subseteq N_{\tree}[c_{i_j}]$.
We further assume $|\sep_{i_k}| \leq 2$ since otherwise, we are done by Case 1 of Claim~\ref{claim:bistar-2} and the fact that $\labeledsubtree{i_j}{},\labeledsubtree[2]{i_j}{}$ have the same short encoding.
Then, there are two cases ({\it i.e.}, exactly the same as for the star case):
\begin{itemize}
\item Assume $\sep_{i_k} = \{v_{i_k}\}$. 
If $i_k=i$ then, the distance between $v_{i_k}$ and $W_{i_j}$ is part of our short encoding, and so we are done.
Otherwise, as explained above (Section~\ref{sec:small}), we only kept in ${\cal T}_{i_k}$ a partial solution $\labeledsubtree{i_k}{}$ maximizing $d_{i_k} := d_{\labeledsubtree{i_k}{}}(v_{i_k})$. In this situation, it follows from the distances' conditions over $\labeledsubtree{i_j}{}$ that we have $dist_{\tree}(W_{i_k},W_{i_j}) > 4$.
\item Otherwise, $\sep_{i_k} = \{u_{i_k},v_{i_k}\}$.
Recall that $\sep_{i_k} \subset \sep_{i_j}$.
We may further assume $|Real(N_{\tree}[c_{i_j}])| \geq 3$ and $\sep_i \neq \sep_{i_k}$ (otherwise, the encoding already includes the distance to $W_{i_j}$ from any node in $\unlabeledsubtree{\sep_{i_k}}$, and so we are done).
Thus, $\unlabeledsubtree{\sep_{i_k}}$ must be an edge and we may assume w.l.o.g. $c_{i_j} = u_{i_k}$.
Since we assume $|Real(N_{\tree}[c_{i_j}])| \geq 3$, $c_{i_j}$ is adjacent to some other real node than $v_{i_k}$.
In other words, $dist_{\tree}(c_{i_j},V \setminus W_{i_k}) = 1$ is minimized.
Then, since we have $dist_{\tree[2]}(W_{i_k},W_{i_j}) > 4$ we must have $dist_{\tree[2]}(c_{i_j},W_{i_k}) = 4$ and so, $dist_{\tree[2]}(v_{i_k},W_{i_k}) = 3$.
It follows from Claim~\ref{claim:edge} and the distances' conditions over $\labeledsubtree{i_j}{}$ that we can always assume $dist_{\tree}(W_{i_k},W_{i_j}) > 4$. 
\end{itemize}
\end{proofclaim}

Finally, an encoding for bistars is transformed into distances' constraints as follows:
\begin{itemize}
\item If $d_{\labeledsubtree{i_j}{}}(r)$ is included in the short encoding then, $d_r := d_{\labeledsubtree{i_j}{}}(r)$. In particular, this will be the case for $c_i,c_{i_j}$.
\item If $v_{i_k} \in \sep_{i_k} \cap \sep_{i_j}$ needs to satisfy some specified distance-condition then, $d_{v_{i_k}}$ is set to the largest such a condition.
\item For all other vertices $v \in \sep_{i_j}$, $d_v = 1$ (trivial constraint).
\end{itemize}
For any fixed bistar $\labeledsubtree{\sep_{i_j}}{}$ the mapping from the encodings to the distances' constraints is bijective.
Indeed, in order to prove it is the case, the only difficulty is to prove that we can correctly identify from the constraints the nodes $c_i,c_{i_j}$.
Since we will always impose $d_{c_{i_j}} \leq 2$ whereas $d_{c_i} \geq 3$, this is always possible.

\section{Step~\ref{step-4}: The dynamic programming}\label{sec:greedy}

%In order to anticipate an intermediate problem that we will introduce in Section~\ref{sec:encoding}, we end up this section with the following consequence of Theorem~\ref{thm:leaf}:
%
%\begin{corollary}\label{cor:leaf}
%Given $G=(V,E)$ and a rooted clique-tree $T_G$ of $G$, let $X_i \in {\cal K}(G)$ be a leaf and let $(d_v)_{v \in V(\tree{S_i})}$ be a sequence of positive integers.
%
%We can construct, in time polynomial in $|X_i|$, a set ${\cal T}_i$ of $4$-Steiner roots for $G_i := G[X_i]$ with the following additional property: If $G$ has a {\em well-structured} $4$-Steiner root $T$ where, for any $v \in V(\tree{S_i})$: $$dist_T(v,X_i \setminus S_i) \geq d_v$$
%then, there exists $T'_i \in {\cal T}_i$ Steiner-equivalent to $\tree{X_i}$.
%\end{corollary}
%
%\begin{proof}
%We construct the family given by Theorem~\ref{thm:leaf}.
%We only keep the trees $T_i \in {\cal T}_i$ that satisfy the additional distance constraints we have.
%\end{proof}

In what follows, let $||G|| := \sum_{\maxk_i \in \MAXK{G}} |\maxk_i|$.
For a chordal graph, $||G|| = {\cal O}(n+m)$~\cite{BlP93}.
We can now state the core result of this paper:

\begin{theorem}\label{thm:dyn-prog}
Let $G=(V,E)$ be strongly chordal, let $\cliquetree{G}$ be a rooted clique-tree as in Theorem~\ref{thm:final-clique-tree} and let $\maxk_i \in \MAXK{G}$. 
There is some polynomial $P$ such that, after a pre-processing in time ${\cal O}(n || G_i ||^5 P(n))$, we can solve {\sc Distance-Constrained Root} for any input $\labeledsubtree{\sep_i}{},(d_r)_{r \in V(\labeledsubtree{\sep_i}{})}$ in time ${\cal O}(P(n))$.
\end{theorem}

Theorem~\ref{thm:dyn-prog} proves Theorem~\ref{thm:main-steiner-power} directly.
Note that we made no effort in order to improve the running time in our analysis.
A very rough analysis shows that we have $P(n) = {\cal O}(n^{15})$.

\begin{proofof}{Theorem~\ref{thm:dyn-prog}}
%If $X_i$ is a leaf of $T_G$ then, this follows from Corollary~\ref{cor:leaf}.
If $\maxk_i$ is a leaf of $\cliquetree{G}$ then, we construct the family given by Theorem~\ref{thm:leaf}.
We only keep the trees $\labeledsubtree{i}{} \in {\cal T}_i$ that satisfy the additional constraints we have.
Thus from now on, assume $\maxk_i$ is an internal node with children $\maxk_{i_1},\maxk_{i_2},\ldots,\maxk_{i_p}$.

\medskip
\noindent
{\bf Preprocessing.}
Let ${\cal T}_{i_1}, {\cal T}_{i_2}, \ldots, {\cal T}_{i_p}$ be as in Theorem~\ref{thm:encoding}.
By induction on $\cliquetree{G}$, the computation of all the ${\cal T}_{i_j}$'s requires total preprocessing time $\sum_{j=1}^p {\cal O}(n||G_{i_j}||^{5}P(n))$, and $\sum_{j=1}^p{\cal O}(|\sep_{i_j}|^5 P(n))$ additional time.
We also need to construct the family ${\cal F}_i$ of Proposition~\ref{prop:internal}, that takes ${\cal O}(|\maxk_i|^7 \cdot n^3m\log{n})$-time.
-- Recall that the elements in ${\cal F}_i$ are of the form $(\labeledsubtree{Y_i}{},{\cal C}_i)$ where $Y_i \subseteq \maxk_i$ and ${\cal C}_i$ must represent the center of $\unlabeledsubtree{\maxk_i}$ (missing vertices of $\maxk_i \setminus Y_i$ are supposed to be located in thin branches, see Lemma~\ref{lem:thin:branch}). --
%Overall since we have: $$ |{\cal K}(G_i)| = 1 + \sum_{j=1}^p |{\cal K}(G_{i_j})|$$
%and: $$ ||G_i|| = |X_i| + \sum_{j=1}^p ||G_{i_j}|| $$
Overall, if we assume w.l.o.g. that $P(n) = \Omega(n^4m\log{n})$ then, this pre-processing phase takes total time:
\begin{align*}
%&\sum_{j=1}^p {\cal O}(n^{19/2} \cdot |{\cal K}(G_{i_j})| \cdot ||G_{i_j}||^{4}) + \sum_{j=1}^p{\cal O}(|S_{i_j}|^3 \cdot n^{7/2}|X_{i_j}|^{5}) + {\cal O}(n|X_i|^{4}) \\
%= \ &{\cal O}(n^{19/2} \cdot ( |{\cal K}(G_i)| - 1 ) \cdot (||G_i|| - |X_i|)^{4}) + {\cal O}(pn^{17/2}|X_i|^3) + {\cal O}(n|X_i|^{6}) \\
%= \ &{\cal O}(n^{19/2} \cdot |{\cal K}(G_i)| \cdot ||G_i||^{4}).
&\sum_{j=1}^p {\cal O}(n||G_{i_j}||^{5}P(n)) + \sum_{j=1}^p{\cal O}(|\sep_{i_j}|^5 P(n)) + {\cal O}(|\maxk_i|^7 \cdot n^3m\log{n}) \\
=& \ {\cal O}(nP(n) \cdot \sum_{j=1}^p ||G_{i_j}||^{5}) + {\cal O}(p|\maxk_i|^5P(n)) + {\cal O}(n|\maxk_i|^5P(n)) \\
=& \ {\cal O}(nP(n) \cdot ( ||G_i||^5 - |\maxk_i|^5 )) + {\cal O}(n|\maxk_i|^5P(n)) \\
=& \ {\cal O}(nP(n)||G_i||^5).
\end{align*}

\medskip
\noindent
{\bf Answering a query.}
In what follows let $\labeledsubtree{\sep_i}{}$ and $(d_r)_{r \in \labeledsubtree{\sep_i}{}}$ be fixed.
Recall that for every $(\labeledsubtree{Y_i}{},{\cal C}_i) \in {\cal F}_i$ we have $\sep_i \subseteq Y_i$, and so, we can check whether $\labeledsubtree{\sep_i}{} \equiv_G \labeledsubtree{Y_i}{\sep_i}$.
This takes total time ${\cal O}(|\sep_i||{\cal F}_i|) = {\cal O}(n|\maxk_i|^9)$.
Then, we consider each $(\labeledsubtree{Y_i}{},{\cal C}_i) \in {\cal F}_i$ that passes this first test above sequentially.
Simply put, we use a series of filtering rules in order to greedily find a solution to {\sc Distance-Constrained Root}, or to correctly conclude that there is none.

\medskip
\underline{Assume first $Y_i = \maxk_i$ (no thin branch).}
For every $r \in \labeledsubtree{\sep_i}{}$ we check whether we have: $$dist_{\labeledsubtree{Y_i}{}}(r, \maxk_i \setminus \sep_i) \geq d_r$$ (otherwise, we violate our distances' constraints).
We will assume from now on it is the case.
In the same way, for every $r_{i_j} \in \labeledsubtree{Y_i}{\sep_{i_j}}, \ j \in \{1,2,\ldots,p\}$, we only keep in ${\cal T}_{i_j}$ those solutions $\labeledsubtree{i_j}{}$ such that we have: $$dist_{\labeledsubtree{Y_i}{}}(r,r_{i_j}) + d_{\labeledsubtree{i_j}{}}(r_{i_j}) \geq d_r.$$
Overall, since $|{\cal T}_{i_j}| = {\cal O}(|\sep_{i_j}|^5) = {\cal O}(|\maxk_i|^5)$, this new verification phase takes total time ${\cal O}(p|\sep_i||\maxk_i|^5) = {\cal O}(n|\maxk_i|^6)$.
Furthermore in doing so, we ensure that {\em any} $4$-Steiner root of $G_i$ that we can obtain from $\labeledsubtree{Y_i}{}$ and the remaining solutions in the ${\cal T}_{i_j}$'s will satisfy all our distances' constraints.
Conversely, if no such a solution can be found then, we can correctly report that our distances' constraints cannot be satisfied in any well-structured $4$-Steiner root of $G$ (by Theorem~\ref{thm:encoding}).

\smallskip
We now introduce another filtering rule, quite similar as the one above, that we will keep using throughout the remaining of the proof.
Specifically, for every $j \in \{1,2,\ldots,p\}$ and $r_{i_j} \in \labeledsubtree{Y_i}{\sep_{i_j}}$, we assign some value $\ell^{i_j}(r_{i_j})$ that intuitively represents the distance of $r_{i_j}$ to $V_i \setminus V_{i_j}$.
Every time the rule is applied, we discard all solutions $\labeledsubtree{i_j}{} \in {\cal T}_{i_j}$ such that $d_{\labeledsubtree{i_j}{}}(r_{i_j}) + \ell^{i_j}(r_{i_j}) \leq 4$.
We set initially $\ell^{i_j}(r_{i_j}) := dist_{\labeledsubtree{Y_i}{}}(r_{i_j}, \maxk_i \setminus \sep_{i_j})$ and we apply the rule.
Overall, updating (initializing, resp.) the values $\ell^{i_j}$ for every $j$ takes time $\sum_{j=1}^p {\cal O}(|\sep_{i_j}|) = {\cal O}(n|\maxk_i|)$.
Applying the rule takes time $\sum_{j=1}^p {\cal O}(|\sep_{i_j}||{\cal T}_{i_j}|) = \sum_{j=1}^p {\cal O}(|\sep_{i_j}|^6) = {\cal O}(n|\maxk_i|^6)$.

In what follows, we explain how to greedily construct a solution (if any), starting from $\labeledsubtree{i}{} := \labeledsubtree{Y_i}{}$.
The procedure is divided into a constant number of phases. Every time we complete one of these phases, we need to apply this above filtering rule.

\begin{itemize}

\item{\it Phase 1: Processing the cut-vertices.}
We consider all the indices $j \in \{1,2,\ldots,p\}$ such that $\sep_{i_j} = \{v\}$ is a cut-vertex.
By Claim~\ref{claim:cut-vertex} there is exactly one solution left in ${\cal T}_{i_j}$.
We add it to the solution, {\it i.e.}, we set $\labeledsubtree{i}{} := \labeledsubtree{i}{} \odot \labeledsubtree{i_j}{}$.
Furthermore, for every $k \in \{1,2,\ldots,p\} \setminus \{j\}$ and $r_{i_k} \in \labeledsubtree{Y_i}{\sep_{i_k}}$ (possibly, $r_{i_k} = v$) we set $\ell^{i_k}(r_{i_k}) := \min\{ \ell^{i_k}(r_{i_k}), \ dist_{\labeledsubtree{Y_i}{}}(r_{i_k},v) + d_{\labeledsubtree{i_j}{}}(v) \}$.
We end up applying the filtering rule above. \\

\item{\it Phase 2: Processing the edges.}
We consider all the indices $j \in \{1,2,\ldots,p\}$ such that $\sep_{i_j} = \{u,v\}$ and $\labeledsubtree{Y_i}{\sep_{i_j}}$ is an edge.
The following claim shows that we can proceed similarly as for Phase 1 provided we know which among $u$ or $v$ will be closest to $V_i \setminus V_{i_j}$.
Therefore, computing this information is the main objective of this phase.

\begin{myclaim}\label{claim:edge-old}
Assume $\sep_{i_j} = \{u,v\}$.
Let $\tree$ be a $4$-Steiner root of $G$ such that $\unlabeledsubtree{\sep_{i_j}}$ is an edge and $dist_{\tree}(u, V \setminus V_{i_j}) \geq dist_{\tree}(v, V \setminus V_{i_j})$.
Then, $\tree[2] := \unlabeledsubtree{V \setminus W_{i_j}} \odot \labeledsubtree[v]{i_j}{}$ is also a $4$-Steiner root of $G$, where $\labeledsubtree[v]{i_j}{} \in {\cal T}_{i_j}$ is, among all solutions in this set such that $\labeledsubtree[v]{i_j}{\sep_{i_j}}$ is an edge and $d_{\labeledsubtree[v]{i_j}{}}(v)$ is maximized, one maximizing $d_{\labeledsubtree[v]{i_j}{}}(u)$.
\end{myclaim}

\begin{proofclaim}
By maximality of $d_{\labeledsubtree[v]{i_j}{}}(v)$ the resulting $\tree[2]$ would not be a $4$-Steiner root of $G$ only if $dist_{\tree}(u, V \setminus V_{i_j}) + d_{\labeledsubtree[v]{i_j}{}}(u) \leq 4$.
But then, since $d_{\labeledsubtree[v]{i_j}{}}(v) - d_{\labeledsubtree[v]{i_j}{}}(u) \leq 1$ (because $\unlabeledsubtree{\sep_{i_j}}$ is an edge), one would obtain $dist_{\tree}(u,V \setminus V_{i_j}) = dist_{\tree}(v,V \setminus V_{i_j})$ and $dist_{\tree}(v,V \setminus V_{i_j}) + d_{\labeledsubtree[v]{i_j}{}}(v) = 5$.
In particular, we should have in the original Steiner root $\tree$: $$\min\{ dist_{\tree}(u,W_{i_j}), dist_{\tree}(v,W_{i_j})\} \geq d_{\labeledsubtree[v]{i_j}{}}(v).$$
As $\labeledsubtree[v]{i_j}{}$ maximizes $d_{\labeledsubtree[v]{i_j}{}}(v)$ and, under this latter condition, $d_{\labeledsubtree[v]{i_j}{}}(u)$ is maximized, we obtain that $d_{\labeledsubtree[v]{i_j}{}}(u) \geq dist_{\tree}(u,W_{i_j}) \geq d_{\labeledsubtree[v]{i_j}{}}(v)$.
\end{proofclaim}

By Claim~\ref{claim:edge-old} we are left to decide which amongst $u$ or $v$ will minimize its distance to $V_i \setminus V_{i_j}$ in the final solution.
If either $u$ or $v$ has a real neighbour in $\labeledsubtree{i}{} \setminus \sep_{i_j}$ then, we are done.
Thus from now on we assume this is not the case.

There may be several other indices $k$ such that $\sep_{i_k} = \sep_{i_j}$.
As an intermediate step, we explain how to merge the solutions in ${\cal T}_{i_j}$ and in ${\cal T}_{i_k}$ into a new set ${\cal T}_{i_j}'$ when this happens.
For that, we consider all the $\labeledsubtree{i_j}{},\labeledsubtree{i_k}{}$ sequentially. 
We put $\labeledsubtree{i_j}{} \odot \labeledsubtree{i_k}{}$ into ${\cal T}_{i_j}'$ if and only if we have $\min\{ d_{\labeledsubtree{i_j}{}}(v) + d_{\labeledsubtree{i_k}{}}(v), d_{\labeledsubtree{i_j}{}}(u) + d_{\labeledsubtree{i_k}{}}(u) \} > 4$.
If so then, $d_{\labeledsubtree{i_j}{} \odot \labeledsubtree{i_k}{}}(u) = \min\{d_{\labeledsubtree{i_j}{}}(u), d_{\labeledsubtree{i_k}{}}(u)\}$, and in the same way $d_{\labeledsubtree{i_j}{} \odot \labeledsubtree{i_k}{}}(v) = \min\{d_{\labeledsubtree{i_j}{}}(v), d_{\labeledsubtree{i_k}{}}(v)\}$.
Overall, since there are at most two solutions stored in each of ${\cal T}_{i_j}$ and ${\cal T}_{i_k}$, this takes constant-time.
We end up applying Claim~\ref{claim:edge} in order to replace ${\cal T}_{i_j}$ by the at most two best solutions in ${\cal T}_{i_j}'$.
By repeating this above procedure, we can assume w.l.o.g. that there is no other index $k$ such that $\sep_{i_k} = \sep_{i_j}$.

We may further assume that there is no index $k$ such that $\sep_{i_k} = \{u\}$ ($\sep_{i_k} = \{v\}$, resp.) for otherwise we already ensured at the last step $d_{\labeledsubtree{i_k}{}}(u) = 4$ ($d_{\labeledsubtree{i_k}{}}(v) = 4$, resp.).
Then, let us assume $dist_{\labeledsubtree{i}{}}(u, {\cal C}_i) \leq dist_{\labeledsubtree{i}{}}(v, {\cal C}_i)$ ($u$ is closer than $v$ to the center of $\labeledsubtree{Y_i}{}$).
In most cases, $u$ will be the closest to $V_i \setminus V_{i_j}$.
Indeed, as we assume $v$ has no real neighbour in $\labeledsubtree{i}{} \setminus \sep_{i_j}$, it is a leaf in $\labeledsubtree{Y_i}{}$.
Therefore, a necessary condition for having $v$ closer than $u$ to $V_i \setminus V_{i_j}$ is that there exists another minimal separator $\sep_{i_k}$ containing $v$.
Let us assume from now on that such a separator $\sep_{i_k}$ exists (otherwise we are done).
Since $v$ is a leaf, this implies $\sep_{i_j} \subset \sep_{i_k}$.
In particular, as we also assume $u$ and $v$ have no real neighbour in $\labeledsubtree{i}{} \setminus \sep_{i_j}$, $\labeledsubtree{Y_i}{\sep_{i_k}}$ must be a bistar. 
We divide our analysis in several subcases:
\begin{itemize}
\item Subcase $\sep_{i_j} = \centre{\labeledsubtree{Y_i}{\sep_{i_k}}}$. 
%By Claim~\ref{claim:no-middle-edge}, there should be a heavy part in $\sep_{i_k}$, and so, one of $u$ or $v$ should have a real neighbour in $\labeledsubtree{i}{}$. 
Both $u$ and $v$ should have a real neighbour in $\labeledsubtree{i}{} \setminus \sep_{i_j}$.
A contradiction.
\item Subcase ${\cal C}_i = \{u\}$. By Claim~\ref{claim:bistar-1}, $v$ is simplicial in $G_{i_k}$. This proves $u$ will be closest than $v$ to $V_i \setminus V_{i_j}$ in this subcase.

\item Otherwise, as $u$ minimizes its distance to the center we must have $\centre{\labeledsubtree{i_k}{\maxk_{i_k} }} = \{u\}$ (this can only be true for at most one index $i_k$).
Note that $v$ is the only leaf of $\labeledsubtree{Y_i}{\sep_{i_k}}$ that is adjacent to $u$.
Therefore, by Lemma~\ref{lem:bistar-center}, $v$ is the only real neighbour of $u$ in any $\labeledsubtree{i_k}{} \in {\cal T}_{i_k}$ {\em and} $v$ must stay the only real neighbour of $u$ in $\labeledsubtree{i}{}$.
This implies that we always have $d_{\labeledsubtree{i_k}{}}(u) = 2$.
We must ensure that the solution $\labeledsubtree{i_j}{} \in {\cal T}_{i_j}$ that we will choose satisfies $d_{\labeledsubtree{i_j}{}}(u) \geq 3$.
Conversely, among all the partial solutions in ${\cal T}_{i_j}$ that satisfies this necessary condition, we can always choose the one $\labeledsubtree{i_j}{}$ maximizing $d_{\labeledsubtree{i_j}{}}(v)$.
\end{itemize}

\item{\it Phase 3: Processing the bistars.}
We consider all the indices $j \in \{1,2,\ldots,p\}$ such that $\labeledsubtree{Y_i}{\sep_{i_j}}$ is a bistar.
Let ${\cal C}_i = \{c_i\}$ and let $\centre{\labeledsubtree{Y_i}{\sep_{i_j}}} = \{c_i,c_{i_j}\}$.
We keep only the solutions $\labeledsubtree{i_j}{} \in {\cal T}_{i_j}$ such that $\centre{\labeledsubtree{i_j}{\maxk_{i_j}}} = \{c_{i_j}\}$.
Then, we have by Claim~\ref{claim:bistar-1} $d_{\labeledsubtree{i_j}{}}(v) = d_{\labeledsubtree{i_j}{}}(c_i) + 1$ for any $v \in N_{\labeledsubtree{i}{}}(c_i) \setminus \{c_{i_j}\}$.
So, we would like to pick $\labeledsubtree{i_j}{} \in {\cal T}_{i_j}$ that maximizes $d_{\labeledsubtree{i_j}{}}(c_i)$.
We prove next that except for one case easy to solve, we can always choose greedily an {\em arbitrary} partial solution $\labeledsubtree{i_j}{}$ which maximizes $d_{\labeledsubtree{i_j}{}}(c_i)$.
Indeed, the only case where we cannot do that w.l.o.g. is when there exists another minimal separator $\sep_{i_k}$ such that $\labeledsubtree{Y_i}{\sep_{i_k}} \cap N_{\labeledsubtree{Y_i}{}}[c_{i_j}] \setminus \{c_i\} \neq \emptyset$.
In fact, as by Lemma~\ref{lem:bistar-center} we will always have $d_{\labeledsubtree{i_j}{}}(c_{i_j}) = 2$ (for any choice for $\labeledsubtree{i_j}{}$), we are only interested in the case when $\sep_{i_k} \cap \left( N_{\labeledsubtree{i}{}}(c_{i_j}) \setminus \{c_i\} \right) \neq \emptyset$.
Then, $\sep_{i_k} \subseteq N_{\labeledsubtree{i}{}}[c_{i_j}]$.
We may further assume $diam(\labeledsubtree{Y_i}{\sep_{i_k}}) \geq 2$ (otherwise, due to Phases 1 and 2, this was already taken into account).
There are two cases:
\begin{itemize}
\item Assume $|\sep_{i_k}| \geq 3$.
By Claim~\ref{claim:bistar-2} and its proof, this implies that there is only one possibility for $d_{\labeledsubtree{i_j}{}}(r)$, for every $r \in \labeledsubtree{Y_i}{\sep_{i_j}}$.
Specifically (Case 1 of the claim), $d_{\labeledsubtree{i_j}{}}(c_{i_j}) = 2$, and for every $u \in N_{\labeledsubtree{Y_i}{\sep_{i_j}}}(c_{i_j})$ either $d_{\labeledsubtree{i_j}{}}(u) = 3$ or (if and only if $u$ belongs to a minimal separator of $G_{i_j}$) $d_{\labeledsubtree{i_j}{}}(u) = 1$.
So, in this situation, there is only one solution stored in ${\cal T}_{i_j}$, and we need to pick this one.
\item Otherwise, $|\sep_{i_k}| = 2$. Then, we impose different properties on the partial solution $\labeledsubtree{i_j}{}$ to choose. Specifically, we will show that we can always choose greedily any solution $\labeledsubtree{i_j}{}$ which maximizes $\sum_{r \in N_{\labeledsubtree{i}{}}[c_{i_j}]} d_{\labeledsubtree{i_j}{}}(r)$.
For that, we start deriving some necessary conditions on the valid partial solutions $\labeledsubtree{i_k}{}$ which we may choose later during the algorithm.

\smallskip
We recall that $diam(\labeledsubtree{Y_i}{\sep_{i_k}}) \geq 2$, and so $\labeledsubtree{Y_i}{\sep_{i_k}}$ is a non-edge star with central node $c_{i_j}$.
By Prop.~\ref{pty-ci:2} of Theorem~\ref{thm:clique-intersection}, it implies that $c_{i_j}$ is Steiner and $Real(N_{\labeledsubtree{i}{}}[c_{i_j}]) = \sep_{i_k}$.
Furthermore, by Lemma~\ref{lem:star-center}, $\sep_{i_k}$ is weakly $\cliquetree{G}$-convergent.
Since $\sep_{i_k} \supset \sep_{i_j}$, this implies that no minimal separator of $G_{i_k}$ can strictly contain $\sep_{i_k}$.
By Prop.~\ref{pty-fct:2} of Theorem~\ref{thm:final-clique-tree}, we so obtain the stronger property that no minimal separator of $G_{i_k}$ can contain $\sep_{i_k}$ ({\it i.e.}, we fall in the situation described in Sec.~\ref{sec:star-uncontained}).
By Lemma~\ref{lem:no-center-intersect}, we must also impose $c_{i_j} \notin \centre{\labeledsubtree{i_k}{\maxk_{i_k}}}$.  
Therefore by Lemma~\ref{lem:center-in-star}, there must exist a node $c_{i_k} \in N_{\labeledsubtree{i_k}{}}(c_{i_j}) \cap \centre{\labeledsubtree{i_k}{\maxk_{i_k}}}$ (possibly, $c_{i_k} \in \sep_{i_k}$).
Altogether combined, we fall in Case 1 of Claim~\ref{claim:star-1}, namely: $d_{\labeledsubtree{i_k}{}}(c_{i_j}) = d_{\labeledsubtree{i_k}{}}(c_{i_k}) + 1$; and for every $v \in \sep_{i_k} \setminus \{c_{i_k}\}$, $d_{\labeledsubtree{i_k}{}}(v) = d_{\labeledsubtree{i_k}{}}(c_{i_k}) + 2$. 
In particular, we must have $d_{\labeledsubtree{i_k}{}}(c_{i_k}) = 2$ (otherwise, we would get $d_{\labeledsubtree{i_k}{}}(c_{i_j}) + d_{\labeledsubtree{i_j}{}}(c_{i_j}) = 4$, regardless of our exact choice for $\labeledsubtree{i_j}{}$).
Then, $d_{\labeledsubtree{i_k}{}}(c_{i_j}) = 3$, and for every $v \in \sep_{i_k} \setminus \{c_{i_k}\}$, $d_{\labeledsubtree{i_k}{}}(v) = 4$.

\smallskip
If we knew $c_{i_k}$ in advance then, we could choose any partial solution $\labeledsubtree{i_j}{}$ such that $d_{\labeledsubtree{i_j}{}}(c_{i_k}) \geq 3$.
Since we do not have this information, we prove the existence of a partial solution $\labeledsubtree{i_j}{}$ such that, for any $v \in \sep_{i_k}$, we have $d_{\labeledsubtree{i_j}{}}(v) < 3$ {\em if and only if} $d_{\labeledsubtree[2]{i_j}{}}(v) < 3$ for any partial solution $\labeledsubtree[2]{i_j}{}$.
Specifically, we prove that any $\labeledsubtree{i_j}{}$ which maximizes $\sum_{r \in N_{\labeledsubtree{i}{}}[c_{i_j}]} d_{\labeledsubtree{i_j}{}}(r)$ has this property.
For that, let us fix such a solution $\labeledsubtree{i_j}{}$ and let $v \in \sep_{i_k}$ be such that $d_{\labeledsubtree{i_j}{}}(v) < 3$ (if any).
Then, there must be a minimal separator $\sep \subset \maxk_{i_j}$ such that $v \in \sep$.
Note that we cannot have $\sep \subseteq \sep_{i_j}$ (otherwise, by Prop.~\ref{pty-fct:2} of Theorem~\ref{thm:final-clique-tree} we have $|\sep| \geq 3$, but then we also have $\sep \subseteq Real(N_{\labeledsubtree{i}{}}[c_{i_j}]) = \sep_{i_k}$ and $|\sep_{i_k}| = 2$, a contradiction).
So, there are only two possibilities:
\begin{itemize}
\item
If $\sep_{i_k} \not\subseteq \sep$ then, in any $\labeledsubtree[2]{i_j}{}$, the only possibility for $\labeledsubtree{i_j}{\sep}$ is to be a star of which $v$ is a central node (possibly, $\labeledsubtree{i_j}{\sep}$ is an edge).
In particular, we have that $d_{\labeledsubtree[2]{i_j}{}}(v) \leq dist_{\labeledsubtree[2]{i_j}{}}(v, \sep \setminus \sep_{i_j}) = 1 < 3$.
\item
Otherwise, $\sep_{i_k} \subseteq \sep$.
Since $\sep \not\subseteq \sep_{i_j}$ we have in fact $\sep_{i_k} \subset \sep$, and so, in any $\labeledsubtree[2]{i_j}{}$, we must have $\labeledsubtree[2]{i_j}{\sep}$ is a bistar. 
By the hypothesis we have $d_{\labeledsubtree{i_j}{}}(v) < 3$, that is possible only if $\centre{\labeledsubtree{i_j}{\sep}} = \{c_{i_j},v\}$.
Then, by maximality of $\sum_{r \in N_{\labeledsubtree{i}{}}[c_{i_j}]} d_{\labeledsubtree{i_j}{}}(r)$, $v$ must be the vertex outputted by the canonical completion method of Lemma~\ref{lem:real-center-bistar} (applied to $\sep, \maxk, R = \sep_{i_k}$ and $c_{i_j}$, where $\sep = \maxk \cap \maxk_{i_j}$).
Conversely, by the intermediate Claims~\ref{claim:op2-before} and~\ref{claim:op2-before-b} we obtain that in any $\labeledsubtree[2]{i_j}{}$, we also have $\centre{\labeledsubtree[2]{i_j}{\sep}} = \{c_{i_j},v\}$.
As a result, $d_{\labeledsubtree[2]{i_j}{}}(v) < 3$.
\end{itemize}
\end{itemize}
Finally, as in the two previous phases, for every $k \in \{1,2,\ldots,p\} \setminus \{j\}$ and $r_{i_k} \in \labeledsubtree{Y_i}{\sep_{i_k}}$ we update $\ell^{i_k}(r_{i_k})$ and then, we end up applying the filtering rule above. \\

\item{\it Phase 4: Processing the stars.}
We finally consider all the indices $j \in \{1,2,\ldots,p\}$ such that $\labeledsubtree{Y_i}{\sep_{i_j}}$ is a star.
%As a guidance towards our next choices, we start analyzing the possibilities we still have among ${\cal T}_{i_j}$.
%Any solution $T_{i_j} \in {\cal T}_{i_j}$ that remained until this step should satisfy $Real({\cal C}(T_{i_j}\langle X_{i_j} \rangle)) \subseteq S_{i_j} \setminus {\cal C}_i$.
%In particular if $d_{T_{i_j}}(v) = 1$ then, $v$ is a leaf of $T_{Y_i}$.
%
Let $\centre{\labeledsubtree{Y_i}{\sep_{i_j}}} = \{c\}$.
We divide the analysis in two subphases: 

\begin{itemize}
\item {\it Subphase 4.a: Processing a star when $c \in {\cal C}_i$.}
As a guidance towards our next choices, we start analyzing the possibilities we still have among ${\cal T}_{i_j}$:

\begin{myclaim}\label{claim:star-4-a}
The following properties are true for any $\labeledsubtree{i_j}{} \in {\cal T}_{i_j}$:
\begin{enumerate}
\item $diam(\labeledsubtree{i_j}{\maxk_{i_j}}) = 4$;
\item and the unique center node $v_j \in \centre{\labeledsubtree{i_j}{\maxk_{i_j}}}$ is either in $\sep_{i_j} \setminus {\cal C}_i$, or it is a Steiner node.
Moreover:
\begin{enumerate}
\item if $v_j \in \sep_{i_j}$ then, $v_j$ is a leaf of $\labeledsubtree{Y_i}{}$;
%\item if $c \in \labeledsubtree{i_j}{\sep_{i_k}}$ for some child $\maxk_{i_k}$ of $\maxk_{i_j}$ then, $dist_{\labeledsubtree{i_j}{}}(c,W_{i_k}) = 4$; 
%\item every vertex of $\sep_{i_j} \setminus \{v_j\}$ must be simplicial in $G_{i_j}$.
\item and for every $r \in \labeledsubtree{Y_i}{\sep_{i_j}}$, $d_{\labeledsubtree{i_j}{}}(r) = dist_{\labeledsubtree{i_j}{}}(r,v_j) + d_{\labeledsubtree{i_j}{}}(v_j)$. 
\end{enumerate}
\end{enumerate}
\end{myclaim}

\begin{proofclaim}
We show that assuming any of these above properties does not hold, some distances' constraints would be violated w.r.t. our previous choices in the other Phases, and so, we should have discarded $\labeledsubtree{i_j}{}$ when we applied the filtering rule. 
Suppose by contradiction $diam(\labeledsubtree{i_j}{\maxk_{i_j}}) < 4$.
Then, the only possibility is $diam(\labeledsubtree{i_j}{\maxk_{i_j}}) = 3$, and so, $c \in \centre{\labeledsubtree{i_j}{\maxk_{i_j}}}$. 
However, the latter would contradict Lemma~\ref{lem:no-center-intersect} as we already assume $c \in {\cal C}_i$.
Therefore, $diam(\labeledsubtree{i_j}{\maxk_{i_j}}) = 4$, thereby implying $\centre{\labeledsubtree{i_j}{\maxk_{i_j}}} = \{v_j\}$ for some $v_j$.

By Lemma~\ref{lem:center-in-star}, $Real(N_{\labeledsubtree{i_j}{}}[c]) = \sep_{i_j}$ and $c \in N_{\labeledsubtree{i_j}{}}[v_j]$.
Thus, either $v_j$ is Steiner, or $v_j \in \sep_{i_j}$.
Furthermore if $v_j \in \sep_{i_j}$ then, $v_j \in \sep_{i_j} \setminus {\cal C}_i$ (otherwise, this would contradict Lemma~\ref{lem:no-center-intersect}).
As every real node adjacent to $v_j$ in $\labeledsubtree{Y_i}{}$ should be in $\maxk_{i_j}$, and we have $\maxk_i \cap \maxk_{i_j} = \sep_{i_j}$, we so obtain that $v_j \in \sep_{i_j} \Longrightarrow v_j \ \text{is a leaf-node of} \ \labeledsubtree{Y_i}{}$.

Finally, we prove that for every $r \in \labeledsubtree{Y_i}{\sep_{i_j}}$, $d_{\labeledsubtree{i_j}{}}(r) = dist_{\labeledsubtree{i_j}{}}(r,v_j) + d_{\labeledsubtree{i_j}{}}(v_j)$. 
Suppose for the sake of contradiction that it is not the case.
We recall that by Lemma~\ref{lem:center-in-star} we have $Real(N_{\labeledsubtree{i_j}{}}[c]) = \sep_{i_j}$. This implies that $ \labeledsubtree{Y_i}{\sep_{i_j}} \setminus \{v_j\}$ is a connected component of $\labeledsubtree{i_j}{\maxk_{i_j}} \setminus \{v_j\}$.
In particular, there must exist a minimal separator $\sep$ of $G_{i_j}$ that intersects $ \labeledsubtree{Y_i}{\sep_{i_j}} \setminus \{v_j\}$ (otherwise, the equality $d_{\labeledsubtree{i_j}{}}(r) = dist_{\labeledsubtree{i_j}{}}(r,v_j) + d_{\labeledsubtree{i_j}{}}(v_j)$ would hold for any node $r \in \labeledsubtree{Y_i}{\sep_{i_j}}$).
%every vertex of $\sep_{i_j} \setminus \{v_j\}$ must be simplicial in $G_{i_j}$.
By Claim~\ref{claim:not-a-container}, $\sep \not\subset \sep_{i_j}$.
%Indeed, we recall that by Lemma~\ref{lem:center-in-star} we have $Real(N_{\labeledsubtree{i_j}{}}[c]) = \sep_{i_j}$.
%By Claim~\ref{claim:not-a-container}, no minimal separator of $G_{i_j}$ can be strictly contained into $\sep_{i_j}$.
So, we are only left with two subcases:
\begin{itemize}
\item Subcase $\sep_{i_j} \subseteq \sep$. By Property~\ref{pty-fct:2} of Theorem~\ref{thm:final-clique-tree} we can assume w.l.o.g. $\sep$ strictly contains $\sep_{i_j}$.
However, this would imply $\labeledsubtree{i_j}{\sep}$ is a bistar, and so, by Lemma~\ref{lem:star-center}, $c$ would also be in $\centre{\labeledsubtree{i_j}{\maxk}}$ for some maximal clique $\maxk$ in $G_{i_j}$.
A contradiction.
\item Subcase $\sep_{i_j} \not\subseteq \sep$. Then, $\sep_{i_j} \cap \sep \subseteq \{c,v_j\}$. 
Note that in particular, all vertices in $\sep_{i_j} \setminus \{c,v_j\}$ are simplicial in $G_{i_j}$. 
For the remaining of this subcase we assume $c \in \sep$ (otherwise we are done). 
Let $C$ be the connected component of $G_{i_j} \setminus \sep$ containing $\sep_{i_j} \setminus \sep$.
We upper bound the distances in $\labeledsubtree{i_j}{}$ as follows:
\begin{itemize}
\item We prove as a subclaim that all the paths between $C \setminus \sep_{i_j}$ and $\sep_{i_j} \setminus \sep$ of length at most four must go by $v_j$. Indeed, suppose by contradiction this is not the case.
In particular, there exist some vertex $u \in C \setminus \sep_{i_j}$ that is adjacent to a vertex of $\sep_{i_j} \setminus \sep$ in $G$ and a $u\sep_{i_j}$-path that does not go by $v_j$.
%Since $C$ is connected, we may take w.l.o.g. $u$ .
Then, as all vertices in $\sep_{i_j} \setminus \sep$ are simplicial in $G_{i_j}$, we must have $u \in \maxk_{i_j}$, and the path between $u$ and $\sep_{i_j} \setminus \sep$ goes by $c$.
However, we have $u \in \maxk_{i_j} \Longrightarrow dist_{\labeledsubtree{i_j}{}}(u,v_j) \leq 2$.
Since we suppose that the path between $u$ and $\sep_{i_j}$ does not go by $v_j$, we obtain $dist_{\labeledsubtree{i_j}{}}(u,c) = 1$, and so $u \in \sep_{i_j}$. A contradiction.
\item Finally, we must have $dist_{\labeledsubtree{i_j}{}}(c,V_{i_j} \setminus N_{G_{i_j}}[\sep_{i_j} \setminus \sep]) \geq 4$ (otherwise we would get $dist_{\labeledsubtree{i_j}{}}(\sep_{i_j} \setminus \sep,V_{i_j} \setminus N_{G_{i_j}}[\sep_{i_j} \setminus \sep]) \leq 4$, a contradiction). We are done in this subcase as we always have $d_{\labeledsubtree{i_j}{}}(c) \leq 1 + d_{\labeledsubtree{i_j}{}}(v_j) \leq 4$.    
\end{itemize}
\end{itemize}
Overall, since we reached a contradiction in both subcases, we proved as claimed that for every $r \in \labeledsubtree{Y_i}{\sep_{i_j}}$, $d_{\labeledsubtree{i_j}{}}(r) = dist_{\labeledsubtree{i_j}{}}(r,v_j) + d_{\labeledsubtree{i_j}{}}(v_j)$. 
\end{proofclaim}

For any $\labeledsubtree{i_j}{} \in {\cal T}_{i_j}$, let $v_j$ be as defined in Claim~\ref{claim:star-4-a}.
Let $\sep_{i_k} \neq \sep_{i_j}$ be unprocessed (in particular, $\labeledsubtree{Y_i}{\sep_{i_k}}$ is a non-edge star).
Since either $v_j$ is Steiner or $v_j$ is a leaf of $\labeledsubtree{Y_i}{}$, we have $v_j \notin \labeledsubtree{Y_i}{\sep_{i_k}}$.
Furthermore for every $r \in \labeledsubtree{Y_i}{\sep_{i_j}}$, $d_{\labeledsubtree{i_j}{}}(r) = dist_{\labeledsubtree{i_j}{}}(r,v_j) + d_{\labeledsubtree{i_j}{}}(v_j)$. 
This implies that w.r.t. $\sep_{i_k}$, any solution $\labeledsubtree{i_j}{}$ that maximizes $d_{\labeledsubtree{i_j}{}}(\centre{\labeledsubtree{i_j}{\maxk_{i_j}}})$ would be a best possible choice -- {\it i.e.}, regardless of our exact choice for $v_j$.
However, we also need to account for the other indices $k$ such that $\sep_{i_k} = \sep_{i_j}$.

\begin{figure}[h!]
\centering
\includegraphics[width=.3\textwidth]{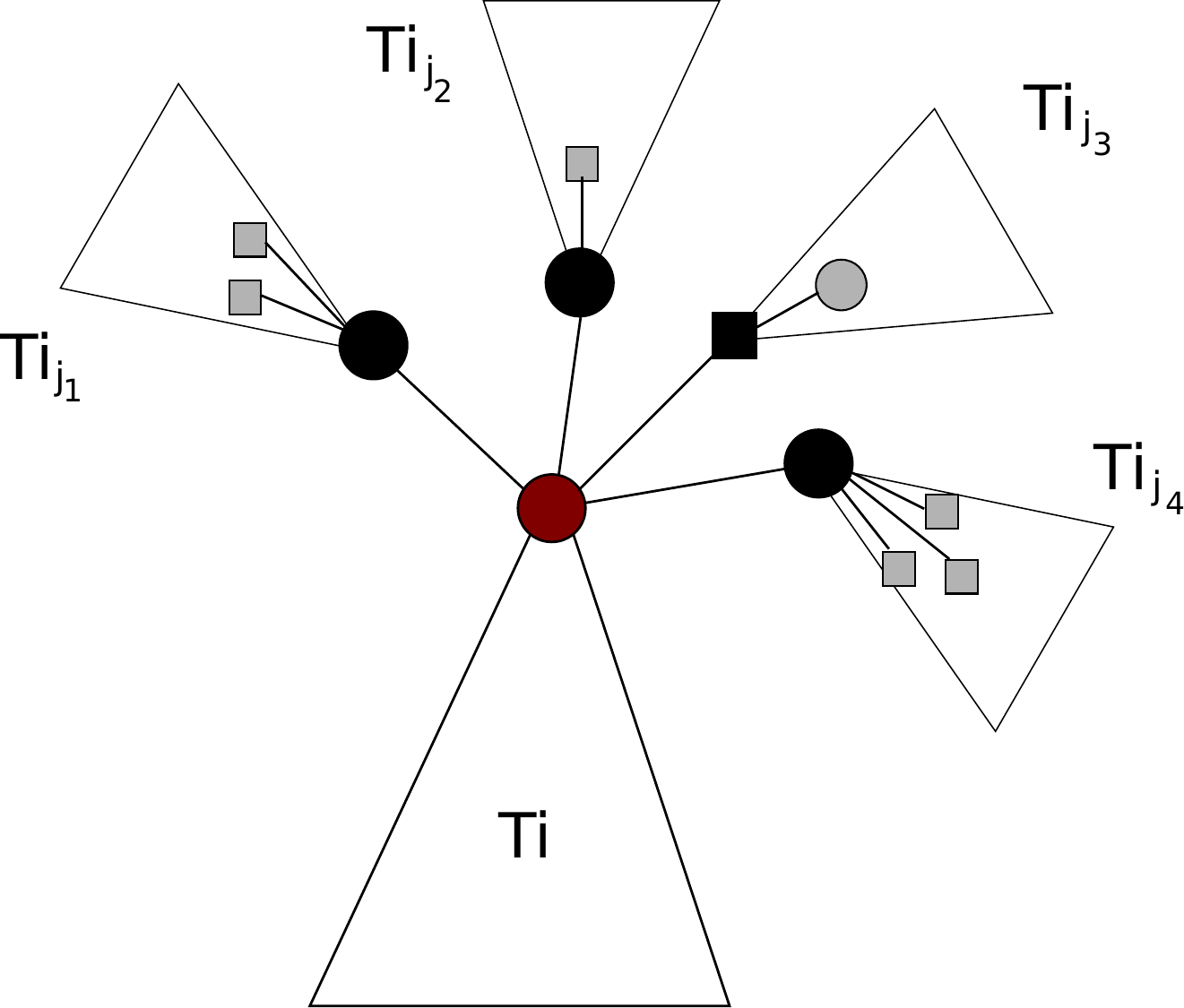}
\caption{An illustration of Phase 4.}
\label{fig:star-matching}
\end{figure}

\smallskip
Let $J = \{ j' \mid \sep_{i_{j'}} = \sep_{i_j} \}$.
One should ensure that in the solutions $\labeledsubtree{i_{j'}}{}, j' \in J$ that we will choose, the center nodes $v_{j'}$ in $\labeledsubtree{i_{j'}}{\maxk_{i_{j'}}}$ will be pairwise different.
Furthermore, since all the $v_{j'}$'s are pairwise at distance two, there can be at most one $j_{\min} \in J$ such that $d_{\labeledsubtree{i_{j_{\min}}}{}}(v_{j_{\min}}) = 1$.
See Fig.~\ref{fig:star-matching} for an illustration.
In order to satisfy all these constraints, while ensuring that such a $j_{\min}$ does not exist if it is possible, we make a reduction to {\sc Maximum-Weight Matching}~\cite{DPS18}. 

\begin{enumerate}
\item Specifically, let $Steiner[J] := \{ \alpha_{j'} \mid j' \in J \}$ be a set of Steiner nodes.
We construct a bipartite graph $Bip(\sep_{i_j})$ with respective sides $J$ and $(\sep_{i_j} \setminus {\cal C}_i) \cup Steiner[J]$.
\item For every $j' \in J$ and $v \in \sep_{i_j} \setminus {\cal C}_i$, there is an edge $j'v$ if there exists a $\labeledsubtree{i_{j'}}{} \in {\cal T}_{i_{j'}}$ such that $\centre{\labeledsubtree{i_{j'}}{\maxk_{i_{j'}}}} = \{v\}$. Furthermore if such a $\labeledsubtree{i_{j'}}{}$ exists then, we choose one maximizing $d_{\labeledsubtree{i_{j'}}{}}(v)$ and we assign the weight $d_{\labeledsubtree{i_{j'}}{}}(v)$ to the edge $j'v$ (this can either be $1$ or $2$).

\smallskip
In the same way, there is an edge $j' \alpha_{j'}$ if there exists a $\labeledsubtree{i_{j'}}{} \in {\cal T}_{i_{j'}}$ such that the unique node in $\centre{\labeledsubtree{i_{j'}}{\maxk_{i_{j'}}}}$ is Steiner. Furthermore if such a $\labeledsubtree{i_{j'}}{}$ exists then, we choose one maximizing $d_{\labeledsubtree{i_{j'}}{}}(\centre{\labeledsubtree{i_{j'}}{\maxk_{i_{j'}}}})$ and we assign the weight $d_{\labeledsubtree{i_{j'}}{}}(\centre{\labeledsubtree{i_{j'}}{\maxk_{i_{j'}}}})$ to the edge $j'\alpha_{j'}$.

\item We compute a matching in $Bip(\sep_{i_j})$ of maximum total weight.
This takes ${\cal O}(n^{5/2})$-time~\cite{DPS18}.
By construction, such a matching should contain an edge incident to every $j' \in J$, and its total weight should be either $2|J|-1$ (if $j_{\min}$ exists) or $2|J|$.
\end{enumerate}
For every $j' \in J$, we pick a solution $\labeledsubtree{i_{j'}}{}$ corresponding to the edge incident to $j'$ in the matching.
Then, as in all previous phases, we end up applying our filtering rule above. \\

\item {\it Subphase 4.b: Processing a star when $c \notin {\cal C}_i$.} 
This situation can happen only if $diam(\labeledsubtree{Y_i}{}) = 4$.
Then, the unique node $c_i \in {\cal C}_i$ is a neighbour of $c$ in $\labeledsubtree{i}{}$ (by Lemma~\ref{lem:center-in-star}).
We may further assume that, if $Real(N_{\labeledsubtree{i}{}}[c_i])$ is a minimal separator $\sep$ then, we already handled with $\sep$ during the previous subphase.
Similarly, we already handled with any minimal separator strictly contained into $\sep_{i_j}$, strictly containing $\sep_{i_j}$ resp., during the previous phases.
Hence, the unique path in $\labeledsubtree{i}{}$ between $\labeledsubtree{Y_i}{\sep_{i_j}}$ and any other $\labeledsubtree{Y_i}{\sep_{i_k}}$ that we did not process yet goes by $c_i$.
We are left with finding a solution $\labeledsubtree{i_j}{} \in {\cal T}_{i_j}$ maximizing $d_{\labeledsubtree{i_j}{}}(c_i)$.
However, as in the previous subphase we also need to account for the other indices $k$ such that $\sep_{i_k} = \sep_{i_j}$. 

\smallskip
Let $J = \{ j' \mid \sep_{i_{j'}} = \sep_{i_j} \}$.
We may assume $|J| \geq 2$ since otherwise, we are done by taking any solution $\labeledsubtree{i_j}{} \in {\cal T}_{i_j}$ that maximizes $d_{\labeledsubtree{i_j}{}}(c_i)$ ({\it i.e.}, as explained above).
If furthermore $|\sep_{i_j}| = 2$ then, we can select the partial solutions $\labeledsubtree{i_{j'}}{}, \ j' \in J$ by using a similar merging process as the one discussed in Phase 2 for the edges.
Thus, let us assume from now on $|\sep_{i_j}| \geq 3$.
Since $\sep_{i_j}$ must be weakly $\cliquetree{G}$-convergent (Lemma~\ref{lem:star-center}), and so, $\cliquetree{G}$-convergent (Property~\ref{pty-fct:1} of Theorem~\ref{thm:final-clique-tree}), it implies that, for any $j' \in J$, there can be no minimal separator of $G_{i_{j'}}$ that contains $\sep_{i_j}$.
However, an additional difficulty compared to the previous subphase is that now the center $c$ of the star can also be in $\centre{\labeledsubtree{i_{j'}}{\maxk_{i_{j'}}}}$.
So, we need to modify our approach in the previous subphase as follows:

\begin{enumerate}
\item We first choose the unique $j_0 \in J$ such that $c \in \centre{\labeledsubtree{i_{j_0}}{\maxk_{i_{j_0}}}}$ (if any).
Then, we choose a corresponding solution in ${\cal T}_{i_{j_0}}$ among ${\cal O}(|\sep_{i_{j_0}}|^5) = {\cal O}(|\maxk_i|^5)$ possibilities.
Overall, there are ${\cal O}(n|\maxk_{i}|^5)$ possibilities.
We test each such a possibility sequentially (including the case where no such a $j_0$ exists). 

\item Assume for this step that we fixed a value for $j_0$ (if we test for the case where no such a $j_0$ exists then, we can go directly to the next step). By Claim~\ref{claim:star-1}, the following property holds for any $v \in \sep_{i_j} \setminus \{c\}$: either $v$ is simplicial in $G_{i_{j_0}}$ (and so, $d_{\labeledsubtree{i_{j_0}}{}}(v) = 3$), or $d_{\labeledsubtree{i_{j_0}}{}}(v) = 1$.
In the former case we discard all solutions $\labeledsubtree{i_{j'}}{} \in {\cal T}_{i_{j'}}, \ j' \in J \setminus \{j_0\}$ such that $d_{\labeledsubtree{i_{j'}}{}}(v) < 2$ whereas in the latter case, we discard all solutions $\labeledsubtree{i_{j'}}{} \in {\cal T}_{i_{j'}}, \ j' \in J \setminus \{j_0\}$ such that $d_{\labeledsubtree{i_{j'}}{}}(v) < 4$.
%${\cal C}(T_{i_{j'}}\langle X_{i_{j'}}\rangle) = \{v\}$.

\item Finally, we apply our reduction to {\sc Maximum-Weight Matching} (from the previous subphase) in order to pick the solutions $\labeledsubtree{i_{j'}}{} \in {\cal T}_{i_{j'}}$ for every $j' \in J \setminus \{j_0\}$.
Indeed, we observe that for every $j' \in J \setminus \{j_0\}$, we will always obtain $d_{\labeledsubtree{i_{j'}}{}}(c_i) = 2 + d_{\labeledsubtree{i_{j'}}{}}(\centre{\labeledsubtree{i_{j'}}{\maxk_{i_{j'}}}}) \in \{3,4\}$.
Therefore, our choice for the partial solutions $\labeledsubtree{i_{j'}}{}$ will ensure that $d_{\labeledsubtree{i_{j'}}{}}(c_i)$ is maximized (w.r.t. our choice for $j_0$ and $\labeledsubtree{i_{j_0}}{}$).

\item Overall, among all the valid solutions computed (for any possible choice of $j_0$ and $\labeledsubtree{i_{j_0}}{}$), we keep the one maximizing $\min_{j' \in J} d_{\labeledsubtree{i_{j'}}{}}(c_i)$.
\end{enumerate}
\end{itemize}
\end{itemize}
This last phase concludes the algorithm.

\medskip
In order to complete the proof, \underline{let us finally assume $Y_i \neq \maxk_i$ (there are thin branches).}
Then, ${\cal C}_i = \{c_i\}$.
We consider all the minimal separators $\sep_{j_1}, \sep_{j_2}, \ldots, \sep_{j_q} \subseteq (\maxk_i \setminus Y_i) \cup \{c_i\}$ sequentially.
For every $\ell \in \{1,2,\ldots,q\}$ we must have $\unlabeledsubtree{\sep_{j_{\ell}}}$ is a thin branch, and so, a star.
We so have ${\cal O}(|\sep_{j_{\ell}}|) = {\cal O}(|\maxk_i|)$ possibilities.
Furthermore, since according to Definition~\ref{def:thin-leg} there can be no minimal separator $\sep_{i_k}$ which intersects both $\sep_{j_{\ell}}$ and $\maxk_i \setminus \sep_{j_{\ell}}$, any solution $\labeledsubtree{j_{\ell}}{} \in {\cal T}_{j_{\ell}}$ that maximizes $d_{\labeledsubtree{j_{\ell}}{}}(c_i)$ would be a best possible choice.
This latter case ressembles to the situation we met in Subphase 4.b.
We can solve it by using the same tools as for this subphase.
Specifically:

\begin{enumerate}
\item We consider each possibility for the star $\unlabeledsubtree{\sep_{j_{\ell}}}$ sequentially;

\item Given a fixed $\unlabeledsubtree{\sep_{j_{\ell}}}$, every minimal separator $\sep_{i_k} \subset \sep_{j_{\ell}}$ must be either a cut-vertex or induce an edge (otherwise, we can discard this possibility for $\unlabeledsubtree{\sep_{j_{\ell}}}$). Then, we can process such minimal separators $\sep_{i_k}$ as in Phases 1 and 2 above. %(but we do not apply the filtering rules).

\item We end up applying the same procedure as for Subphase 4.b. in order to select the partial solutions $\labeledsubtree{i_{j'}}{}$ such that $\sep_{i_{j'}} = \sep_{i_j}$. Namely, this procedure combines a brute-force enumeration with our reduction to {\sc Maximum-weight Matching}.

\item Overall, among all the valid solutions computed, we keep the one maximizing $d(c_i)$. Then, we can apply our filtering rule above.
\end{enumerate}
Once we applied this above procedure to all the thin branches, we can reuse our previous four-phase algorithm in order to process all the other minimal separators. 
\end{proofof}

\section{Conclusion}\label{sec:ccl}

There are essentially two dominant approaches in order to solve \KLP[k]~and \KSP[k]~in the literature.
The first one, and by far the most elegant, is based on structural characterization of the corresponding graph classes~\cite{BaBV06,BLS08}.
Unfortunately such characterizations -- mostly based on forbidden induced subgraphs -- look challenging to derive for larger values of $k$.
Furthermore, some recent work suggests that even a nice characterization of $k$-leaf powers ($k$-Steiner powers, resp.) by forbidden induced subgraphs might not be enough in order to obtain a polynomial-time recognition algorithm~\cite{Laf17}.

The second approach consists in a clever use of dynamic programming.
Although this approach is much less satisfying on the graph-theoretic side, it may be more promising than the first one.
For instance, the only known algorithms so far for recognizing $5$-leaf powers and $3$-Steiner powers are based on this approach~\cite{ChK07}.
Unfortunately, standard dynamic programming techniques are challenging to apply as the value of $k$ increases, which is probably why no improvement has been obtained for this problem for over a decade -- until this paper.

We propose several new avenues for research on dynamic programming algorithms for $k$-leaf powers and $k$-Steiner powers.
In particular, we hope that our structural analysis of these roots -- based on a renewed interest for clique-intersections -- can be helpful in order to generalize our algorithmic framework to larger values of $k$.
Some of our side contributions, especially the design of a problem-specific clique-tree and our greedy procedures in order to select partial solutions, can also be of independent interest for future research on this topic.

\section*{Acknowledgements}
%This work was supported by a grant of Romanian Ministry of  Research and Innovation CCCDI-UEFISCDI. project no. 17PCCDI/2018.
The author would like to thank the Research Institute of the University of Bucharest (ICUB) that founded the postdoctoral project of which this paper is the outcome.

\bibliographystyle{alpha}
\bibliography{Steiner-biblio}

\end{document}